%% file: main.tex
\documentclass[11pt, a4paper, logo, copyright]{gensyn} 

\input{packages}
\input{macros}

\newtheorem{theorem}{Theorem}
\newtheorem{lemma}{Lemma}
\newtheorem{proposition}{Proposition}
\newtheorem{corollary}{Corollary}
\theoremstyle{definition}
\newtheorem{definition}{Definition}
\newtheorem{assumption}{Assumption}

\usepackage{graphicx}
\usepackage{amssymb}
\usepackage{amsmath}
\usepackage{algorithm}
\usepackage{booktabs}
\usepackage{float}
\usepackage{caption}
\usepackage{subcaption}

\title{Who Aggregates Information? Screening, Rent, and the Coexistence of CLOB and AMM Prediction Markets}
\author[1,2]{Chengqi Zang}
\author[1]{Gabriel P.~Andrade}
\author[2]{Tomoyuki Nakajima}

\affil[1]{Gensyn AI}
\affil[2]{The University of Tokyo}

\begin{abstract}
Prediction-market shares differ from traditional financial products in that, with no information or outside utility, classical delta-neutral Central Limit Order Book~(CLOB) market making cannot be financed by payoff-uninformative noise flow.
Transaction-level evidence from a major prediction-market CLOB platform shows makers profiting not from spread but from carrying an under-priced side to settlement --- the empirical signature of behavioral tail demand rather than classical, randomized noise. 
We build this tail demand directly into the model and study an LMSR and a CLOB on the same event.
A pre-shock CLOB quote inside the common-signal band is picked off; competitive quotes therefore screen informed traders out of the book.
CLOB makers earn screening rent by carrying the under-priced side to resolution, while informed flow routes to the LMSR.
The venues coexist: the CLOB supplies the tail-demand rent margin that lets the LMSR recover part of its loss to informed flow, and AMM depth moves the CLOB premium with a sign set by maker-side contestability---widening it where standing quotes can be undercut, compressing it where a committed maker carries the book.
With three or more outcomes, binary-book CLOBs pin switch prices but leave implied beliefs indeterminate, whereas the LMSR prices the outcome simplex coherently and uses collateral more efficiently.

\end{abstract}

\begin{document}
\maketitle           
\section{Introduction}
\label{sec:intro}

Prediction markets aggregate dispersed information into market-implied
probabilities by rewarding correct prediction. The concept traces
to~\cite{hanson2003combinatorial,hanson2007logarithmic}, whose
logarithmic market scoring rule~(LMSR) is a cost-function automated
market maker~(CF-AMM): traders reveal beliefs through their orders, and the maker absorbs a bounded loss at settlement as the price of information
extraction.

Implementation has diverged from theory. 
Seven-day volume across prediction markets reached \$2.8 billion in late 2024, concentrated on venues that do \emph{not} use a CF-AMM: Polymarket, Kalshi, and the sports moneyline books all run on central limit order books~(CLOBs). 
The gap is structural: a CF-AMM maker cannot choose its position-level PnL, whereas a CLOB maker posts only at prices that clear a profit---which pulls real venues toward CLOBs even though the information-aggregation case for CF-AMMs is better developed.

CLOB dominance does not resolve a deeper tension. 
Classical order-book making~\cite{glosten1985bid} needs noise traders: without payoff-uninformative flow to absorb adverse selection, competitive dealers cannot quote informative spreads and stay solvent. A prediction-market share carries little or no such noise, and, unlike a stock or a token, no consumption value, cash flow, utility, or reference price---nothing anchors it before resolution. 
Three questions follow.

First, \emph{what plays the role of noise?} With nothing to anchor the price, informed traders always move first, and only then does uninformed demand arrive; what is that demand, and how does it interact with the bounded $[0,1]$ range and the $0/1$ settlement payoff? Second, \emph{what is the maker's information role?} Can a prediction-market CLOB aggregate information at all---can a committed quote vector even display a probability---and what does a venue give up by relying on one? Third, \emph{can a CF-AMM and a CLOB prediction market coexist productively?} Hybrid designs exist~\cite{othman2010automated} but have not scaled; when does coexistence dominate either pure mechanism, and how are flow, rent, and information divided across the two venues?

\paragraph{Contributions.}
We answer the three questions with one screening mechanism that carries
three implications. 
First, in place of the absent noise trader we model the empirically relevant demand as state-contingent \emph{tail demand}: informed flow moves the price to the informative posterior first, and uninformed traders then overpay or underpay with respective to the informative posterior. 
When CLOB quotes are placed before the signal, any quote inside the signal band is picked off; a viable CLOB therefore posts outside the band, screens informed flow out, and earns a screening rent by carrying the complementary claim to resolution---the tail demand that finances the maker.

Second, when CLOB and LMSR coexist, the usual venue hierarchy is inverted: informed traders route to the LMSR, while the CLOB survives on behavioral rent. 
Coexistence is not redundant. 
CLOB screening disciplines the tail-demand margin the LMSR draws on to recover its losses to informed flow, while the book's committed premium caps how far the LMSR pool's price can walk. 
Which venue keeps that margin depends on how traders route: under firm-quote priority the premium is pinned by free entry alone and is neutral to the pool depth of LMSR, while once traders compete for the pool's cheap early prices, depth strips revenue from the book and spills over to its spread, with the sign of the spillover set by maker-side contestability---contestable books widen their premium as depth grows, committed books compress it---and past a threshold the pool removes the book's clientele altogether.

Third, the separation sharpens with $n\ge3$ outcomes, where the book's problem stops being about rent and becomes about the product.
Traders on multi-outcome events transact switches, long the more preferred and short the less preferred, to express preference, and the price of such a switch trade is a difference of two quotes, so trades price only differences and never the level~(displayed belief)---which is precisely the probability the book purports to show.
That level is therefore indeterminate along a gauge~\footnote{A one-dimensional family of quote-level shifts that preserves all switch prices.}, wrong by at least the shock size wherever a normalization is imposed, and left free even with maker competition, which pins the book's premium and nothing else, and backed by no bet; since a maker can relocate it with no order executed and no payoff changed.
The LMSR instead prices the whole simplex from one potential, reprices its display with every trade it absorbs, and collateralizes the event as a whole at a factor $\Theta(n/\log n)$ saving.
Set beside the classical characterization of cost-function makers, these results pincer: up to the axioms, the cost-function AMM is not merely the better way to sell a probability over $n\ge3$ outcomes but the only committed mechanism that sells one at all.

\S\ref{sec:market_info} sets up the model, \S\ref{sec:binary_model} solves the binary coexistence equilibrium, and \S\ref{sec:3_outcome} derives the multi-outcome information and collateral implications; all proofs are in the appendix.

\subsection{Stylized Facts from a Prediction-Market CLOB}
\label{sec:stylized_facts}

Tail demand is well documented (\S\ref{sec:related}); what it does to the maker who absorbs it is not.
We use two samples.
\citet{akey2026wins} cover the whole Polymarket record---$588$ million trades and \$67 billion of volume, across every market category the platform lists---with profit and loss measured account by account.
We add transaction-level detail for one contract, Polymarket's Bitcoin five-minute up/down claim: $738{,}393$ trades across $275$ markets in about twenty-four hours, \$10.38M of notional, and $15{,}134$ accounts sorted into makers and retail by how they quote (Table~\ref{tab:app_classification}).
Fees rule out the first objection anyone would raise: takers pay $1.8\%$, makers pay nothing and get no rebate, and we measure maker profit after fees, so a maker who makes money here is not being paid to.
Two limits bound the transaction-level samples: neither contains an AMM, so neither tests the venue choice of~\S\ref{sec:binary_model}, and the maker-rent facts below are measured on binary contracts.
The multi-outcome prediction of~\S\ref{sec:3_outcome} we take up separately, by reconstructing multi-outcome books from the~\citet{akey2026wins} record (\S\ref{app:empirical_spread}).

Three facts describe one agent.

\emph{First, makers earn at resolution, not on the spread.}
In our sample $91.4\%$ of the persistent makers' profit arrives at settlement and under $9\%$ from buying and selling before it, with makers ending ahead ($+\$45.3$K) and retail behind ($-\$66.7$K).
The same split holds across all Polymarket markets, and holds by order type: accounts that make money post limit orders whose positions turn out to be on the winning side, while accounts that lose money hit those orders with market orders~\citep{akey2026wins}.
Where a maker earns also tells us where it is quoting.
A quote left sitting inside the range the news will move into gets hit and loses at settlement, so a maker that keeps earning there must be quoting outside that range.
This shows up in the order book directly: as a five-minute market nears a close that is still too close to call, Polymarket makers cut their depth sharply~\citep{dai2026settlement} to avoid a shock that can be large enough to flip the market resolution.
In our model, this CLOB maker behavior is well-explained by the no-pickoff floor of Lemma~\ref{lem:no_pickoff}.

\emph{Second, the edge is one-sided: the maker buys the under-priced side and holds it.}
Almost all the profit comes from buying (${\approx}+\$70$K) rather than selling (${\approx}+\$2.6$K), and it sits above the midpoint: $+\$52.4$K between $0.50$ and $0.80$, $+\$24.8$K between $0.80$ and $1.00$, and $-\$0.5$K from selling above $0.80$ (Table~\ref{tab:app_pnl_decomp}).
What the maker holds resolves in its favor $56.8\%$ of the time.
Could it simply be informed?
Looking at how the biggest winners trade across all Polymarket markets, \citet{akey2026wins} conclude that insider trading is unlikely to explain their returns.
An edge that is directional and repeatable but not based on information is a screening position, which is why the model gives the two sides of a quote their own tail functions $H_F$ and $H_U$ instead of one spread, and why the maker's payoff is written as holding the other side of the claim until it settles.

\emph{Third, the mispricing switches sign with the volatility regime.}
Only our sample can show this, because it needs each market's price path.
Pooled together, the mid-band gaps are small: $+0.043$ on the longshot side and $-0.037$ on the favorite side.
Split by price path they reverse: the longshot band runs $+0.225$ in trending markets and $-0.071$ in flippy ones, whose prices keep changing direction, and the favorite band $-0.205$ and $+0.066$ (Table~\ref{tab:regime_calibration}).\footnote{The two rows are near-complements, so their being equal and opposite is mostly an accounting identity: $p$ and $1-p$ cannot be mispriced the same way. What matters is the sign flip \emph{within} a row, not the symmetry across rows.}
In plainer terms, a side priced near \$0.29 wins about $7\%$ of the time when the market is trending and about $37\%$ when it is flippy.\footnote{See~\S\ref{app:empirical} for how trending and flippy markets are told apart.}
A bias that flips sign cannot be a fixed misreading of probabilities, so the model lets the tail-demand population depend on the state rather than on the claim.

So the agent is a maker that quotes both sides, earns almost everything at settlement, and holds the under-priced side that depends on the regime.
That is why the model makes profit at resolution the maker's objective~(\S\ref{sec:market_info}), screens informed traders off the book before any rent is earned, gives the quote two different tails, and lets tail demand depend on the state.
The first two properties hold across the whole platform; the third we document in one contract type---short-horizon binary markets that resolve quickly---and it is those markets the calibration speaks to.
None of it depends on the pipeline: the facts hold whether accounts are keyed on wallet address or display label, and the sign reversal survives cluster-bootstrap resampling at the market level.
\S\ref{sec:binary_model} shows why a competitive prediction-market CLOB produces this agent.

A book that prices outcomes one leg at a time cannot sell the complete set at par.
\S\ref{sec:3_outcome} makes this concrete: each leg carries its own committed markup~\eqref{eq:switch_prices_and_premium}, so assembling every outcome costs strictly more than one dollar, and the excess widens as legs are added---whereas a single cost-function potential prices the complete basket to par (Proposition~\ref{prop:self_collateral}(i)).
A third slice of the~\citet{akey2026wins} record lets us look for that excess directly.
We reconstruct each multi-outcome event from its component binary Yes/No books ($16{,}036$ validated events, outcome counts $2$--$10$), price every leg from executed trades on \emph{both} sides of each book---the last taker-buy as the ask, the last maker-buy as the bid---and measure the \emph{set premium}, the cost of buying YES on every outcome minus one dollar.
Because a price stands until the next trade replaces it, each book can be priced continuously over its whole life, and our estimates weight every price by how long it stood: the measured premium is what a buyer arriving at a random moment of the book's life would pay.
Pricing both sides is what makes the test sharp, because the premium then splits, as an accounting identity, into exactly the two objects the model separates: a book-center term $\sum_j \mathrm{mid}_j - 1$, nonzero only if displayed beliefs fail to sum to one, and a committed-markup term $\sum_j (\mathrm{ask}_j - \mathrm{mid}_j)$, the empirical counterpart of the per-leg premium of~\eqref{eq:switch_prices_and_premium}.

Both halves of the \S\ref{sec:3_outcome} prediction show up in the data, each in its own term.
The cost of completing the set rises in the outcome count---the coefficient on $\log n$ is $+0.060$ ($p{<}0.001$), estimated with one vote per event---and $96\%$ of that rise sits in the committed-markup term, not in the mid prices.
The markup accumulation behaves exactly as a per-book committed premium should.
It replicates out of sample on the NO side of the same events ($+0.048$, of which ${\sim}105\%$ is markup); roughly half of it is reproduced by placebo ``events'' assembled from $k$ \emph{unrelated} binary markets matched on calendar week, horizon, and liquidity---bundles that carry a per-leg markup but, by construction, no multi-outcome economics; and it lives in small executions: requiring every leg's most recent trade to be at least $100$ shares drives the slope to statistical zero, the footprint of thin far-from-center quotes~(the depth withdrawal of~\citealt{dai2026settlement}).
The book-center term tells the level story of~\S\ref{subsec:3o_beliefs}: averaged over each book's whole life, the centers sum to one dollar at every $n$ (slope $+0.002$, $p{=}0.52$)---no committed bet pins the level, and none is systematically taken---yet compared at the same point in a book's life, or at the same time to resolution, the sum drifts above par by $+0.02$ to $+0.08$ per $\log n$: widest at the $7$--$30$-day horizon, where carrying the correcting set trade ties up the most capital for the longest, and narrowest, though still present, inside the final day.
The lifetime average is flat precisely because small books are only fully priced near the end of their lives, when every book is tight, while large books are priced throughout; the flat average mixes ages, and holding age fixed restores the gradient.

As for the robustness of the estimation, each estimate survived counterfactual checks against the confounds this construction brings.
Price staleness is not the driver with three checks: the slope is unchanged when regressing small and large books with equal price ages; secondly, assuming that mid-price/public belief may drift and makers might reprice, we add price drift correction estimated from $8.4$M consecutive-trade pairs and it moves the slope by at most $0.006$; thirdly, deliberately reprice cleanly measured fast trading books with the more stale transactions like large books produces a slope of the \emph{wrong} sign and under a tenth the size.
Together, those three checks show that staleness is not the driver for increasing slope of larger books.
Selection, outliers, and functional form are bounded: reweighting by how often a book is measurable moves no estimate by more than $0.009$, the largest change from dropping any single event category is $0.0008$, and $\log n$ fits better than linear $n$.
Full construction, every specification, and the tables are in~\S\ref{app:empirical_spread}.

\section{Background and Related Work}
\label{sec:related}

Four strands bear on the question, and none of them places a committed book beside a cost-function maker on a claim that resolves.
From each we take one tool or one precedent, name the antecedent that comes closest to our result, and locate the departure in a single object rather than in a list of differences.

\paragraph{Cost-Function Market Makers.}
The cost-function maker is the one mechanism whose loss and whose displayed probability are both objects of design.
Three things come from its convex-analytic theory.
First, the LMSR cost function and its bounded-loss guarantee~\citep{hanson2003combinatorial,hanson2007logarithmic}, which fixes the AMM benchmark against which the segregated order book's collateral bill is measured in~\S\ref{sec:3_outcome}.
Second, the axiomatic characterization~\citep{chen2007betting,abernethy2013efficient,lambert2015axiomatic}: a maker that quotes a full probability vector, prices path-independently, and admits no arbitrage \emph{is} a convex-potential cost-function maker.
That result is load-bearing since it tells us what the admissible alternatives to an LMSR are before we compare the LMSR to one, and it supplies the pincer in~\S\ref{subsec:3o_balance}.
Third, the equivalence between a constant-function AMM on the probability simplex and a cost-function prediction market~\citep{frongillo2023axiomatic}, whose reductions carry trade rules and market properties in both directions.
The equivalence is what lets us treat ``AMM'' and ``scoring-rule market'' as one object and draw comparison like~\citep{aoyagi2025coexisting}, and its slogan---a market is good at facilitating trade exactly when it is good at revealing beliefs---is what makes our venue comparison a comparison of information mechanisms and not merely of exchange technologies.

Liquidity-sensitive and profit-charging variants~\citep{othman2013practical,othman2010automated,bhaskara2023general} make the maker's loss budget parameter(s) of the pricing rule rather than the outcome of a participation decision; we instead let the AMM's liquidity be set by free entry of providers who weigh the adverse-selection payment against the rent available on the competing order book~(\S\ref{subsec:lmsr_lp_problem}), so depth is an equilibrium object that the book can move.

A second thread within this literature asks what the resulting price means, and it is the thread~\S\ref{subsec:3o_beliefs} joins.
Whether a market price is the traders' mean belief at all is contested~\citep{manski2006interpreting,wolfers2006interpreting}; whether dynamic trade aggregates dispersed information turns on the security structure, since separable securities drive prices to the full-information value while non-separable ones need not~\citep{ostrovsky2012information}; and strategic traders facing a scoring-rule maker may withhold or bluff rather than reveal on arrival~\citep{dimitrov2008non,chen2010gaming}.
Combinatorial market makers attack the same coherence problem from the computational side, they aim to price logically related claims from one potential precisely so that separately quoted books cannot disagree~\citep{fortnow2005betting,chen2008complexity,dudik2013combinatorial}.
Our information defect is a further failure mode, and the least escapable of them: when makers price a $n(n\ge3)$-outcome market from a network of binary books based on their belief of trader's preference between any pair of outcomes, they are pricing the \emph{difference} between quote prices and thus, the displayed price~(probability) is unidentified even with truthful traders, immediate revelation, and market maker computation.

\paragraph{Automated Market Makers in DeFi.}
The sharpest available account of what a passive maker loses to informed flow is measured against a reference price we do not have.
Loss-versus-rebalancing decomposes the provider's shortfall into a term that grows with the variance of an exogenous reference price~\citep{milionis2022automated,milionis2024automated}, which is the counterpart of the adverse-selection payment that~\eqref{eq:lp_foc} prices.
Around it, one literature treats that exposure as a predictable cost of provision~\citep{cartea2023predictable} and asks what fee schedule compensates it~\citep{evans2021optimal,hasbrouck2026need}, while another compares venue quality and pricing rules~\citep{barbon2023quality,park2023conceptual,lehar2025decentralized,capponi2021adoption,malinova2024learning} and studies the provider's entry problem directly~\citep{aoyagi2020liquidity,angeris2020improved,adams2021uniswap}.
One channel transfers intact: \citet{hasbrouck2026need} show that deeper inventory lowers price impact and can therefore raise volume, and the same depth-to-volume force appears in our crowd-out result, except that it operates across venues, pulling the pool's clientele off the book rather than a single venue's volume up~(\S\ref{subsec:capacity_competition}).

Two features of the DeFi setting drive our departures, and both are absent by construction here.
First, tokens carry intrinsic utility~\citep{benedetti2023utility} and trade against a continuous reference price~\citep{milionis2022automated,milionis2024automated}, which regenerates the payoff-uninformative flow a prediction-market share sometimes cannot sustain; the fee stream that finances LVR in these models has no counterpart on a outcome share nobody holds for any reason but its resolution, so the maker/provider must be financed by something else, and identifying that something is the first question of~\S\ref{sec:binary_model}.
Second, there is no terminal settlement, so provider risk is a continuous exposure rather than the bounded, all-or-nothing risk a prediction-market share creates: a token inventory can be unwound at some price, whereas a share carried into resolution pays $0$ or $1$ whatever its holder then believes.
Both differences trace to one root---a claim with no direct reference value~(price) and no life after resolution---and it is that root, not the AMM technology, that our model changes.

\paragraph{Venue Choice.}
Two venues coexist only when each holds a clientele the other cannot profitably serve.
Classical fragmentation theory supplies both halves of that condition: \citep{pagano1989trading} shows higher liquidity kills differentiation when entry cost is homogeneous; and differentiation survives when venues react differently to the tradeoffs between price discovery and execution certainty~\citep{parlour2003liquidity,hendershott2000crossing,foucault2008competition}.

Two precedents place a maker beside a book inside a prediction market.
\citet{othman2010automated} deploys the hybrid and reports what it costs to run; \citet{heidari2018integrating} comes closest on mechanism, giving the first algorithm that combines a cost-function maker with limit orders under continuous trade through the notion of an $\epsilon$-fair trading path, and settling when such a path can terminate efficiently.
That is a computational guarantee about how orders execute, and it carries no trader types, so it does not deliver the equilibrium \emph{division} of flow between the venues, the conditions under which they coexist rather than one crowding the other out, or the behavior of a strategic maker unbound by any scoring rule---the three gaps~\S\ref{sec:binary_model} fills.

Our direct microstructure precedent is \citet{aoyagi2025coexisting}, which models endogenous venue choice between a CLOB and a CFMM of crypto asset trading in the Glosten--Milgrom--Kyle tradition~\citep{glosten1985bid,kyle1985continuous} and finds that AMM depth spills over positively onto book liquidity, because informed and noise traders react to depth at different rates.
We adopt their venue-choice setup and change three things: three periods, because claim will resolve in prediction market; a network of binary books, so an event can carry more than two outcomes compare to uniswap-style pools where only two assets can be in one pool; and tail demand traders with subjective valuations in place of noise.
The third change flips the venue's role of information aggregation, where in our model, AMM becomes the sole venue for price discovery and CLOB serves as the venue for makers to extract tail-demand.
Tail-demand also create spillover effect that can be either positive or negative depending on the maker's quoting convention, \citep{aoyagi2025coexisting} positive spillover is akin to the contestable-book case of ours.

\paragraph{Behavioral Flow.}
Tail demand is a documented population and three bodies of evidence establish it at different scales.
At the market level, the favorite--longshot bias is among the oldest documented pricing anomalies~\citep{griffith1949odds,ali1977probability,thaler1988anomalies}, where traders over or under estimation reads as probability weighting or misperception~\citep{snowberg2010explaining,ottaviani2008favorite,ottaviani2010noise,andrikogiannopoulou2021behavioral}, and surfaces in prediction markets as overpriced longshots and underpriced favorites~\citep{page2013prediction,lee2020anomalies}.
At the level of the underlying preference, the same lottery demand and extrapolative signature appears in retail equities~\citep{barberis2008stocks,kumar2009gambles,bali2011maxing,greenwood2014expectations,barberis2018extrapolation}, with the probability-weighting function in prospect theory as the common primitive~\citep{tversky1992advances,prelec1998probability}, i.e.~people attribute excessive weight to events with low probability and insufficient weight to events with high probability; this is why we treat tail demand as structural to retail-facing markets.
At transaction scale, recent work documents directional, belief-driven retail flow and a maker--taker profit asymmetry on Kalshi and Polymarket~\citep{burgi2025makers,le2026decomposing,tsang2026anatomy}, with losing traders concentrating at the tails~\citep{akey2026wins}. 
The sharpest evidence on how that asymmetry is monetized comes from \citet{bartlett2026adverse}. 
Across 41.6 million Kalshi trades, they show that makers earn roughly twice as much per contract in single-name markets~\footnote{Markets referencing a specific company or individual.} despite facing greater informed price impact there. 
Crucially, the maker's profit is not delta-neutral spread capture: makers accumulate the unpopular side and carry it to resolution, so the behavioral overbetting itself cross-subsidizes their adverse-selection losses---the empirical shadow of the rent structure we derive in~\S\ref{sec:binary_model}.

\section{Market and Information Structure}
\label{sec:market_info}

An LMSR and a CLOB operate simultaneously on the same binary event, which settles to YES or NO at $t=2$. 
Each outcome is an Arrow--Debreu claim: $X_Y$ pays $\$1$ if the event resolves YES and $\$0$ otherwise, $X_N$ the reverse, so $X_Y+X_N=1$. 
The LMSR charges a linear fee $\tau\in(0,1)$ on taker volume,
rebated to its liquidity providers; the CLOB charges no fee, but its makers incur a fixed posting cost $\kappa>0$, which can be understand as a capital lockup cost. 
LMSR prices are set mechanically by its cost function and net inventory, while CLOB makers choose their own quotes; we develop both pricing maps in~\S\ref{sec:binary_model}. 

\subsection{Timing}
The model has three periods.

\smallskip
\noindent\emph{$t=0$: liquidity provision and quoting.}
Liquidity providers commit capital and cannot revise it afterward.
Within the period, the LMSR LP moves first, funding the collateral that backs settlement and fixing the liquidity parameter $L>0$; CLOB makers then post competitive bid and ask quotes, observing $L$, and cannot cancel afterwards.

\smallskip
\noindent\emph{$t=1$: information arrives and trade occurs, in two stages.} In Stage~1, a payoff-relevant common shock is realized, the shock $\sigma_c$ is given by $\sigma_c\in\{+\sigma,-\sigma\}$, each sign equally likely; a unit mass of \emph{informed} $c$-traders observe it, trade first, and move prices to the post-shock posterior, lifting any stale quote. 
In Stage~2, \emph{tail-demand} $p$-traders arrive. 
They neither observe the shock nor know how to price it, so they take the prevailing post-shock price as their reference and trade only on a behavioral willingness-to-pay wedge that does \emph{not} move fundamentals.
Each trader routes to a single venue.

\smallskip
\noindent\emph{$t=2$: settlement.}
Each claim pays its realized $0/1$ value and maker PnL is realized. Unlike a DeFi AMM, settlement is terminal and need not coincide with the $t=1$ price: a claim trading at $\mu$ at $t=1$ still pays $0$ or $1$.

\subsection{Participants and Information}
\label{subsec:market_info}

\paragraph{Common-information traders.}
A unit mass of risk-neutral $c$-traders observes a common shock and each trades a single unit. 
Conditional on the shock, the public posterior values of the two
claims are
$$
\mu_Y(\sigma_c)=p_0+\sigma_c,\qquad \mu_N(\sigma_c)=1-\mu_Y(\sigma_c),
$$
where $p_0 = \tfrac12$ is the $t=0$ YES price assumed throughout the paper. 
Writing the favored and unfavored values as $\mu_F=\max\{\mu_Y,\mu_N\}$ and $\mu_U=1-\mu_F$, each $c$-trader buys the favored claim $F$. 

\paragraph{Tail-demand traders.}
After the common shock, exactly one continuum of $p$-traders of mass $z\in(0,1)$ is active. Conditional on the shock, the active direction is the favored direction $D=F$ with probability $\tfrac12$ and the unfavored direction $D=U$ with probability $\tfrac12$. 
The direction draw occurs after Stage~1 and before tail-demand traders choose a venue. 
Thus $F$ and $U$ index mutually exclusive continuation states condition on realization. 
An active $p$-trader $i$ draws a valuation $m_i\in[0,1]$ of the YES share from the corresponding shock-conditional tail of
Assumption~\ref{ass:tail_demand_distribution} and trades whenever $m_i$ makes the active side of the maker's quote executable. 
The rent this generates is behavioral overpayment concentrated around the post-shock posterior, not information---her only private input is $m_i$.

\paragraph{Liquidity Providers and Market Makers.}
A continuum of competitive LP/makers supply liquidity on both venues. 
At $t=0$ their beliefs are calibrated to the shock law and, via Assumption~\ref{ass:tail_demand_distribution}, the shock-conditional valuation distribution facing each committed direction, and price to the resulting settlement belief. 
Without loss of generality, a CLOB maker posts a bid-ask quote pair on claim $X_Y$ , at cost $\kappa$; an LMSR LP funds the collateral that fixes $L$.

In the following assumption, we will introduce the shock-conditional tail-demand flow distribution $m_i$.

\begin{assumption}[Tail-demand Distribution]
\label{ass:tail_demand_distribution}
The LP's and makers' common belief over $p$-trader valuations satisfies:
(i) \emph{(shock-conditional unimodality)} conditional on $\sigma_c$, valuations $m_i$ are i.i.d.\ with a continuously differentiable, log-concave density $f(\cdot\,|\,\sigma_c)$ on $[0,1]$, peaked at the realized posterior $\tfrac12+\sigma_c$ and strictly decreasing away from the peak toward both endpoints;
(ii) \emph{(mirror symmetry)} the two shock-conditional densities are reflections of each other across $\tfrac12$;
(iii) \emph{(tail ordering)} writing
$$
H_F(y)\;\equiv\;\Pr\!\big(m_i\ \ge\ \tfrac12+\sigma_c+y\ \big|\ \sigma_c\big),
\qquad
H_U(y)\;\equiv\;\Pr\!\big(m_i\ \le\ \tfrac12+\sigma_c-y\ \big|\ \sigma_c\big),
\qquad y\ge0,
$$
for the near- and far-tail masses at displacement $y$ from the peak---each shock-invariant as a function of $y$ by clause (ii), with support ceilings $\bar y_F\equiv\sup\{y:H_F(y)>0\}$ and $\bar y_U\equiv\sup\{y:H_U(y)>0\}$---there exists $\bar s>0$ such that
$$
H_U(2\sigma+s)\;\le\;H_F(s)
\qquad\text{for all } s\in[0,\bar s];
$$
(iv) \emph{(one active direction)} conditional on the common shock, exactly one direction $D\in\{F,U\}$ is active, with $D=F$ and $D=U$ each occurring with probability $\tfrac12$.
The active direction carries trader mass $z$: at markup $y$ over its realized posterior, the mass willing to buy is $zH_F(y)$ in the favored-direction state and $zH_U(y)$ in the unfavored-direction state;
and (v) \emph{(single-peaked pooled revenue)} the pooled screening revenue
$$
s\;\longmapsto\;s\,H_F(s)+(2\sigma+s)\,H_U(2\sigma+s)
$$
is strictly quasi-concave on the full premium range $[0,\bar y_F)$, where the far-tail term is read as zero whenever $2\sigma+s\ge\bar y_U$.
\end{assumption}

The near tail runs into the close outcome wall ($\bar y_F\le\tfrac12-\sigma$) while the far tail keeps its full room ($\bar y_U\le\tfrac12+\sigma$); beyond clauses (iii) and (v), no relation between the two tails' shapes is imposed.
Intuitively, (iii) says that at any committed quote in the coverage band, the extreme demand above the ask outweighs the extreme demand below the bid---the far side pays a markup inflated by the full shock $2\sigma$, and fewer traders clear it.
It is also directly estimable, being a statement about executed volumes on the two sides of a symmetric book.
Clause (iv) fixes what the tail functions count: demand for a claim at markup $y$ over that claim's realized peak has mass $zH_F(y)$ or $zH_U(y)$ according to the side, whichever the shock.
Clause (v) is a joint restriction on the two tails and the shock size that clauses (i)--(iv) do not imply: log-concavity makes each tail's revenue single-peaked on its own, but a sharply concentrated near tail over a long thin far tail can make the pooled revenue twin-peaked.
It holds for Gaussian and skew-normal tails across standard parameterizations.
We caution that neither (iii) nor (v) is a revenue ordering: near the zero premium the far side's revenue strictly dominates the near side's, with certainty---near revenue vanishes with $s$ while far revenue tends to $2\sigma H_U(2\sigma)>0$---and that floor rent is exactly what the posting cost must clear (the lower bound of Assumption~\ref{ass:low_fee}(ii)) for price competition to pin the premium (\S\ref{subsec:platform_choice_and_eq}).

The displacement mechanics of a committed quote---rent $s$ against mass $zH_F(s)$ on the favored side, rent $2\sigma+s$ against mass $zH_U(2\sigma+s)$ on the unfavored---are developed with the CLOB problem in~\S\ref{sec:clob_eq}; here we record the one distributional consequence the equilibrium selection uses.

The density peaks at the posterior---most tail-demand traders overpay only slightly, and the mass thins toward either extreme---while each \emph{revenue} peak sits strictly inside its band, where the thinning mass and the rising rent balance.
The pooled peak $s^\ddagger$ splits the screening band into a thick, low-premium segment ($s<s^\ddagger$) and a thin, high-premium segment ($s>s^\ddagger$): a contestable book is competed down into the thick segment, while a committed book prices at the peak itself, and this maker-side distinction sets the sign of the liquidity spillover, pinned down explicitly in Proposition~\ref{prop:binary_spillover}.

\section{Binary Market Model}
\label{sec:binary_model}

Two venues now compete for the same two flows, and everything downstream turns on a single sorting margin: which $p$-traders transact at the CLOB's firm price and which ride the LMSR's inventory-driven path.
Screening, routing, and free entry are determined by the primitives above, and we resolve the sorting margin under two routing conventions.
Under \emph{firm-quote priority} (Convention~A, \S\ref{subsec:platform_choice_and_eq}--\S\ref{subsec:lmsr_lp_problem}), traders who can transact at the CLOB's committed price do so, and the equilibrium is recursive: the CLOB premium is pinned by free entry alone, and the LP's depth problem separates.
Under \emph{price competition} (Convention~B, \S\ref{subsec:capacity_competition}), traders compete for the LMSR's cheap early prices, the two venues' volumes become jointly determined, and the depth-to-premium spillover of Proposition~\ref{prop:binary_spillover} emerges.
Let $\alpha_D\in[0,1]$ denote the equilibrium fraction of served direction-$D$ tail demand absorbed by the LMSR, and $\beta_F\in[0,1]$ the fraction of $c$-traders it absorbs; we derive both objectes under Convention~A.

\subsection{Central Limit Order Book}
\label{sec:clob_eq}

We write the CLOB problem on the active demand leg $D$, the claim $p$-traders over(under)value once the common information is realized.
The complementary carry claim is $C=\bar D$: by binary complementarity, a maker implements the same exposure either by selling $D$ at an ask or by buying $C$ at a bid.
Write the equivalent active-leg ask and implied carry bid as
$$
A^D=\mu_D+s_D,\qquad B^C=\mu_C-s_D,
$$
where $s_D\ge 0$ is the CLOB premium on leg $D$.

The premium cannot be arbitrary. 
Makers commit a single two-sided quote at $t=0$, before $\sigma_c$ is realized, while the unit $c$-mass observes the shock and lifts any stale side at $t=1$ (Stage~1).
The following lemma gives the rule for where rational makers quote pre shock.

\begin{lemma}[No-pickoff Screening Floor]
\label{lem:no_pickoff}
A pre-shock quote is pickoff-proof if no realization can cause the maker to gain pure loss.
A two-sided YES quote is pickoff-proof if and only if its YES ask is at least $\tfrac12+\sigma$ and its YES bid is at most $\tfrac12-\sigma$.
Equivalently, when the quote is written on the realized active leg as $A^D=\mu_D+s_D$ and $B^C=\mu_C-s_D$, its executable premium must satisfy
$$
s_D\ge \underline s_D\equiv\mu_F-\mu_D.
$$
Hence $\underline s_F=0$ in the favored-direction state and $\underline s_U=2\sigma$ in the unfavored-direction state.
The pickoff-proof premia that also serve a strictly positive mass of $p$-traders form the screening band
$$
s_D\in[\underline s_D,\bar s_D),
\qquad
\bar s_F\equiv\bar y_F,
\qquad
\bar s_U\equiv\bar y_U.
$$
If $\underline s_D\le\bar s_D$, no pickoff-proof quote can serve positive tail demand in direction $D$.
\end{lemma}

Lemma~\ref{lem:no_pickoff} lets us carry one premium parameter for one committed book.
A two-sided quote at half-spread $\sigma+s$, $s\ge 0$, serves its favored side at displacement $s$ over the realized peak and its unfavored side at displacement $2\sigma+s$: the same absolute price, measured from two mirror-image peaks.
A $p$-trader with valuation $m_i$ on $D$ obtains $\pi_i^{CLOB}(D)=\max\{m_i-A^D,\,0\}$, so by Assumption~\ref{ass:tail_demand_distribution}(iv) the served CLOB mass is $zH_F(s)$ on the favored side and $zH_U(2\sigma+s)$ on the unfavored side, with per-unit rents $s$ and $2\sigma+s$ respectively.
Figure~\ref{fig:no_pickoff} uses ask-side quoting and positive shock as an example.

\begin{figure}[H]
    \centering
    \includegraphics[width=0.75\linewidth]{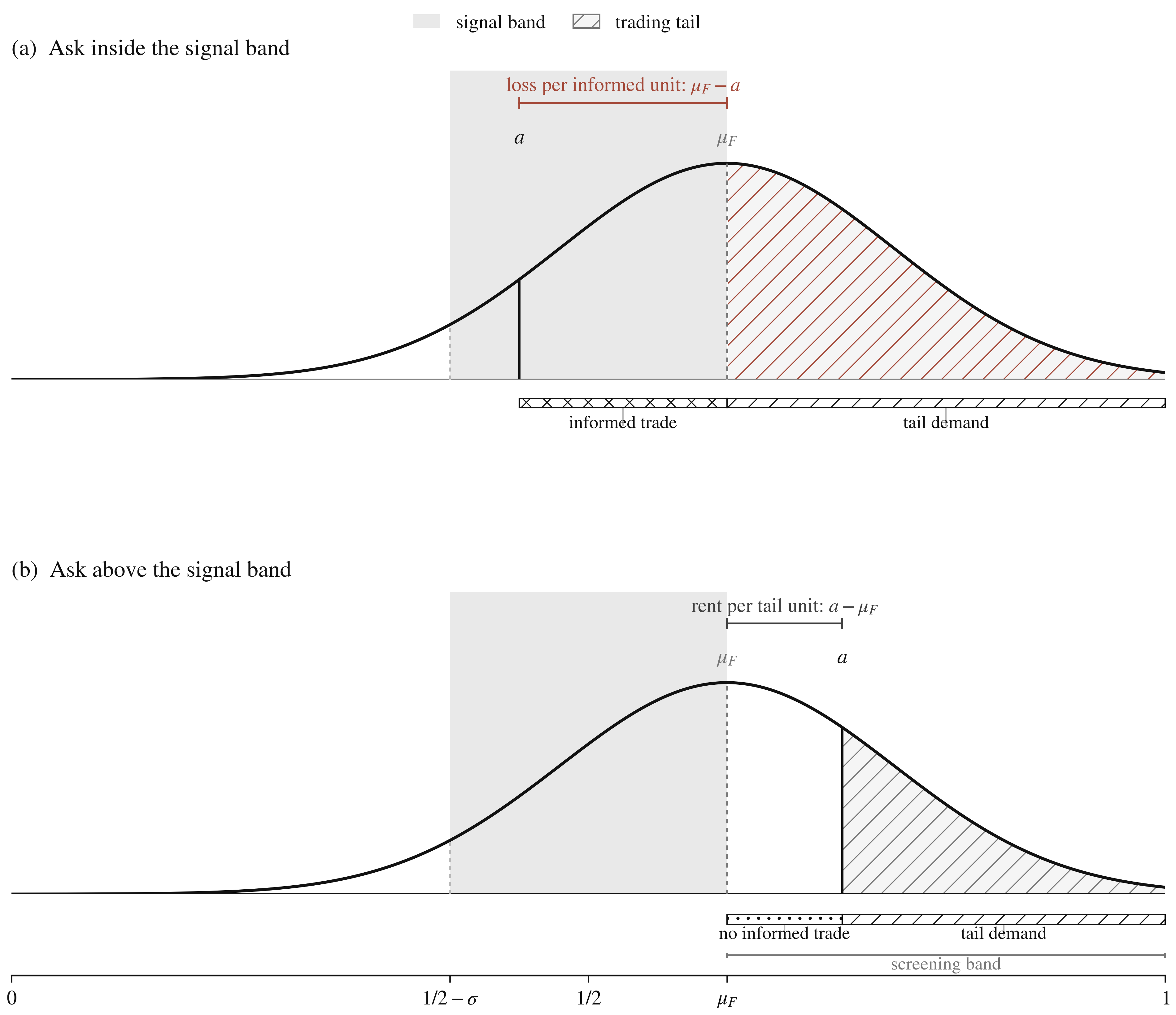}
    \caption{PnL of Quoting on Ask Leg: Inside v.s.~Outside Signal Band}
    \label{fig:no_pickoff}
\end{figure}

If the ask is lower than $\mu_F = \frac{1}{2}+\sigma$, and the active leg direction $D = F$, then after a non-negative mass of informed $c$-traders take the ask in Stage~1, all tail-demand $p$-trader will also take the CLOB ask until the quote is exhausted, the pink area under the curve are the mass of $p$-trader mass captured with the quote and this causes pure loss for the CLOB maker.
When quoting outside the signal band, informed $c$-traders are screened out of the CLOB, and the green area under the curve is the captured mass of $p$-traders which generates rent for the CLOB maker.

\subsection{LMSR}
\label{subsec:LMSR}

The LMSR is a cost-function maker for the complementary claims $Y$ and $N$, with liquidity parameter $L>0$ and net YES inventory $q$ measured from the initial YES price $p_0$.
For a claim $X\in\{Y,N\}$ let $p_X(q;p_0,L)$ be the raw marginal price and $P_X\equiv p_X/(1-\tau)$ the fee-adjusted purchase price; we write claim signs $\iota_Y=+1$, $\iota_N=-1$.
The following lemma derives the binary LMSR cost.

\begin{lemma}[Binary LMSR costs]
\label{lem:marginal_and_average_cost}
The raw marginal prices are
$$
p_Y(q;p_0,L)=\frac{p_0e^{q/L}}{p_0e^{q/L}+1-p_0},
\qquad p_N=1-p_Y,
$$
and the fee-adjusted marginal purchase price is $P_X=p_X/(1-\tau)$.
The finite-trade transaction cost $T_X(\cdot\,;q,p_0,L)$ and its fee-adjusted form $\widetilde T_X=T_X/(1-\tau)$ are derived in~\S\ref{app:proof_for_marginal_and_average_cost}.
\end{lemma}

The logistic form makes all capacities linear in $L$ once prices are measured in log-odds.
Define the logit distance from the $t=0$ price,
$$
\lambda(x)\;\equiv\;\operatorname{logit}\big((1-\tau)x\big)-\operatorname{logit}(p_0),
\qquad \operatorname{logit}(u)=\ln\tfrac{u}{1-u},
$$
so that moving the fee-adjusted price of a claim from level $a$ to level $b$ absorbs inventory $L[\lambda(b)-\lambda(a)]$, and write $\lambda_F\equiv\lambda(\mu_F)$ consistently with~\S\ref{subsec:platform_choice_and_eq}.
Two Stage-1 landmarks then follow from Lemma~\ref{lem:marginal_and_average_cost} and the fee.
First, on the side the shock favors, $c$-flow drives the fee-adjusted price exactly to the posterior $\mu_F$.
Second, on the side the shock disfavors, complementarity leaves the fee-adjusted price at $\mu_U+w_\tau$, where
$$
w_\tau\;\equiv\;\frac{\tau}{1-\tau}
$$
is the fee wedge: the taker fee is levied on purchases of either claim, so the buyer of the unfavored claim faces the posterior plus the wedge even though the raw price sits at $\mu_U$.

\subsection{Platform Choice and Equilibrium}
\label{subsec:platform_choice_and_eq}

Three conditions pin the continuation equilibrium for a fixed direction $D$: capacity for $c$-flow, market clearing for $p$-flow, and free entry for makers.
We state each with its economic content and defer the algebra to Appendix~\ref{app:proof_for_active_leg_existence}.
Throughout this subsection and the next we impose Convention~A.

\paragraph{Routing Convention A (firm-quote priority).}
Each trader transacts at the best venue available to her, with one tie-break: a $p$-trader who can execute at the CLOB's committed price does so.
Conditional on the realized active direction $D\in\{F,U\}$, traders priced out of the CLOB enter the LMSR while the corresponding LMSR path price remains below their valuation.
Within this band the pool serves traders in descending valuation order---those with the most surplus at stake win any contest for queue position---so each trader faces the marginal price left by all higher-valuation traders.
The favored and unfavored clearing paths are therefore mutually exclusive continuation states and are never realized together.
The convention reflects the immediacy value of a firm quote over a path whose terminal price depends on the queue.
Convention~B in~\S\ref{subsec:capacity_competition} removes the tie-break and lets venue compete for traders.

\paragraph{Capacity.}
\emph{C-traders} never transact on the CLOB, since any persistent quote leaves them nonpositive surplus; they enter the LMSR with mass $\beta_F$ until its fee-adjusted favored-claim price reaches their reservation value $\mu_F$,
\begin{equation}
\label{eq:c_capacity}
P_F(\iota_F\beta_F^*;p_0,L,\tau)\le \mu_F,
\qquad\text{with equality if } \beta_F^*<1 .
\end{equation}
The logistic price form in Lemma~\ref{lem:marginal_and_average_cost} makes the binding case $\beta_F^*(L)=L\lambda_F$, with $\lambda_F=\lambda(\mu_F)>0$ by the first inequality of Assumption~\ref{ass:low_fee}; we study the coexistence branch $L<\bar L\equiv 1/\lambda_F$ on which capacity binds before the unit $c$-mass is exhausted~($\beta_F^*(L)<1$).

\paragraph{Market Clearing.}
Conditional on the realized active direction, $p$-traders priced out of the CLOB fill the single LMSR path associated with that direction.
The two equations in~\eqref{eq:p_indifference} therefore characterize mutually exclusive continuation states.
Conditional on $D=F$, the terminal favored displacement is $y_F$ and satisfies $y_F\in(0,s)$.
Conditional on $D=U$, the terminal unfavored displacement is $y_U$ and satisfies
$$
y_U\in\big(w_\tau,\min\{2\sigma+s,\bar y_U\}\big).
$$
Each equation has a unique interior root because its left-hand side is continuous and strictly decreasing in the terminal displacement, while its right-hand side is continuous and strictly increasing.
\begin{equation}
\label{eq:p_indifference}
z\big[H_F(y_F)-H_F(s)\big]=L\big[\lambda(\mu_F+y_F)-\lambda_F\big],
\qquad
z\big[H_U(y_U)-H_U(2\sigma+s)\big]=L\big[\lambda(\mu_U+y_U)+\lambda_F\big],
\end{equation}
on $y_F\in(0,s)$ and $y_U\in(w_\tau,\,2\sigma+s)$ respectively.
Each equation has a unique interior root, since the left side falls continuously from a positive value to zero across its band while the right side rises from zero.\footnote{On the unfavored band the capacity term vanishes at $y=w_\tau$, since $(1-\tau)(\mu_U+w_\tau)=1-(1-\tau)\mu_F$ gives $\lambda(\mu_U+w_\tau)=-\lambda_F$ under $p_0=\tfrac12$}
The left side of each equation is the mass the priority discipline serves; the right side is the units the curve dispenses in reaching the cutoff; the terminal displacement is the unique price at which the two coincide.
Stated informally, the pool fills until the next-cheapest excluded trader can no longer afford it; the terminal price lands strictly between the posterior and the CLOB's committed level, so the LMSR neither saturates at the CLOB price nor exhausts the moderate mass.
The captured masses,
\begin{equation}
\label{eq:lmsr_mass}
M_F(L,s)\;=\;z\big[H_F(y_F)-H_F(s)\big],
\qquad
M_U(L,s)\;=\;z\big[H_U(y_U)-H_U(2\sigma+s)\big],
\end{equation}
are strictly increasing in $L$, while $y_F,y_U$ are strictly decreasing in $L$: a deeper pool ends closer to the posterior and captures more of the band at compressed prices.
The corresponding net YES inventories, one per realization of the committed direction, are
\begin{equation}
\label{eq:yes_inventory_p}
Q_Y^{\mathrm{fav}}\;=\;\iota_D\,L\,\lambda(\mu_F+y_F),
\qquad
Q_Y^{\mathrm{unf}}\;=\;\iota_D\,L\,\lambda(\mu_U+y_U),
\end{equation}
so the terminal inventory is the logit distance to the terminal price in either case.
The routing shares are the derived objects
$$
\alpha_F(L)\;\equiv\;\frac{M_F(L,s)}{zH_F(y_F)},
\qquad
\alpha_U(L)\;\equiv\;\frac{M_U(L,s)}{zH_U(y_U)},
$$
These are the LMSR's fractions of served direction-$D$ tail flow in the two mutually exclusive continuation states.
The favored share lies in $(0,1)$ and is strictly increasing in $L$.
If $2\sigma+s<\bar y_U$, the unfavored share also lies in $(0,1)$ and is strictly increasing in $L$.
If $2\sigma+s\ge\bar y_U$, then $H_U(2\sigma+s)=0$, the CLOB serves no unfavored tail demand, and $\alpha_U(L)=1$.
Traders with displacement below the applicable terminal displacement do not trade at either venue: their valuations lie below the committed quote, and below the marginal price at their point in the queue, which the priority discipline places at the terminal displacement.

\paragraph{Free Entry.}
Conditional on $D=F$, the committed book earns revenue $zsH_F(s)$, 
Conditional on $D=U$, it earns revenue $z(2\sigma+s)H_U(2\sigma+s)$, with the second term read as zero whenever $2\sigma+s\ge\bar y_U$. 
Write $R_F(s) = sH_F(s), R_U(s) = sH_U(s)$.
Since the two active-direction states occur with probability $\tfrac12$ each, zero ex-ante expected profit at per-book posting cost $\kappa$ requires
\refstepcounter{equation}\label{eq:lp_free_entry}
$$
\tfrac12 z\big[R_F(s^*)+R_U(s^*+2\sigma)\big]=\kappa.
\eqno(\theequation)
$$
The factor $\tfrac12$ is the probability of each mutually exclusive active-direction state and does not apply to the posting cost.
Realized CLOB revenue can differ across the two direction states, but risk-neutral entry depends only on the expectation above.
Under Convention~A, the CLOB clientele at its committed quote is determined by the tail distributions and does not depend on LMSR depth.
Consequently, the free-entry condition contains no $L$. the consequence is recorded in Proposition~\ref{prop:active_leg_existence}(iv) below and motivates~\S\ref{subsec:capacity_competition}.

Write $\Phi_D(s)\equiv\tfrac12 z\big[R_F(s)+R_U(s+2\sigma)\big]$ for the pooled screening revenue; it is single-peaked at some $s^\ddagger$ by Assumption~\ref{ass:tail_demand_distribution}(v).
Which point of this curve the market selects depends on how makers compete, and we bracket maker competition the same way Conventions~A and~B bracket trader routing.

\paragraph{Maker Convention I (contestable book).}
Posted quotes are firm, but the book stays open: at any time before the shock a new maker may post a better price against a standing quote.
A standing premium $s$ therefore survives only if no cut is profitable, i.e.\ only if $\Phi_D(s')\le\kappa$ for every $s'<s$, and the surviving zero-profit premium is the root of~\eqref{eq:lp_free_entry} on the \emph{increasing} branch of $\Phi_D$.
A root on the decreasing branch never survives: there a small cut \emph{raises} revenue above $\kappa$, so cutting continues.
The lower bound in Assumption~\ref{ass:low_fee}(ii) below puts the posting cost above the floor rent, so cutting is stopped by profits~\footnote{Were $\kappa\le\Phi_D(0)$, cutting would run to the floor and the premium would corner at $s^*=0$ and we prevent that with assumption.}.
We note that the floor rent $z\sigma H_U(2\sigma)$ is single-peaked in $\sigma$, so the lower bound binds hardest at intermediate shocks and relaxes for extreme ones, where the unfavored tail empties.

\paragraph{Maker Convention II (committed book).}
Entry is decided first and the posting cost $\kappa$ is sunk; quotes are then posted once, simultaneously, and never revised.
With two or more committed books, one-shot price competition drives all of them to the pickoff floor, where each earns at most $\tfrac12\Phi_D(0)<\tfrac12\kappa$ under Assumption~\ref{ass:low_fee}(ii); a second entrant therefore anticipates a loss and stays out, and the unique entry outcome is a single committed book, which prices at the revenue peak $s^\ddagger$.
Under this convention entry pins the number of books rather than their profit: the lone book keeps the rent $\Phi_D(s^\ddagger)-\kappa\ge0$.

The two conventions are the extremal intensities of maker-side price competition---most and least competitive---so results under both bracket any intermediate posting technology, exactly as Conventions~A and~B bracket routing.
Under trader Convention~A the choice is structurally immaterial: either premium ($s^{FE}$ or $s^\ddagger$) is pinned by $\Phi_D$ and $\kappa$ alone and carries no $L$.
Under trader Convention~B the two conventions separate, and the separation is what signs the spillover in Proposition~\ref{prop:binary_spillover}.

The far leg is not always reachable.
Its markup is $2\sigma+s$, so it contributes CLOB revenue only while $2\sigma+s<\bar y_U$.
Write
$$
\sigma^\dagger\equiv\tfrac12\bar y_U,
\qquad
s_x(\sigma)\equiv\bar y_U-2\sigma.
$$
The market is potentially two-sided when $\sigma<\sigma^\dagger$, because the unfavored tail is reachable at the no-pickoff floor $s=0$.
Within that region, the equilibrium CLOB actually serves unfavored tail demand only if $s^*<s_x(\sigma)$, equivalently $2\sigma+s^*<\bar y_U$.
When $2\sigma+s^*\ge\bar y_U$, the unfavored CLOB participation set is empty and the active unfavored tail clears entirely on the LMSR.
When $\sigma\ge\sigma^\dagger$, the market is structurally one-sided because no pickoff-proof premium can reach the unfavored tail.

\begin{assumption}[Viability and screening room]
\label{ass:low_fee}
The primitives $(p_0,\tau,\sigma,\kappa,z)$ and the constant $\bar s$ of Assumption~\ref{ass:tail_demand_distribution}(iii) satisfy
$$
\begin{aligned}
&\text{(i)}\ \ \tau<\frac{2\sigma}{1+2\sigma},
\qquad
\text{(ii)}\ \ \Phi_D(0)\;<\;\kappa\;<\;\sup_{s\in[0,\bar y_F)}\Phi_D(s),
\\
&\text{(iii)}\ \ s^\ddagger<\bar y_F,
\qquad
\text{(iv)}\ \ s^\ddagger\le\bar s,
\qquad
\text{(v)}\ \ w_\tau<\bar y_U ,
\end{aligned}
$$
where $\Phi_D(0)=z\sigma H_U(2\sigma)$ is the floor rent---equal to zero in the one-sided regime---and $s^\ddagger$ is the peak of $\Phi_D$ at zero depth.
\end{assumption}

Condition (i) makes the favored LMSR path viable and implies $w_\tau<2\sigma$.
Condition (ii) places the posting cost above the no-pickoff floor rent and below the peak screening revenue, which gives a unique increasing-branch free-entry root under maker Convention~I.
Condition (iii) keeps the favored-direction premium inside the favored tail's support.
Condition (iv) keeps every selected baseline premium inside the range on which the tail-ordering comparison is imposed.
Condition (v) leaves a nonempty unfavored LMSR recovery interval above the fee wedge.
No condition requires $2\sigma+s^*<\bar y_U$, because the model intentionally permits the equilibrium CLOB quote to exit the unfavored tail's support.

\begin{proposition}[Existence and structure of the Convention-A equilibrium]
\label{prop:active_leg_existence}
Suppose $0<\sigma<\tfrac12$.
Under Assumptions~\ref{ass:tail_demand_distribution} and~\ref{ass:low_fee} and maker Convention~I, every $L\in(0,\bar L)$ induces a unique state-contingent continuation policy
$$
\big(s^*,y_F(L),y_U(L),\alpha_F^*(L),\alpha_U^*(L),\beta_F^*(L)\big)
$$
that satisfies~\eqref{eq:c_capacity},~\eqref{eq:p_indifference}, and~\eqref{eq:lp_free_entry}.
The pair $(y_F,\alpha_F^*)$ is realized only when $D=F$, and the pair $(y_U,\alpha_U^*)$ is realized only when $D=U$.
The continuation policy has the following properties.

\smallskip
\noindent(i) \emph{(Unique premium.)}
The baseline premium is the unique increasing-branch free-entry root
$$
s^*=s^{FE}\in(0,s^\ddagger)\subset(0,\bar y_F).
$$

\smallskip
\noindent(ii) \emph{(Common-flow capacity and screening.)}
The capacity condition binds at
$$
\beta_F^*(L)=L\lambda_F\in(0,1),
$$
and every $c$-trader is strictly screened out of the CLOB.

\smallskip
\noindent(iii) \emph{(Favored active direction.)}
There is a unique terminal displacement $y_F(L)\in(0,s^*)$.
The favored-direction routing share satisfies
$$
\alpha_F^*(L)=1-\frac{H_F(s^*)}{H_F(y_F(L))}\in(0,1).
$$
Moreover, $y_F(L)$ is strictly decreasing in $L$, $\alpha_F^*(L)$ is strictly increasing in $L$, and $\alpha_F^*(L)=O(L)$ as $L\downarrow0$.
The CLOB retains strictly positive favored-direction mass $zH_F(s^*)$.

\smallskip
\noindent(iv) \emph{(Unfavored active direction.)}
If $2\sigma+s^*<\bar y_U$, there is a unique terminal displacement
$$
y_U(L)\in(w_\tau,2\sigma+s^*),
$$
and the routing share satisfies $\alpha_U^*(L)=1-\frac{H_U(s^* + 2\sigma)}{H_U(y_U(L))}\in(0,1)$.
On this branch, $y_U(L)$ is strictly decreasing in $L$, $\alpha_U^*(L)$ is strictly increasing in $L$, and $\alpha_U^*(L)=O(L)$ as $L\downarrow0$.
If $2\sigma+s^*\ge\bar y_U$, there is instead a unique terminal displacement
$$
y_U(L)\in(w_\tau,\bar y_U),
$$
the CLOB serves no unfavored tail demand, and $\alpha_U^*(L)=1$.
On this branch, $y_U(L)$ is strictly decreasing in $L$, $y_U(L)\uparrow\bar y_U$ as $L\downarrow0$, and the total unfavored mass served by the LMSR is $O(L)$.

\smallskip
\noindent(v) \emph{(Depth neutrality.)}
The Convention-A premium is independent of LMSR depth $\frac{ds^*}{dL}=0$.
\end{proposition}

The equilibrium is recursive.
Free entry pins $s^*$, the premium and depth pin the conditional clearing displacement in the realized active-direction state, and neither clearing equation feeds back into the Convention-A free-entry condition.
Under maker Convention~II, replace $s^*=s^{FE}$ by $s^*=s^\ddagger$ throughout Proposition~\ref{prop:active_leg_existence}.
The branch classification then depends on whether $2\sigma+s^\ddagger$ is below or at or above $\bar y_U$.
Because $\Phi_D$ contains no $L$ under Convention~A, the depth-neutrality conclusion is unchanged.
The spillover announced in~\S\ref{sec:market_info} therefore requires venue competition over clientele, which is exactly what Convention~B restores in~\S\ref{subsec:capacity_competition}.
See~\S\ref{app:proof_for_active_leg_existence} for the proof.

Under maker Convention~II the premium is the peak $s^\ddagger$ in place of the root $s^{FE}$, and claims (ii)--(iv) of Proposition~\ref{prop:active_leg_existence} hold verbatim: under firm-quote priority $\Phi_D$ carries no $L$, so the committed book's peak, like the contestable book's root, is depth-neutral.
The two maker conventions separate only once traders compete for the pool as in~\S\ref{subsec:capacity_competition}.

Convention~A can be seen as the CLOB maker(s) quote without an LMSR and LMSR is only a venue for informed traders to profit from their information and make the informed shock public in Stage~1.
We see that the CLOB maker(s) do not serve the role of information aggregator at all when they have to quote from a no-information state~(when the market price is uniform).

In Figure~\ref{fig:mass_split_eq}, we illustrate the mass split of $p$-traders between LMSR and CLOB, and thus the expected CLOB inventory holding when the shock is small v.s.~large.
When the shock is small, CLOB makers two-sided quotes can both attract $p$-trader mass and LMSR serve the rest in expectation; when the shock is large, a symmetric quote on the far-side of the tail can attract no $p$-trader mass and thus only extract rent, and hold share inventory, on one side of the market outcome. 

\begin{figure}[H]
\centering
\begin{subfigure}{.5\textwidth}
  \centering
  \includegraphics[width=1\linewidth]{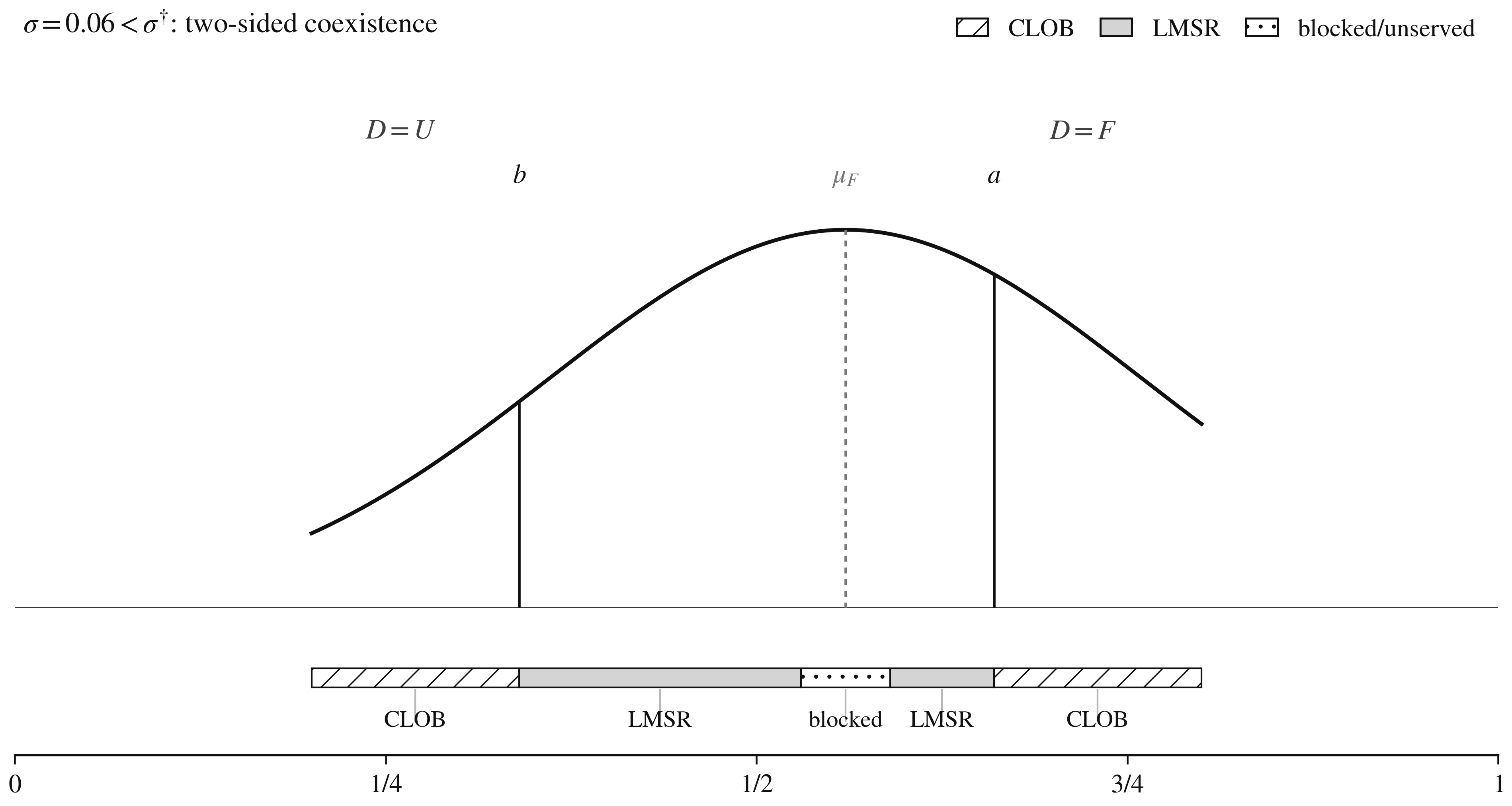}
  \caption{Equilibrium Mass Split - Small Shock}
  \label{fig:sub1}
\end{subfigure}%
\begin{subfigure}{.5\textwidth}
  \centering
  \includegraphics[width=1\linewidth]{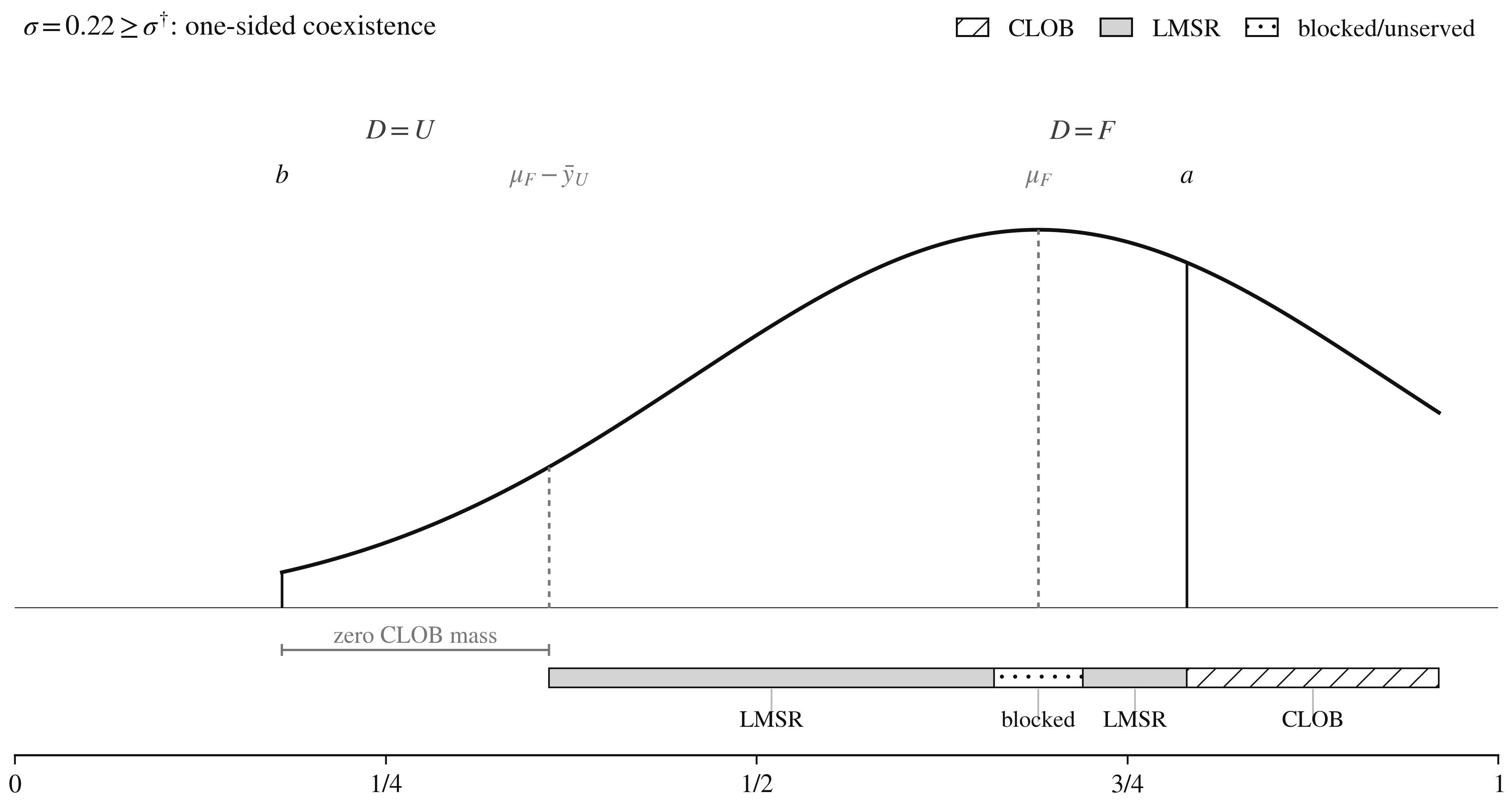}
  \caption{Equilibrium Mass Split - Large Shock}
  \label{fig:sub2}
\end{subfigure}
\caption{Equilibrium Mass Split Comparison}
\label{fig:mass_split_eq}
\end{figure}

\subsection{LP's LMSR Problem and Endogenous Interior Liquidity}
\label{subsec:lmsr_lp_problem}

Depth is costly on $c$-flow and profitable on $p$-flow, and the LP balances these two effects.
Stage-1 $c$-flow exposes the LMSR LP to adverse selection because informed traders buy at prices below the posterior they have observed.
The loss is linear in depth:
$$
\mathcal L_c(L)=\ell_c L,
\qquad
\ell_c\equiv\mu_F\lambda_F-\frac{1}{1-\tau}\ln\!\frac{1-p_0}{1-(1-\tau)\mu_F}>0.
$$
By mirror symmetry, the sign of the common shock does not affect the LP's payoff conditional on the active direction.
Conditional on either shock sign, the favored and unfavored active-direction states occur with probability $\tfrac12$ each.
The Stage-1 loss occurs in every realization, while exactly one of the two Stage-2 recovery terms is realized.
The ex-ante payoff is therefore
\refstepcounter{equation}\label{eq:lp_payoff}
$$
\Pi^{LMSR}(L)
=-\ell_c L
+\frac{L}{2}\int_0^{y_F(L)}y\,d\lambda(\mu_F+y)
+\frac{L}{2}\int_{w_\tau}^{y_U(L)}y\,d\lambda(\mu_U+y).
\eqno(\theequation)
$$
The first integral is the recovery conditional on $D=F$, and the second is the recovery conditional on $D=U$.
They are probability-weighted conditional payoffs rather than two recoveries earned in the same realization.

Because the LMSR cost depends on inventory only through $q/L$, the payoff is homogeneous of degree one in depth given the per-unit-liquidity intensities, which is the factorization $\Pi^{LMSR}(L)=L\,G\big(y_F(L),y_U(L)\big)$ used in the proof; the explicit accounting is written out in~\S\ref{App:proof_of_interior_L}.

The trade-off in $L$ is transparent in~\eqref{eq:lp_payoff}.
Marginal $c$-loss is the constant $\ell_c$.
Marginal recovery is the band integral evaluated at the current terminals, and it is eroded from two directions as depth grows: the added capacity is filled at prices closer to the posterior, and the terminals $y_F,y_U$ themselves fall, compressing every inframarginal margin.
Too little depth starves the recovery channel---as $L\downarrow0$ the terminals rise to the CLOB ceilings but the captured mass vanishes---while too much depth expands the loss-making $c$-flow, $\beta_F=L\lambda_F$, against recovery that flattens.
We impose three payoff-regularity conditions this logic needs.

\begin{assumption}[Profitable recovery and diminishing returns to depth]
\label{ass:lmsr_recovery_boundary}
On $\mathcal I=(0,\bar L)$, the following conditions hold.

\smallskip
\noindent(i) \emph{(Profitable recovery at zero.)}
$ \ell_c <\frac12\int_0^{s^*}y\,d\lambda(\mu_F+y)
+\frac12\int_{w_\tau}^{\min\{2\sigma+s^*,\bar y_U\}}y\,d\lambda(\mu_U+y).$

\smallskip
\noindent(ii) \emph{(Diminishing returns.)}
The marginal recovery
$\rho(L)\equiv\frac{d}{dL}\big[\Pi^{LMSR}(L)+\ell_c L\big]$
is strictly decreasing on $\mathcal I$ and satisfies 
$\lim_{L\uparrow\bar L}\rho(L)<\ell_c$.

\smallskip
\noindent(iii) \emph{(No profitable off-branch deviation.)}
Choices outside the capacity-binding Convention-A branch do not yield a payoff above the branch optimum.
\end{assumption}

Clause (i) evaluates the two conditional recovery bands at their effective zero-depth terminal displacements.
The unfavored endpoint is $2\sigma+s^*$ when the CLOB retains unfavored flow and $\bar y_U$ when the CLOB quote lies outside the unfavored tail's support.
Clause (ii) supplies the single-crossing property that delivers uniqueness.
Clause (iii) lets the branch solution be compared with choices outside $\mathcal I$.

\begin{proposition}[Interior optimal LMSR liquidity]
\label{prop:interior_optimal_L}
Under Assumptions~\ref{ass:tail_demand_distribution},~\ref{ass:low_fee}, and~\ref{ass:lmsr_recovery_boundary}, the problem
$$
\max_{L\in\mathcal I}\Pi^{LMSR}(L)
$$
has a unique solution $L^*\in\mathcal I$.
The solution is characterized by
\refstepcounter{equation}\label{eq:lp_foc}
$$
\ell_c
=\rho(L^*)
=\left.\frac12\frac{d}{dL}\left[
L\int_0^{y_F(L)}y\,d\lambda(\mu_F+y)
+L\int_{w_\tau}^{y_U(L)}y\,d\lambda(\mu_U+y)
\right]\right|_{L=L^*}.
\eqno(\theequation)
$$
The optimal payoff is strictly positive.
At $L^*$, $\beta_F^*(L^*)\in(0,1)$ and $\alpha_F^*(L^*)\in(0,1)$.
If $2\sigma+s^*<\bar y_U$, then $\alpha_U^*(L^*)\in(0,1)$.
If $2\sigma+s^*\ge\bar y_U$, then $\alpha_U^*(L^*)=1$ because the CLOB serves no unfavored tail demand.
In each possible active-direction state, the corresponding recovery integral in~\eqref{eq:lp_payoff} is strictly positive.
Every unfavored-direction LMSR trade earns a margin of at least $w_\tau$ over the unfavored posterior.
By Assumption~\ref{ass:lmsr_recovery_boundary}(iii), $L^*$ is also optimal relative to choices outside the capacity-binding branch.
\end{proposition}

Intuitively, the LP stops deepening the pool when the marginal expected recovery from the two mutually exclusive active-direction states equals the marginal adverse-selection loss on common flow.
The proposition is stated at the Convention-I premium $s^*=s^{FE}$.
Under maker Convention~II, the same result applies if Assumption~\ref{ass:lmsr_recovery_boundary} is imposed on the Convention-II payoff with $s^*=s^\ddagger$.
Clause (i) is weakly easier to satisfy under Convention~II because the effective recovery endpoints weakly expand, but clauses (ii) and (iii) must still be imposed for that convention.

See~\S\ref{App:proof_of_interior_L} for the proof.
This closes the Convention-A model: makers' premium from free entry, LMSR terminals from clearing, optimal depth from~\eqref{eq:lp_foc}, each block feeding the next and none feeding back.

\subsection{Capacity Competition, Crowd-Out, and the Liquidity Spillover}
\label{subsec:capacity_competition}

This subsection removes Convention~A's firm-quote priority and lets traders compete for the LMSR's cheaper early path units.
The allocation rule is introduced as a convention, and the resulting residual CLOB volume is then derived as a lemma.

\paragraph{Routing Convention B (price competition).}
Conditional on the realized active direction, $p$-traders now compare the marginal price at LMSR and CLOB when executing a trader.
At a path price $p<A^D$, a trader with $m_i\ge A^D$ gains $A^D-p$ from receiving the LMSR unit rather than using her CLOB outside option.
All CLOB-eligible traders therefore also trade for LMSR units available below $A^D$, and they stop when the terminal price of LMSR exceeds CLOB.
This allocation rule defines Convention~B.

\begin{lemma}[Residual CLOB volume under price competition]
\label{lem:queue_priority}
Fix a committed baseline premium $s\ge0$ and LMSR depth $L>0$.
Under Convention~B, the CLOB volumes conditional on the favored- and unfavored-direction states are
\refstepcounter{equation}\label{eq:residual_volumes}
\begin{equation}
V_F(s,L)=\Big(zH_F(s)-L\big[\lambda(\mu_F+s)-\lambda_F\big]\Big)_+,
\qquad
V_U(s,L)=\Big(zH_U(2\sigma+s)-L\big[\lambda(\mu_U+2\sigma+s)+\lambda_F\big]\Big)_+.    
\end{equation}
If $V_F(s,L)>0$, the LMSR reaches the favored CLOB quote and the CLOB serves exactly $V_F(s,L)$ traders.
If $V_F(s,L)=0$, all favored-direction traders who could use the CLOB are absorbed by the LMSR, and the terminal displacement is the unique $y\in(0,s]$ satisfying
$$
zH_F(y)=L\big[\lambda(\mu_F+y)-\lambda_F\big].
$$
If $V_U(s,L)>0$, the LMSR reaches the unfavored CLOB quote and the CLOB serves exactly $V_U(s,L)$ traders.
If $V_U(s,L)=0$, all unfavored-direction traders who could use the CLOB are absorbed by the LMSR, and the terminal displacement is the unique $y\in(w_\tau,2\sigma+s]$ satisfying
$$
zH_U(y)=L\big[\lambda(\mu_U+y)+\lambda_F\big].
$$
Whenever the corresponding CLOB-eligible mass is positive, the Convention-B LMSR routing shares are
$$
\alpha_F(s,L)=1-\frac{V_F(s,L)}{zH_F(s)},
\qquad
\alpha_U(s,L)=1-\frac{V_U(s,L)}{zH_U(2\sigma+s)}.
$$
If the corresponding CLOB-eligible mass is zero, its routing share is defined as one.
\end{lemma}

The lemma separates the maintained allocation rule from its implication: Convention~B relaxes the tie break and let traders execute based solely on price, thus increasing depth $L$ will steal $p$-trader volume from CLOB when they find it profitable, and Lemma~\ref{lem:queue_priority} computes the volume left for the CLOB.
Conventions~A and~B are the two extremal allocations of path rents, respectively the most CLOB-favorable and the most LMSR-favorable.
The LMSR payment identity is unchanged because cost-function payments depend only on terminal inventory; Convention~B changes the terminal displacement, not the payment formula conditional on that terminal.
See~\S\ref{App:proof_for_queue_priority} for the proof.

Figure~\ref{fig:mass_split_eq_B}'s comparison illustrate how business stealing work under Convention~B compare to A.
When firm quote is priority, $p$-traders who have valuation higher than ask~(lower than bid) execute transaction on CLOB instantly, the mass is indicated by the green area of left figure; under Convention~B, the $p$-traders who originally trade at CLOB will compare the marginal price of LMSR $y_F(L)$~($y_U(L)$), which in Convention~A, satisfy $y_F \in (0,s)$~($y_U \in (0, 2\sigma+s)$), and thus will be diverted to LMSR until the marginal price meets the quote, the new purple area in the right plot is the diverted $p$-trader mass.

\begin{figure}[H]
\centering
\begin{subfigure}{.5\textwidth}
  \centering
  \includegraphics[width=0.9\linewidth]{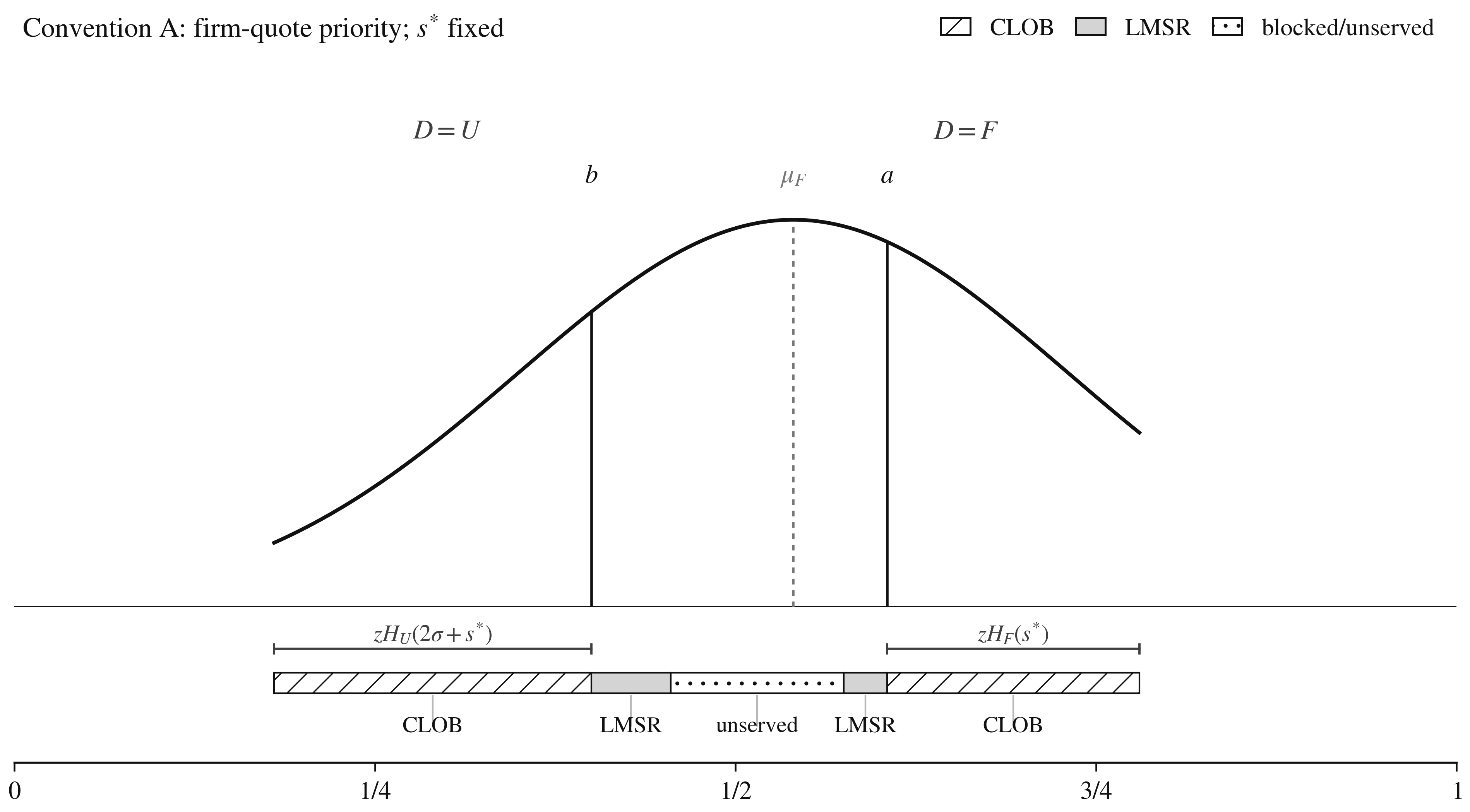}
  \caption{Equilibrium Mass Split - Convention~A}
  \label{fig:sub1}
\end{subfigure}%
\begin{subfigure}{.5\textwidth}
  \centering
  \includegraphics[width=0.9\linewidth]{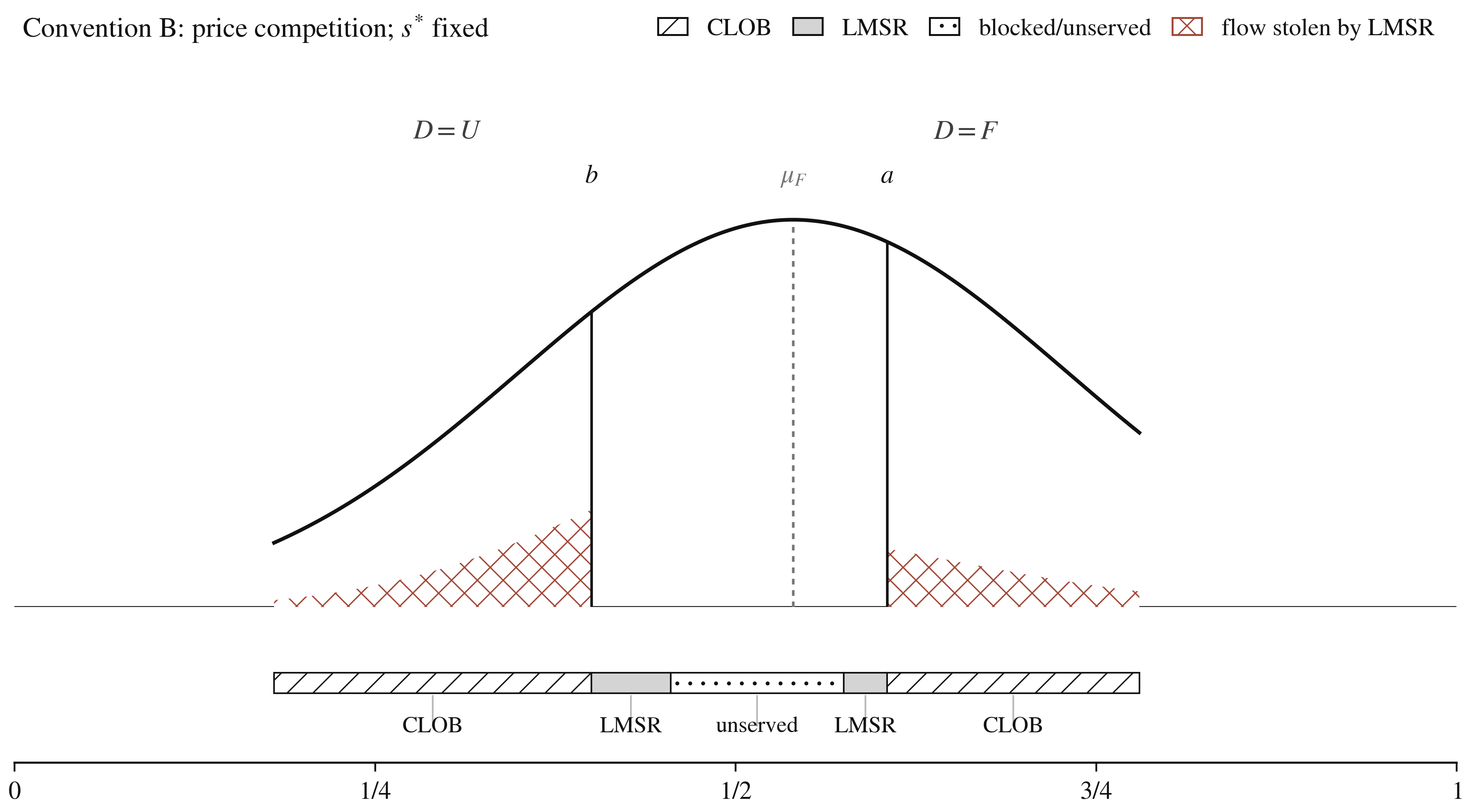}
  \caption{Equilibrium Mass Split - Convention~B}
  \label{fig:sub2}
\end{subfigure}
\caption{Convention~A v.s.~B: Business Stealing}
\label{fig:mass_split_eq_B}
\end{figure}

Under Convention~B, expected maker revenue at premium $s$ and depth $L$ is
\refstepcounter{equation}\label{eq:lp_free_entry_B}
$$
\Phi_D(s;L)\equiv\tfrac12\Big[sV_F(s,L)+(2\sigma+s)V_U(s,L)\Big],
\qquad
\Phi_D(s;0)=\Phi_D(s).
\eqno(\theequation)
$$
\paragraph{Depth-uniform Revenue Shape.}
For every $L\in[0,\bar L)$ for which $\sup_s\Phi_D(s;L)>0$, we maintain the depth-uniform extension of Assumption~\ref{ass:tail_demand_distribution}(v): each nonzero side-revenue term in $\Phi_D(\cdot;L)$ is strictly quasi-concave on its positive-revenue support, and $\Phi_D(\cdot;L)$ has a unique maximizer, is strictly increasing before that maximizer, and is strictly decreasing after it while expected revenue remains positive.
At $L=0$, this restriction reduces to Assumption~\ref{ass:tail_demand_distribution}(v) together with the side-level single-peakedness implied by log-concavity.
The depth-uniform extension is required because subtracting LMSR capacity from tail demand does not mechanically preserve the zero-depth quasi-concavity restriction.

Under maker Convention~I, the selected premium is the increasing-branch root of $\Phi_D(s;L)=\kappa$.
Under maker Convention~II, the selected premium is the unique maximizer $s^\ddagger(L)$ of $\Phi_D(\cdot;L)$.
Under either maker convention, a book operates if and only if $\sup_s\Phi_D(s;L)\ge\kappa$.

\begin{proposition}[LMSR depth and the CLOB premium]
\label{prop:binary_spillover}
Maintain Convention~B and the depth-uniform revenue-shape restriction above, and consider the capacity-binding region $L\in(0,\bar L)$ in which a book can enter.

\smallskip
\noindent(i) \emph{(Contestable book: depth widens the premium.)}
Under maker Convention~I, the selected premium $s_D^*(L)$ is the unique increasing-branch root of $\Phi_D(s_D^*(L);L)=\kappa$.
The root is continuous and strictly increasing in $L$ throughout the entry region.
On every differentiable segment on which the set of positive residuals is unchanged,
$$
\frac{ds_D^*}{dL}=-\frac{\partial\Phi_D/\partial L}{\partial\Phi_D/\partial s}>0.
$$
The routing shares $\alpha_F^*(L)$ and $\alpha_U^*(L)$ are nondecreasing in $L$ and are strictly increasing until the corresponding residual reaches zero, after which that share equals one.
If the last-entry depth $L^{det}$ in Proposition~\ref{prop:crowdout_deterrence}(ii) lies below $\bar L$, then
$$
\lim_{L\uparrow L^{det}}s_D^*(L)=s^\ddagger(L^{det}),
\qquad
\Phi_D(s^\ddagger(L^{det});L^{det})=\kappa.
$$

\smallskip
\noindent(ii) \emph{(Committed book: depth compresses the premium.)}
Under maker Convention~II, the selected premium is the unique peak $s_D^*(L)=s^\ddagger(L)$ of $\Phi_D(\cdot;L)$.
The peak is continuous and nonincreasing in $L$ throughout the entry region.
At every regular interior peak on a differentiable residual branch,
$$
\frac{ds^\ddagger}{dL}=-\frac{\partial^2\Phi_D/\partial s\,\partial L}{\partial^2\Phi_D/\partial s^2}<0.
$$
The routing shares are nondecreasing in $L$ and are strictly increasing until full crowd-out whenever the following condition holds, with the maximum taken only over directions whose CLOB-eligible mass is positive:
\refstepcounter{equation}\label{eq:alpha_monotone_condition}
$$
L\left|\frac{ds^\ddagger}{dL}\right|
\max\left\{
\frac{\lambda'(\mu_F+s^\ddagger)}{\lambda(\mu_F+s^\ddagger)-\lambda_F}-\frac{H_F'(s^\ddagger)}{H_F(s^\ddagger)},
\frac{\lambda'(\mu_F+s^\ddagger)}{\lambda(\mu_F+s^\ddagger)+\lambda_F}-\frac{H_U'(2\sigma+s^\ddagger)}{H_U(2\sigma+s^\ddagger)}
\right\}<1.
\eqno(\theequation)
$$

\smallskip
\noindent(iii) \emph{(Entry margin and common markup movement.)}
Under either maker convention, $\sup_s\Phi_D(s;L)$ is strictly decreasing in $L$ whenever it is positive.
Hence the set of depths at which a book can enter is an interval, and no maker enters above its boundary $L^{det}$ whenever that boundary lies below $\bar L$.
In both active-direction states the CLOB markups are $s_D^*(L)$ and $2\sigma+s_D^*(L)$, so they move by the same $ds_D^*/dL$ and retain the fixed wedge $2\sigma$.
\end{proposition}

The mechanism has a common component and a maker-convention component.
Added LMSR depth steals CLOB business at every premium for which expected maker revenue is positive.
On a smooth branch with both residuals positive,
$$
\frac{\partial\Phi_D}{\partial L}
=-\tfrac12\Big[s\big(\lambda(\mu_F+s)-\lambda_F\big)+(2\sigma+s)\big(\lambda(\mu_U+2\sigma+s)+\lambda_F\big)\Big]<0,
$$
and
$$
\frac{\partial^2\Phi_D}{\partial s\,\partial L}
=-\Big[\lambda(\mu_F+s)+(\sigma+s)\lambda'(\mu_F+s)\Big]<0.
$$
If one residual is already zero, the same two signs hold after dropping that side's term.
Under Convention~I, the book moves right along the increasing branch to restore zero profit, so the premium \textit{rises}.
Under Convention~II, the book \textit{re-optimizes}, and the strictly negative cross-partial moves every regular peak left, so the premium \textit{falls} on each smooth branch.
The positive-part kinks are handled by Lemma~\ref{lem:queue_priority}: each residual cutoff itself moves left with depth.
See~\S\ref{App:proof_for_prop_binary_spillover} for the proof and the Figure~\ref{fig:spillover_effect} below for the spillover effect under two maker conventions.

\begin{figure}[H]
\centering
\begin{subfigure}{.5\textwidth}
  \centering
  \includegraphics[width=0.9\linewidth]{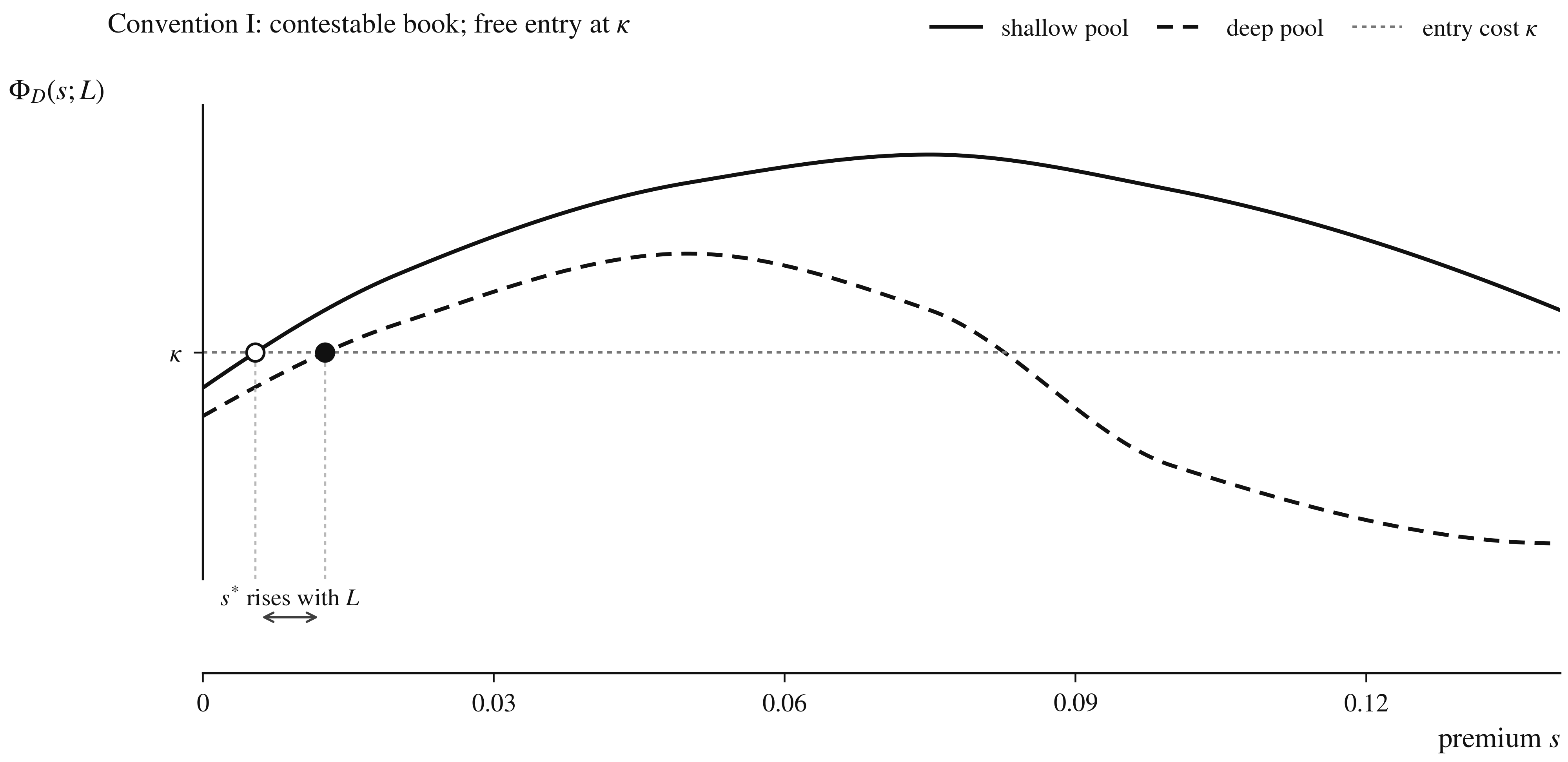}
  \caption{Positive Spillover - Maker Convention~I}
  \label{fig:sub1}
\end{subfigure}%
\begin{subfigure}{.5\textwidth}
  \centering
  \includegraphics[width=0.9\linewidth]{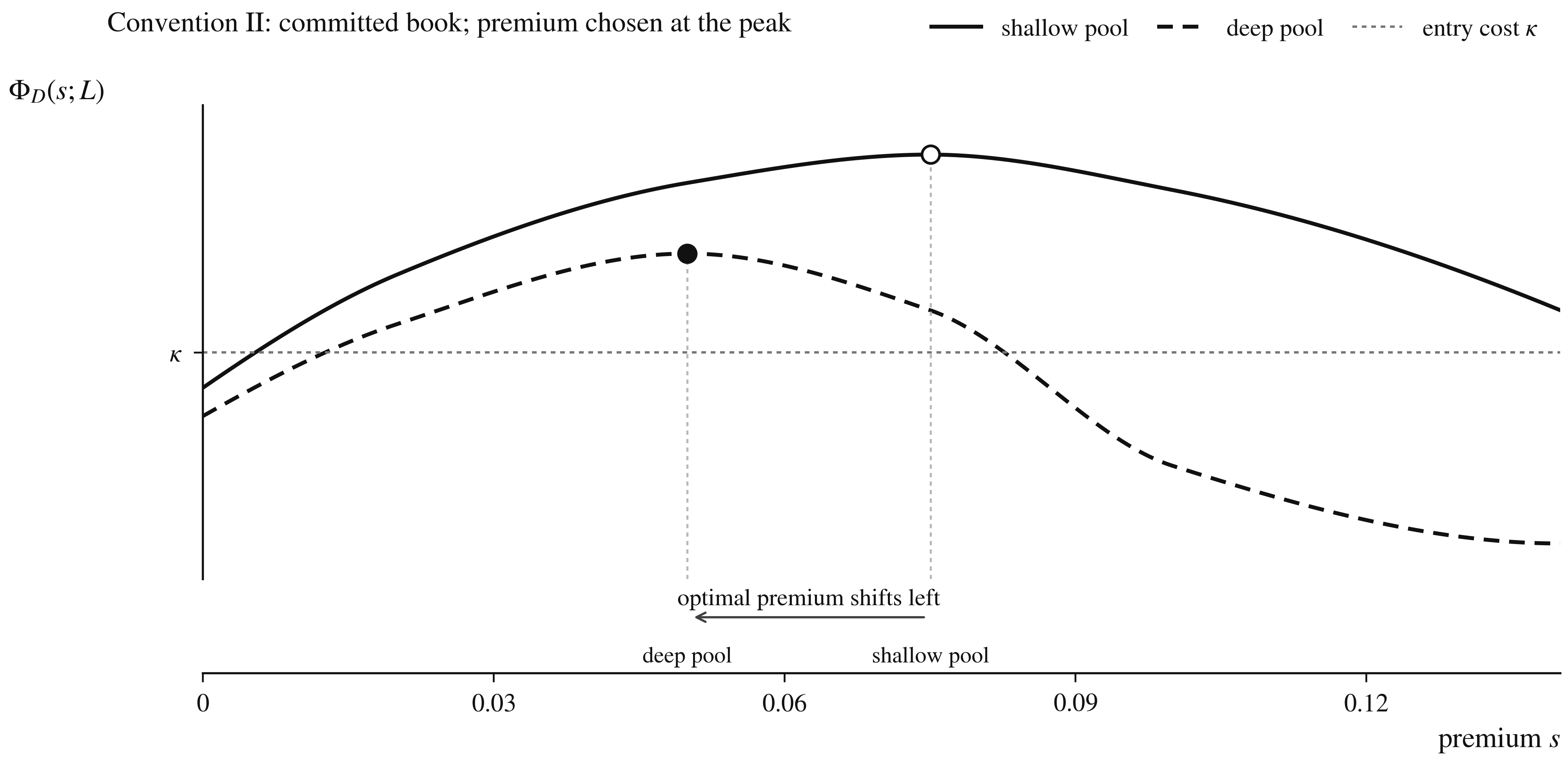}
  \caption{Negative Spillover - Maker Convention~II}
  \label{fig:sub2}
\end{subfigure}
\caption{Maker Convention~I v.s.~II: Positive v.s.~Negative Spillover}
\label{fig:spillover_effect}
\end{figure}

Depth does more than move the premium.
Proposition~\ref{prop:binary_spillover} shows that entry viability for CLOB makers shrinks monotonically with depth, while the positive parts in Lemma~\ref{lem:queue_priority} identify when each conditional CLOB residual disappears.
Proposition~\ref{prop:crowdout_deterrence} now records the side-specific crowd-out thresholds and the final entry boundary for the CLOB maker.

\begin{proposition}[Crowd-out and entry deterrence]
\label{prop:crowdout_deterrence}
Consider an equilibrium pair $(s_D^*(L),L)$ in the Convention-B entry region of Proposition~\ref{prop:binary_spillover}.

\smallskip
\noindent(i) \emph{(One-sided crowd-out.)}
At any fixed premium $s>0$, the residuals in~\eqref{eq:residual_volumes} vanish at
$$
L^{co}_U(s)=\frac{zH_U(2\sigma+s)}{\lambda(\mu_U+2\sigma+s)+\lambda_F},
\qquad
L^{co}_F(s)=\frac{zH_F(s)}{\lambda(\mu_F+s)-\lambda_F}.
$$
For every $\sigma>0$ and $s\in(0,\bar s]$ with $H_F(s)>0$, we have $L^{co}_U(s)<L^{co}_F(s)$.
At an equilibrium pair satisfying
$$
L^{co}_U(s_D^*(L))\le L<L^{co}_F(s_D^*(L)),
$$
the unfavored-direction state clears entirely on the LMSR while the favored-direction state retains positive CLOB volume.

\smallskip
\noindent(ii) \emph{(Entry deterrence.)}
Define the boundary of the entry interval identified in Proposition~\ref{prop:binary_spillover} by
$$
L^{det}\equiv\inf\left\{L\in(0,\bar L):\sup_s\Phi_D(s;L)<\kappa\right\}.
$$
If this set is nonempty, then no maker enters for $L>L^{det}$, and the LMSR clears both active-direction states at the unconstrained displacements characterized in Lemma~\ref{lem:queue_priority}.
If the unique unconstrained favored displacement at $L^{det}$ exceeds $s^\ddagger(L^{det})$, which is the last coexistence premium under either maker convention, the removal of the CLOB ceiling creates an upward jump in the LP's value at the entry boundary.
Whenever the best monopoly payoff above $L^{det}$ exceeds the best coexistence payoff below it, the Stackelberg LP optimally deters maker entry.
\end{proposition}

Part~(i) follows directly from Lemma~\ref{lem:queue_priority}.
The two CLOB ceilings coincide in absolute price because $\mu_U+2\sigma+s=\mu_F+s$, while their LMSR capacities differ by the common-flow offset $2\lambda_F$ and their eligible masses satisfy the tail ordering in Assumption~\ref{ass:tail_demand_distribution}(iii).
The unfavored residual therefore vanishes first at every fixed premium in the coverage range.
Economically, the unfavored LMSR path starts near its posterior while the CLOB quote sits a full $2\sigma+s$ away, so relatively shallow depth can absorb all CLOB-eligible unfavored flow.

Part~(ii) uses Proposition~\ref{prop:binary_spillover}(iii), which makes the peak entry value strictly decreasing in depth and therefore makes the maker-entry set an interval.
While the CLOB remains active, its committed quote caps the LMSR terminal displacement.
Deterring entry removes that cap and lets the LMSR clear at the unconstrained displacement from Lemma~\ref{lem:queue_priority}.
See~\S\ref{App:proof_for_crowdout_deterrence} for the proof and Figure~\ref{fig:LMSR_rev_jump} for a plot on the LMSR LP's payoff curve on the coexistence and~(possible) monopolistic equilibrium path under maker Convention~I. 

Figure~\ref{fig:LMSR_rev_jump} is one payoff surface with two branches, separated by whether maker's quotes present. 
Below $L^{\text {det }}$ a maker enters and its committed quote caps the pool's terminal displacement at $s^*$, truncating the recovery band in both active-direction states; the solid lower branch is the pool's payoff under that cap, and it rises with depth because the Convention-I premium rises with depth and widens the band the pool is permitted to fill. 
The dashed continuation is the counterfactual where the same depth would earn with no maker standing, therefore the vertical distance between the two curves measures what the CLOB ceiling costs the provider. 
At the first jump, makers shift from holding both side of the event to just one-side; then at $L^{\text {det }}$ the entry condition $\sup _s \Phi_D(s ; L)<\kappa$ binds, the book exits, the ceiling is removed, and the payoff jumps.
Above the boundary the pool clears both states at unconstrained displacements, and the branch is single-peaked for the usual reason: marginal recovery falls as added capacity fills at prices nearer the posterior, while the marginal adverse-selection loss $\ell_c$ on common flow is constant. 
The maximum sits strictly to the right of $L^{\text {det }}$ and the arrow is a discontinuous jump in chosen depth. 

None of this ranks the two venue forms.
Deterrence requires a conjunction that Proposition~\ref{prop:crowdout_deterrence} (ii) states as hypotheses rather than derives, such that the unconstrained favored displacement at $L^{\text {det }}$ must exceed the last coexistence premium, so that lifting the ceiling actually raises the terminal, and the best monopoly payoff must exceed the best coexistence payoff. 
The model hands the maker a rather thin instrument, one committed two-sided quote at a single premium with no choice to revise quotes as the state changes.
Books are prevalent in practice largely because they grant makers precisely that expressivity, and this model suppresses it in order to isolate the depth spillover - which means the absence of a book above $L^{\text {det }}$ is a statement about the model's maker, not about market makers in practice. 
What the figure establishes is that, a Stackelberg liquidity provider facing a capacity-constrained book may, on part of the primitive space, prefer to deepen past the maker entry boundary rather than share the flow. 

\begin{figure}[H]
    \centering
    \includegraphics[width=0.9\linewidth]{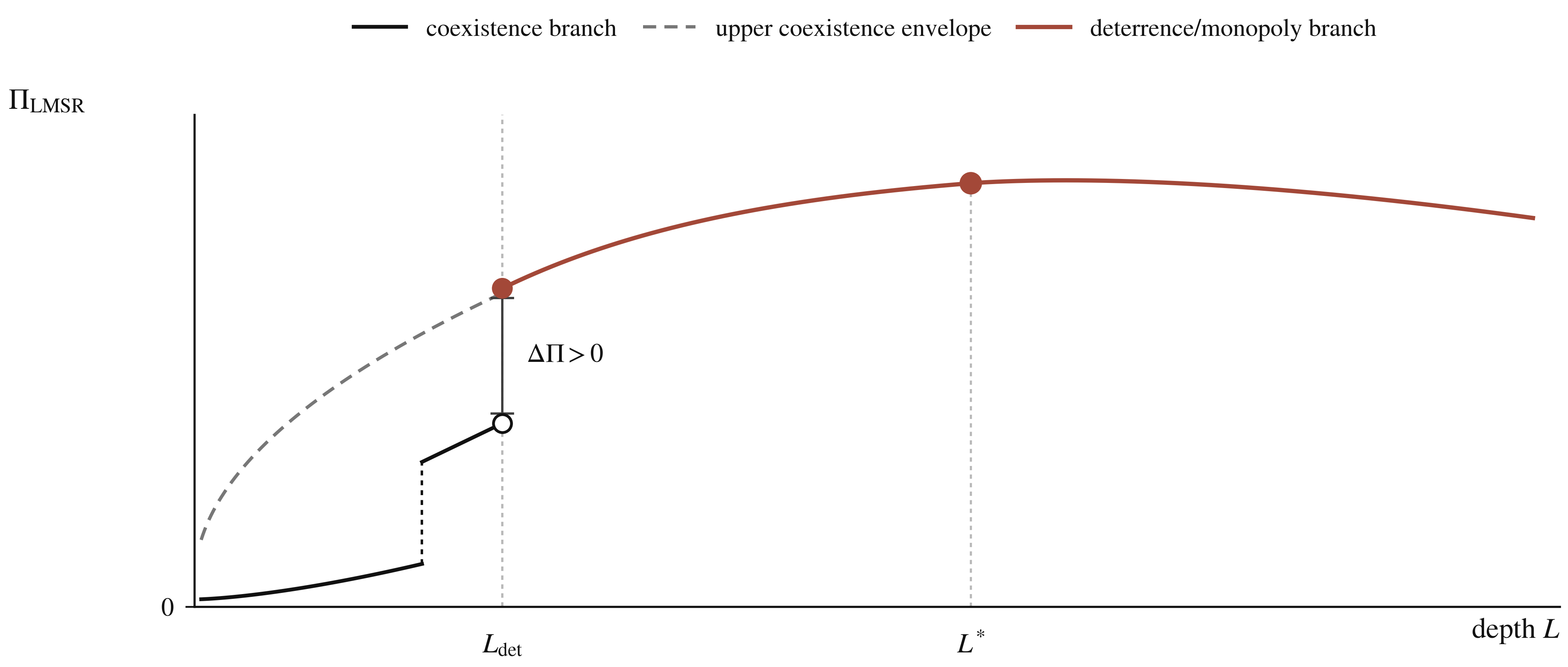}
    \caption{LMSR LP's Equilibrium Payoffs}
    \label{fig:LMSR_rev_jump}
\end{figure}

Two closing remarks scope the pair of conventions.
First, the observable content differs sharply: Convention~A predicts that when pool depth $L$ does not affect CLOB maker's tail-demand harvesting, their inventory is two-sided when volatility is low and becomes one-sided when volatility is high; while Convention~B predicts depth-driven premium movement---positive co-movement under contestable making, compression under a committed maker---together with migration of maker volume to the AMM and possible full crowd-out of CLOB makers.
Second, the clientele differs: under A the $p$-traders with valuation below quote trade the LMSR path and makers capture rents from the ones with valuation above it, while under B, CLOB becomes a outside option for $p$-traders and they only choose CLOB when the marginal price at LMSR exceeds the CLOB quote.

\section{Three-Outcome Extension}
\label{sec:3_outcome}

The binary model explained who finances a committed book; with more outcomes the question becomes what a committed book \emph{sells}, and the data have already constrained the answer.
Two measured facts from \S\ref{sec:stylized_facts}, documented in \S\ref{app:empirical_spread}, discipline this section.
First, completing an outcome set on separate binary books costs more than a dollar, the excess grows in the outcome count ($+0.060$ per unit of $\log n$, time-weighted over each book's life), and it is almost entirely ($96\%$) per-leg markup: it appears identically on the YES and NO sides of the same events, roughly half of it is reproduced by bundles of unrelated binary markets that share no event structure at all, and it vanishes at execution size.
Whatever the multi-outcome book is doing, it attaches a committed premium to each leg \emph{separately}---a premium that accumulates across the set whether or not the legs form a coherent event.
Second, the books' displayed centers sum to one dollar on average over a book's life, yet drift above par by $+0.02$ to $+0.08$ per $\log n$ whenever books are compared at the same age---most at long horizons, least, though still present, inside the final day.
So the \emph{level} of the display is not pinned by trade: it is exactly right on average, predictably wrong at any given moment, and corrected only where correcting is cheap.
This section derives both properties from the committed book's primitives---the per-leg premium from the switch quotes of \S\ref{subsec:3o_info} (equation~\eqref{eq:switch_prices_and_premium}), the unpinned level from the gauge of \S\ref{subsec:3o_beliefs}, and the cost of tethering that level from the balance sheet of \S\ref{subsec:3o_balance}---and shows that a single cost-function potential removes both at the source (Proposition~\ref{prop:self_collateral}(i)).

With three or more outcomes, binary complementarity no longer holds: in a match with outcomes ``home win,'' ``draw,'' and ``away win,'' the claim opposite to a home win is not one outcome but two, so buying NO on the home side says nothing about draw versus away.

The atomic belief object in a multi-outcome market is a pairwise preference.
To trade one, we define the switch---buy the favored claim, sell the unfavored one---and a maker who serves switches must commit a vector of primitive asks and bids across three books.
The CLOB maker commits a vector of bid-ask pairs and the LMSR still commits a potential $L$; and the gap between the two shows up as escalating costs for the book, each with an elementary proof and each vanishing at $n=2$, which is why the binary model could not see them.
\S\ref{subsec:3o_beliefs} isolates the common coordinate of the committed vector and shows that no trade prices it: the displayed level is unidentified, stale wherever a normalization is imposed, and left free even by competition, which pins the premium and nothing else.
\S\ref{subsec:3o_balance} then asks what capital stands behind that display, and finds the book short on two counts: its level can be relocated with no bet placed, and its per-book collateral exceeds the LMSR's event-level requirement by a factor $\Theta(n/\log n)$ as the outcome space grows.

We do not re-derive the coexistence equilibrium at $n\ge3$ under Convention~B:
The only equilibrium objects the section consumes are the no-pickoff floor and the committed premium $s^*$ of Convention-A makers, pinned by one free-entry equation in~\S\ref{subsec:3o_info}.
Throughout, LMSR depth is taken as strictly positive and privately optimal---the three-outcome counterpart of Assumption~\ref{ass:lmsr_recovery_boundary}, under which the LP's recovery from tail demand covers its adverse-selection cost at some interior $L$---so that the pool is present to absorb Stage-1 flow.

\subsection{State Space, Information, and Committed Switch Quotes}
\label{subsec:3o_info}

Let $X_1,X_2,X_3$ be Arrow--Debreu claims with uniform prior.
The symmetric shock favors one outcome and splits its mass evenly across the other two: the post-shock value of the favored outcome is $\mu_F^{(3)}=\tfrac13+\sigma$ and each of the two unfavored outcomes is worth $\mu_U^{(3)}=\tfrac13-\tfrac{\sigma}{2}$, with $\mu_F^{(3)}+2\mu_U^{(3)}=1$; the superscript distinguishes these from the binary values of~\S\ref{sec:market_info}.
This is the direct three-outcome version of the binary posterior move: one coordinate rises by $\sigma$, the other two fall equally, and no realization ever separates the two unfavored outcomes from each other.

\paragraph{Switches and Committed Quotes.}
With three or more outcomes, the most atomic directional view a trader can hold is ``outcome $k$ is more likely than outcome $j$.''
The trade that expresses it is a switch $k\to j$: buy $X_k$ and sell $X_j$.
If the CLOB posts asks and bids $(A_i,B_i)$ on the three books, define the ask and bid markups relative to a benchmark belief $r$ by $a_{i,r}\equiv A_i-r_i$ and $b_{i,r}\equiv r_i-B_i$.
The price of executing a switch, and its premium over the benchmark gap, are
\begin{equation}
\label{eq:switch_prices_and_premium}
S^{kj}=A_k-B_j,\qquad
s^{kj}(r)=S^{kj}-(r_k-r_j)=a_{k,r}+b_{j,r}.
\end{equation}
When the benchmark is clear, we suppress the subscript $r$.

The book cannot choose the six switch prices independently: it posts only three asks and three bids, and every switch price is one ask plus one bid.
The following lemma states exactly how much freedom this leaves.

\begin{lemma}[Separability of committed switch premia]
\label{lem:separability}
Fix a benchmark belief $r$.
A premium matrix $(s^{kj})_{k\ne j}$ is algebraically implementable by primitive quote legs if and only if there exist real markups $(a_i,b_i)_{i=1}^3$ such that $s^{kj}=a_k+b_j$ for every $k\ne j$.
For three outcomes, additive separability is equivalent to the cycle identity
$$
s^{12}+s^{23}+s^{31}
=
s^{21}+s^{32}+s^{13}.
$$
Equivalently, the executable switch prices satisfy
$$
S^{12}+S^{23}+S^{31}
=
S^{21}+S^{32}+S^{13}.
$$
Whenever these conditions hold, the primitive levels
$$
A_i=r_i+a_i,
\qquad
B_i=r_i-b_i
$$
reproduce all six switch prices algebraically.
They form a feasible committed CLOB quote vector if and only if the decomposition admits a common gauge shift $c$ satisfying $0 \le r_i-b_i+c \le r_i+a_i+c \le 1$ for every $i$.
Thus separability characterizes algebraic implementation, while primitive quote feasibility is an additional restriction.
\end{lemma}

Lemma~\ref{lem:separability} is a leg-consistency result rather than a complete quote-feasibility or no-arbitrage condition.
The cycle identity determines whether the six switch prices can arise from one collection of primitive ask and bid legs.
The additional inequalities in the lemma determine whether that algebraic representation can be placed inside the feasible CLOB quote region.
The symmetric quotes constructed below satisfy both restrictions automatically.
See~\S\ref{App:proof_for_separability} for the proof.

\paragraph{Stage-1 Landmarks.}
Stage~1 works as in the binary model and all $c$-traders buy in the favored direction until the fee-adjusted price reach public posterior; we record two facts that the rest of the section uses repeatedly.
By the no-pickoff logic of Lemma~\ref{lem:no_pickoff}, committed quotes must survive every realization of the shock, so informed $c$-traders find nothing profitable resting on the book and trade the favored claim against the LMSR, filling it until its fee-adjusted favored price reaches $\mu_F^{(3)}$.

From a uniform start the LMSR's favored price after displacement $\beta$ is $p_F(\beta)=e^{\beta/L}/(e^{\beta/L}+2)$, so the absorbed flow is $\beta^*(L)=L\lambda_F^{(3)}$ with
$$
\lambda_F^{(3)}\;=\;\ln\frac{2(1-\tau)\,\mu_F^{(3)}}{\,1-(1-\tau)\,\mu_F^{(3)}\,},
$$
strictly positive if and only if $\tau<\tfrac{3\sigma}{1+3\sigma}$.
This is the three-outcome version of~\eqref{eq:c_capacity} and of Assumption~\ref{ass:low_fee}(i).
Since the Stage-1 flow trades only the favored coordinate, so the two unfavored inventories---and hence the two unfavored raw prices---remain exactly equal in every realization.

\paragraph{Which directions attract tail demand.}
$p$-traders have no information and see only prices, so they can form views only about comparisons that prices have made visible.
After Stage~1, prices separate the favored outcome from each unfavored one by a gap of $\mu_F^{(3)}-\mu_U^{(3)}=\tfrac32\sigma$, while the two unfavored claims still trade at identical prices on the pool.
The zero tail-demand for the unfavored pair comply with the no information, no tail-demand formulation in~\S\ref{sec:market_info}. 

For the active direction after the common shock, we only assume that tail-demand exist on the direction of unfavored to favored realization.
The Stage-2 primitive is therefore: each of the six ordered directions hosts a continuum of mass $\tfrac12 z^{(3)}$ in expectation, active only in the realization whose price move points its way, so exactly two directions are live per realization, favored to unfavored into each unfavored claim.

Two clarifications keep the choice in proportion.
First, it restricts who shows up, not what is tradable: the book quotes all six directions in every realization regardless, since posting any ask and any bid prices their combination the ladder below is a fact about quotes and the primitive decides only where volume sits.
Second, the demand primitive is not what generates this section's three-outcome results.
The gauge, the belief bound, and the collateral comparison of \S\ref{subsec:3o_beliefs}--\S\ref{subsec:3o_balance} are statements about the committed quote vector, and their hypotheses contain no demand object: $z^{(3)}$ and $H^{(3)}$ enter only through the free-entry condition that prices the commitment.
What changes at $n\ge3$ is the commitment technology, and that change is in place before any trader arrives.

\paragraph{Valuations on a Live Pair.}
A switch pays $+1$ if $k$ occurs, $-1$ if $j$ occurs, and $0$ otherwise, so its fair value is the probability gap $\mu_k-\mu_j$---and the pair's total probability is not $1$ but $M\equiv\mu_F^{(3)}+\mu_U^{(3)}=\tfrac23+\tfrac{\sigma}{2}$, the same constant in every realization because the shock treats outcomes symmetrically.
We handle this in the most direct way.
The live trader takes the third outcome's probability as displayed---prices gave her no reason to hold a view on it---which fixes the pair's mass at $M$ and makes her private view a single number: a conditional probability $\pi\in[0,1]$ that the favored leg wins, given that one of the pair does.
Her valuation of the switch is then $m=M(2\pi-1)$, which ranges over $[-M,M]$ and is centered at the realized gap $2\mu_F^{(3)}-M=\tfrac32\sigma$.
We impose clauses (i)--(ii) of Assumption~\ref{ass:tail_demand_distribution} on the law of $m$---unimodal, log-concave, $C^1$ density, peaked at $\tfrac32\sigma$, the same law in every realization---and write the single tail
$$
H^{(3)}(y)\;\equiv\;\Pr\!\big(m\ \ge\ \tfrac32\sigma+y\ \big|\ \sigma_c\big),
\qquad y\ge0,
$$
with support ceiling $\bar y^{(3)}\le M-\tfrac32\sigma=\tfrac23-\sigma$.
This matches the binary ceiling $\bar y_F\le\tfrac12-\sigma$, with the pair mass $M$ in place of total mass $1$.
The map from $\pi$ to $m$ is affine, so it makes no difference whether the shape conditions are stated on $\pi$ or on $m$; we use gap units throughout because the quotes and the survival floors of the quoting problem below are all in those units.

\begin{assumption}[Three-outcome viability]
\label{ass:low_fee_3o}
The primitives satisfy
(i) $\tau<\dfrac{3\sigma}{1+3\sigma}$;~(ii) $0<\kappa<\tfrac13\,z^{(3)}\sup_{s\ge0}sH^{(3)}(s)$;
(iii) the selected root $s^*$ of~\eqref{eq:lp_free_entry_3o} satisfies $s^*<\bar y^{(3)}$ and $s^*+\tfrac32\sigma\le\tfrac23$.
\end{assumption}

Here (i) keeps the $c$-trader capacity positive, (ii) says the posting cost is low enough that entry is profitable at some premium, and (iii) keeps the quotes inside $[0,1]$ and the traded markup inside the support of demand---the mirror of Assumption~\ref{ass:low_fee}.

\paragraph{The Maker's Quoting Problem.}
Two requirements decide the maker's quotes.
First, no quote may sit inside the range of possible posteriors.
Second, she must recover the posting cost $\kappa$ from tail-demand flow, just as in the binary model.
The first requirement puts a floor under each leg; the second pins the total premium; and nothing else about the six numbers affects her profit, for the following reason.
On the equilibrium path nothing executes against a single leg: $c$-traders find no surplus at surviving quotes, and $p$-traders trade pairs.
So profit depends on the quote vector only through the switch prices.
We accordingly let a maker enter by posting the two-sided quotes of one book at cost $\kappa$; the full six-quote vector is three books and rescales the cost to $3\kappa$, changing nothing.
This accounting presumes both legs of a switch execute against quotes of the same posted collection; \S\ref{subsec:3o_beliefs} shows what routing across makers does and does not discipline.

Because the shock treats the three outcomes symmetrically and every live direction faces the same demand law at the same realized gap, competitive makers that have to quote all six directions drives all live directions to a common markup.
Write the pair's committed premium as $a+b=s+\tfrac32\sigma$, so that $s$ is the markup a live direction pays over its realized gap.
All expected traded premium should cover the posting cost of all three books, thus the free entry conditions gives:
\begin{equation}
\label{eq:lp_free_entry_3o}
\Phi^{(3)}(s)\;\equiv\;6 \; \frac{1}{3}\; \frac{z^{(3)}}{2}\,s\,H^{(3)}(s)\;= z^{(3)}\,s\,H^{(3)}(s) =\;3\kappa .
\end{equation}
In words, out of six possible trading pairs, each pair earn $\frac{z^{(3)}}{2}\,s\,H^{(3)}(s)$ expected profit and together they need to cover $3\kappa$ in total in expectation.
At $s=0$ revenue is zero, therefore the market makers always rest on choosing the smallest root $s^*$ that covers quoting cost $\kappa$. 

\paragraph{One Premium, Three Realized Markups.}
Free entry pins $a^*+b^*=s^*+\tfrac32\sigma$, and the book posts this same number on all six directions.
What differs across directions is the gap each one has in the realized state: $\tfrac32\sigma$ for favored$\to$unfavored, $0$ for the unfavored pair, $-\tfrac32\sigma$ for unfavored$\to$favored.
Subtracting, the realized markups form the ladder
$$
\underbrace{\,s^*\,}_{F\to U}\;<\;\underbrace{\,s^*+\tfrac32\sigma\,}_{U\leftrightarrow U}\;<\;\underbrace{\,s^*+3\sigma\,}_{U\to F},
$$
and only the lower rung trades for simplicity.
The $\tfrac32\sigma$ built into the committed premium plays the role the floor $2\sigma$ played in the binary unfavored markup $s^*+2\sigma$.
The difference between what the book posts and what it trades is itself observable: posted premia take three values, transactions concentrate on one.
The premium $s^*$ and its ladder are all this section takes from equilibrium; see the Figure~\ref{fig:realized_markups} for a illustration of the three realized markup.

\begin{figure}[H]
    \centering
    \includegraphics[width=0.8\linewidth]{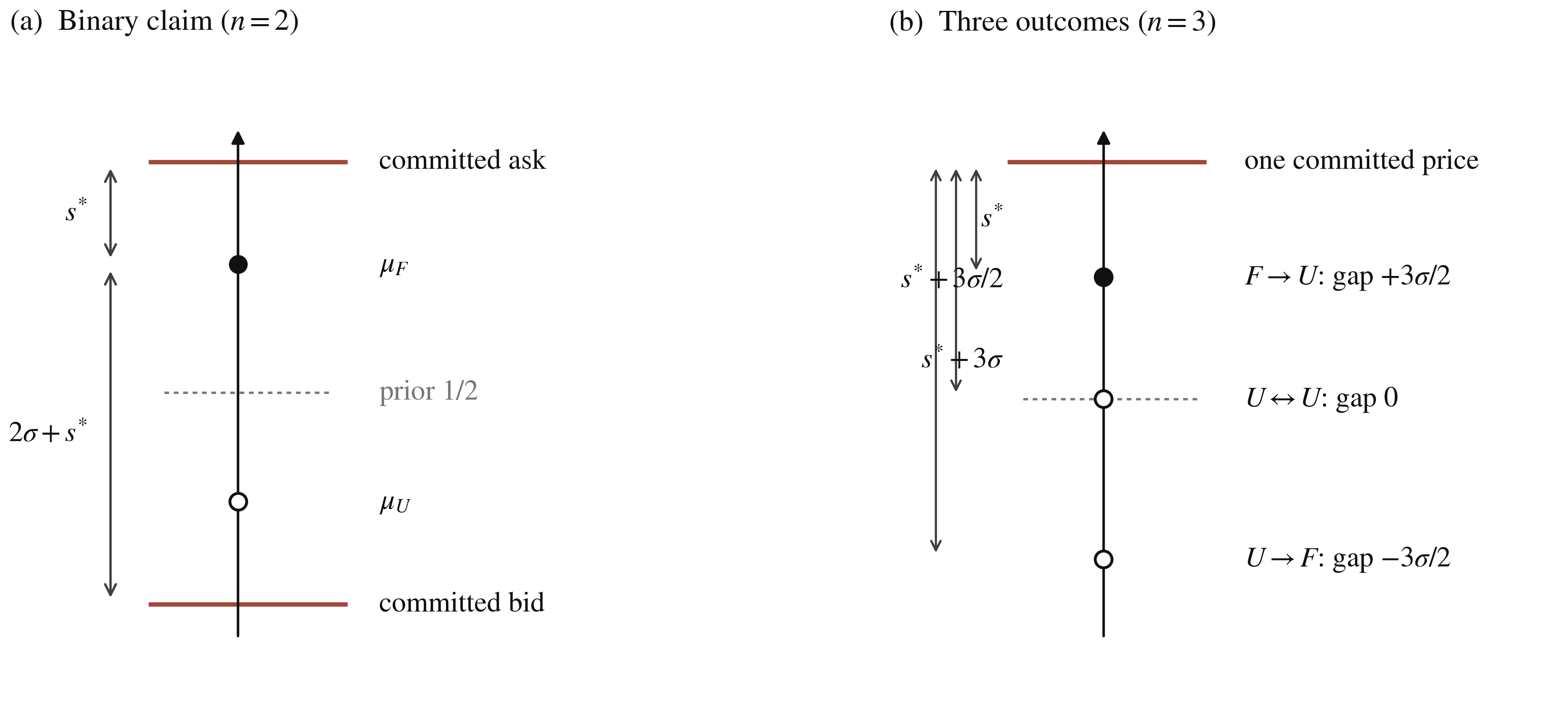}
    \caption{Realized Markup: Binary v.s.~Three-outcome}
    \label{fig:realized_markups}
\end{figure}

\subsection{Beliefs: the Unpriced Coordinate}
\label{subsec:3o_beliefs}

A prediction market's product is the probability itself, so the first question for any mechanism is whether its quotes identify one.
For the committed book, trading splits the vector in two: every transaction is a switch, a switch price is a difference of two quotes, so trades price the differences and never touch the common level---and the level is exactly where the probability lives.
This subsection follows that unpriced coordinate through three results: it is unidentified (Lemma~\ref{lem:information_defect}), stale wherever identified (Corollary~\ref{cor:belief_error}), and---once makers may compete---pinned by entry in its premium but never in its level (Theorem~\ref{thm:no_quoting_equilibrium}).
Whether anything backs that level is a balance-sheet question, and \S\ref{subsec:3o_balance} takes it up.

\begin{lemma}[Leg gauge and the indeterminacy of CLOB-implied beliefs]
\label{lem:information_defect}
Fix a benchmark belief $r$ and suppress the subscript.
Two committed quote vectors produce the same six executable switch prices---hence the same premia against every benchmark---if and only if they differ by a uniform translation of all quotes, $A_i\mapsto A_i+c$, $B_i\mapsto B_i+c$; in markups,
$$
a_i\mapsto a_i+c,\qquad b_i\mapsto b_i-c,
$$
for any $c$ such that every shifted quote remains feasible.
The switch-price matrix therefore identifies the committed quotes only up to this one-parameter family.
\end{lemma}

The lemma says that when makers quote based on their belief of the atomic preference of traders for $n\ge3$ outcomes, then the displayed belief is indeterminate.
The lemma applies whenever $n\ge3$, because no individual claim is the complement of every other claim.
Under the free-entry quotes of \S\ref{subsec:3o_info}, the translations respecting the survival floors and the quote bounds form a nondegenerate interval.
Its endpoints are determined jointly by the ask floor, the bid floor, and the constraints $0\le B_i\le A_i\le1$.
See~\S\ref{App:proof_for_information_defect} for the proof.

The unpriced quote coordinate creates a second information problem that, any probability vector extracted before the shock must miss at least one realized posterior by the full shock size, we show it using the corollary below.

\begin{corollary}[Deterministic switch distortion and shock-sized belief error]
\label{cor:belief_error}
Consider the symmetric committed free-entry quotes of \S\ref{subsec:3o_info}.
In any realization, let $k$ denote the favored outcome and let $i$ and $j$ denote the two unfavored outcomes.

\smallskip
\noindent (i) \emph{(Deterministic unfavored-pair distortion.)}
The two unfavored outcomes have identical posteriors,
$$
\mu_i(\sigma_c) = \mu_j(\sigma_c) = \mu_U^{(3)},
$$
and hence $\mu_i(\sigma_c)-\mu_j(\sigma_c)=0$.
Nevertheless, their executable switch prices satisfy $S^{ij}=S^{ji}=a^*+b^*=s^*+\frac32\sigma$.
Consequently,
$$
S^{ij}
-
\big(
\mu_i(\sigma_c)-\mu_j(\sigma_c)
\big)
=
S^{ji}
-
\big(
\mu_j(\sigma_c)-\mu_i(\sigma_c)
\big)
=
s^*+\frac32\sigma
$$
in every realization.

\smallskip
\noindent (ii) \emph{(Coordinate-level belief error.)}
Let $r=(r_1,r_2,r_3)$ be any probability vector extracted from the committed quotes by a $t=0$-measurable normalization rule.
Then
$$
\Pr\!\left(
\max_{\ell}
\big|
r_\ell-\mu_\ell(\sigma_c)
\big|
\ge
\sigma
\right)
\ge
\frac13.
$$
Thus any fixed normalization of the unidentified quote level incurs a coordinate error of at least $\sigma$ in at least one of the three equally likely shock realizations.

\smallskip
\noindent (iii) \emph{(LMSR comparison.)}
After Stage-1 $c$-flow, the two unfavored LMSR inventories remain equal, so their raw marginal prices satisfy
$$
p_i\big(\beta^*(L)\big)
-
p_j\big(\beta^*(L)\big)
=
0
=
\mu_i(\sigma_c)-\mu_j(\sigma_c).
$$
The favored raw marginal price and fee-adjusted purchase price satisfy
$$
p_k\big(\beta^*(L)\big)= (1-\tau)\mu_F^{(3)} \qquad
\frac{p_k\big(\beta^*(L)\big)}{1-\tau} =
\mu_F^{(3)}.
$$
Hence the LMSR responds to the realized state and preserves the zero gap between the two unfavored outcomes, whereas the committed CLOB charges a switch premium of $s^*+\tfrac32\sigma$ between them.
\end{corollary}

The two conclusions concern different dimensions of the committed book.
Part~(i) is an equilibrium statement about executable switch prices: because every outcome must be protected against both becoming favored and becoming unfavored, the book charges $s^*+\tfrac32\sigma$ to switch between the two outcomes whose realized probability gap is zero.
Part~(ii) concerns the unidentified common level: once a normalization converts the quote vector into a probability vector, that fixed normalization must miss some favored realization by at least $\sigma$.
Thus the committed CLOB is stale both in its executable pairwise prices and in any normalized claim-level belief extracted from its unpriced coordinate.

In the no-regret reading of cost-function market makers \citep{chen2010new}, the LMSR performs a state-contingent log-loss update against order flow while the committed book uses a predictor chosen before the shock.
Corollary~\ref{cor:belief_error} makes the resulting worst-realization error exact.
See~\S\ref{App:proof_for_belief_error} for the proof.

\paragraph{Contestable Entry Stability.}
Any number of risk-neutral makers may post two-sided outcome books at cost $\kappa$ per book, and one-sided posting is not allowed.
Each switch leg routes to the best posted quote, and ties split evenly.
Call a complete quote configuration \emph{contestably entry-stable} if at least one two-sided book is posted for every outcome, every active book weakly covers its posting cost, and no inactive maker can earn strictly positive expected profit by posting any finite collection of two-sided outcome books against the standing quotes.
This is the contestable-book concept of Maker Convention~I: it characterizes the standing configuration after all profitable entry and undercutting opportunities have been exhausted, rather than a simultaneous sunk-cost quote game.
We restrict attention to symmetric standing configurations in which the best active asks share one margin above the ask floor and the best active bids share one margin below the bid floor.
Write these winning margins as $m_a$ and $m_b$, so the live-switch premium is
$$
s_{\mathrm{eff}}=m_a+m_b.
$$

\begin{theorem}[Competition pins the premium, not the common level]
\label{thm:no_quoting_equilibrium}
Under leg-by-leg routing, book-level posting cost $\kappa$, Maker Convention~I, and Assumption~\ref{ass:low_fee_3o}, the symmetric contestably entry-stable configurations are exactly those with one active book per outcome and a feasible common margin split satisfying
$$
m_a\ge0,
\qquad
m_b\ge0,
\qquad
m_a+m_b=s^*.
$$
The corresponding primitive quotes are
$$
A_i=\frac13+\sigma+m_a,
\qquad
B_i=\frac13-\frac\sigma2-m_b.
$$

\smallskip
\noindent (i) \emph{(Premium selection.)}
Every active book earns zero expected profit, and takers pay the common executable premium
$$
s_{\mathrm{eff}}=s^*,
$$
where $s^*$ is the smallest increasing-branch root of~\eqref{eq:lp_free_entry_3o}.

\smallskip
\noindent (ii) \emph{(Common-gauge invariance.)}
For any $c$ such that the translated quotes remain feasible and pickoff-proof, applying
$$
A_i^{(m)}\mapsto A_i^{(m)}+c,
\qquad
B_i^{(m)}\mapsto B_i^{(m)}+c
$$
to every quote of every active maker leaves every executable switch price, execution volume, and expected book profit unchanged.
Equivalently, the common translation sends
$$
m_a\mapsto m_a+c,
\qquad
m_b\mapsto m_b-c
$$
while preserving $m_a+m_b=s^*$.

\smallskip
\noindent (iii) \emph{(Level indeterminacy.)}
The admissible common translations form a nondegenerate interval determined by quote feasibility.
Consequently, competition identifies the executable premium $s^*$ but identifies the primitive quote vector only up to the common gauge of Lemma~\ref{lem:information_defect}, so the displayed claim-level belief is not pinned by the competitive conditions.
A translation applied to only a proper subset of the three active books is not invariant and can change package prices, routing, and profits.
\end{theorem}

Theorem~\ref{thm:no_quoting_equilibrium} characterizes contestable free-entry outcomes within the symmetric shock class.
It does not claim that every asymmetric simultaneous quote-posting game has the same equilibrium set.
Premium selection and level invariance are separate results.
Complete-book entry eliminates every common premium below or above $s^*$, while a common translation of all active quotes preserves $s^*$, all switch prices, all routing decisions, and every book's expected profit.
The deviation proof is direction-by-direction because an entrant may improve only an ask, only a bid, or both legs of a switch, but ended up not being able to cover the quoting cost.
Competition therefore disciplines relative quote locations $s^*$ but leaves one global quote coordinate, the gauge, free.
Figure~\ref{fig:free_gauge} is a illustration of how competitive makers quote same switch premium in equilibria but in each equilibrium, the displayed belief is different by a mid-point rule of bid~/~ask prices. 

\begin{figure}[H]
    \centering
    \includegraphics[width=0.9\linewidth]{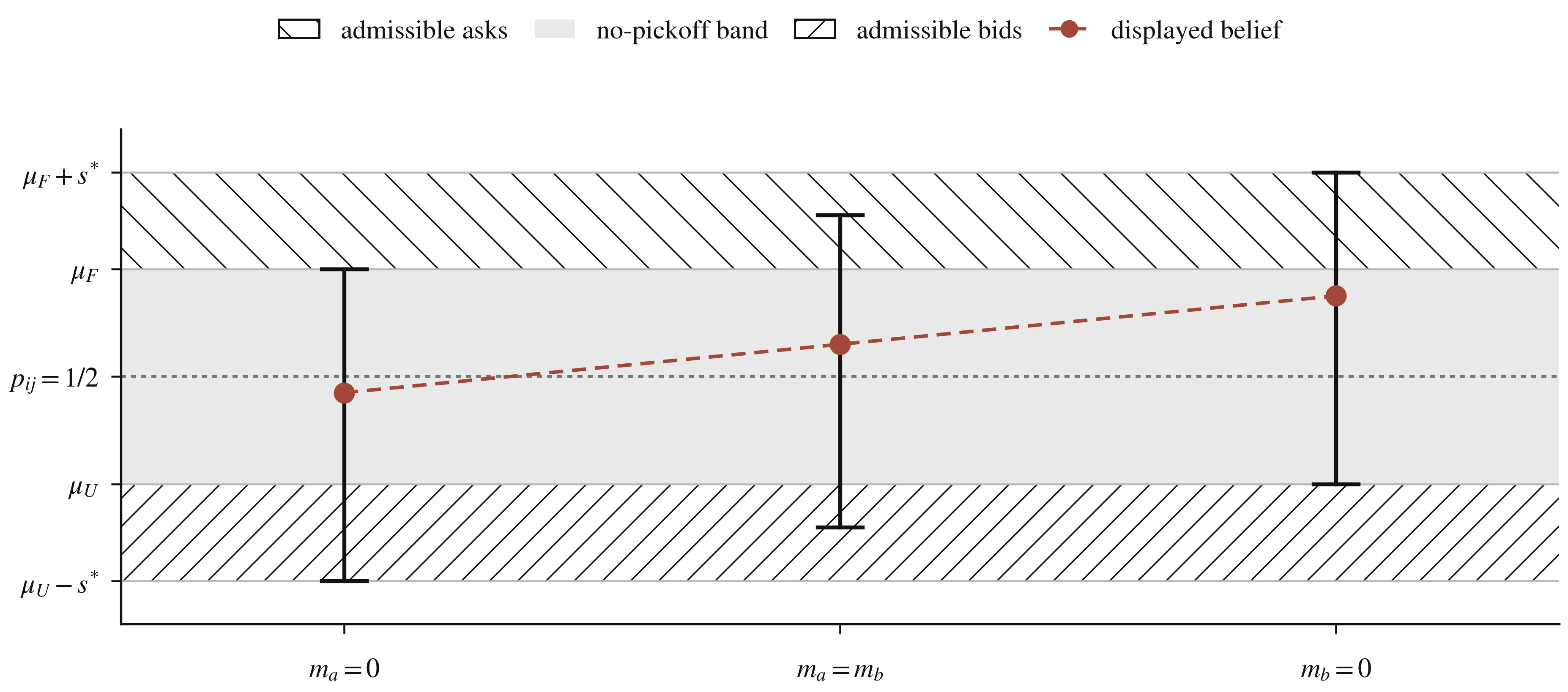}
    \caption{Equilibrium Quotes with Competitive Makers}
    \label{fig:free_gauge}
\end{figure}

Two of the three defects are now established: the level is unpriced by trades (Lemma~\ref{lem:information_defect}) and undetermined by competition (Theorem~\ref{thm:no_quoting_equilibrium}).
Neither says what a maker must put at risk to move it, and that is the question the balance sheet answers.

\subsection{Balance Sheets: Backing the Display, Funding the Event}
\label{subsec:3o_balance}

Two collateral questions separate the mechanisms, and the outcome space sharpens both.
The first is what a maker must put at risk to move the belief it displays; the second is how much capital the event as a whole requires as $n$ grows.
One definition makes the first precise.

\begin{definition}[Collateralized display]
\label{def:collateralized}
A mechanism's displayed belief is \emph{collateralized} if it cannot change without some participant acquiring a position whose settlement payoff differs across states.
It is \emph{uncollateralized} along a family of moves if the display changes along that family while every participant's settlement payoff in every state stays the same.
\end{definition}

\begin{proposition}[Self-collateralized beliefs]
\label{prop:self_collateral}
Consider the $n$-outcome LMSR with liquidity $L$, cost
$$
C_L(q)
=
L\log\sum_i e^{q_i/L},
$$
and displayed belief
$$
p(q)
=
\nabla C_L(q).
$$
Consider also an admissible committed CLOB quote collection $(A_i,B_i)$ whose Stage-2 orders are atomic switches.

\smallskip
\noindent (i) \emph{(Invariant directions.)}
For the LMSR,
$$
p(q')=p(q)
\quad\Longleftrightarrow\quad
q'-q
\in
\operatorname{span}(\mathbf 1).
$$
An inventory translation $q\mapsto q+c\mathbf 1$ is a complete-basket trade that changes the LMSR cost by $c$, pays $c$ in every state, and therefore has zero net state-contingent payoff.
Modulo this translation, $q\mapsto p(q)$ is a bijection onto the interior of the probability simplex.

For the CLOB, Lemma~\ref{lem:information_defect} implies that the admissible repostings preserving every executable switch price are exactly the common gauges
$$
A_i\mapsto A_i+c,
\qquad
B_i\mapsto B_i+c.
$$
An admissible gauge reposting executes no order and changes no position, yet moves every primitive displayed level by $c$.

\smallskip
\noindent (ii) \emph{(Moving the LMSR display is a bet at scale $L$.)}
For a fee-free cost-function trade from $q$ to $q'$ with displayed beliefs $p$ and $p'$, the gross settlement payoff in state $i$ is
$$
(q'_i-q_i)
-
\big(C_L(q')-C_L(q)\big)
=
L\log\frac{p'_i}{p_i}.
$$
Under any belief $\pi$, the expected gross payoff is $L \left[ \mathrm{KL}(\pi\|p) - \mathrm{KL}(\pi\|p')\right]$
where $\mathrm{KL}(\pi\|p) = \sum_i \pi_i \log\frac{\pi_i}{p_i}$.
Under the paper's proportional fee convention, a purchase satisfying $C_L(q')-C_L(q)\ge0$ instead delivers
$$
(q'_i-q_i)
-
\frac{C_L(q')-C_L(q)}{1-\tau}
=
L\log\frac{p'_i}{p_i}
-
\frac{\tau}{1-\tau}
\big(C_L(q')-C_L(q)\big).
$$
The fee term is state independent, so every movement $p'\ne p$ still creates a position whose payoff differs across states.

\smallskip
\noindent (iii) \emph{(The CLOB gauge is uncollateralized.)}
Let
$$
(A'_i,B'_i)
=
(A_i+c,B_i+c)
$$
be any common translate that remains feasible and pickoff-proof.
Suppose each switch $k\to j$ executes atomically as a purchase of $X_k$ and a sale of $X_j$ against the same committed quote collection.
Then
$$
A'_k-B'_j
=
(A_k+c)-(B_j+c)
=
A_k-B_j.
$$
Hence the original and translated books induce identical switch participation, execution quantities, net switch payments, and settlement positions in every realization, while every primitive displayed level differs by $c$.
With multiple active makers, the same invariance applies to a common translation of every active maker, whereas a relative translation of only one maker need not preserve routing or package prices.
\end{proposition}

Each part of Proposition~\ref{prop:self_collateral} has a direct economic interpretation.
The LMSR ignores only a complete-basket trade, which changes its cash account and settlement liability by the same amount in every state.
Every other LMSR trade changes the displayed probability and creates state-dependent exposure.
Part~(ii) gives the gross logarithmic-scoring-rule payoff, while the proportional taker fee subtracts a state-independent charge and therefore does not alter the collateralization result.
For the committed book, an admissible gauge reposting moves every primitive quote level without changing an atomic switch price, net package payment, or settlement position.
The qualification to a common translation is essential once several makers are active, because translating only one maker can change leg rankings and routing.
In the language of Definition~\ref{def:collateralized}, the LMSR display is collateralized---it moves only through a state-contingent bet---while the committed CLOB display is uncollateralized along its common gauge.
Part~(ii) is the classical identity between LMSR trades and Hanson's logarithmic scoring rule \citep{hanson2003combinatorial}; the contribution here is its comparison with the uncollateralized CLOB gauge.
See~\S\ref{App:proof_self_collateral} for the proof.

The three defects of the committed level are now in hand: it is unpriced by trades (Lemma~\ref{lem:information_defect}), undetermined by competition (Theorem~\ref{thm:no_quoting_equilibrium}), and unbacked by bets (Proposition~\ref{prop:self_collateral})---three statements of one fact, that a book which quotes legs cannot own the probability it displays.
Owning a display is one thing and funding it is another.
A maker must post collateral to guarantee settlement whatever its display means, and the requirement scales with the outcome space.

\begin{proposition}[Event-level collateral efficiency]
\label{prop:event_level_collateral_efficiency}
For an $n$-outcome event with $n>2$ initialized at the uniform prior, consider the fee-free normalized LMSR cost
$$
C_L(q)
=
L\log
\left(
\frac1n
\sum_{i=1}^n
e^{q_i/L}
\right).
$$
Its worst-case loss is the supremum
$$
K_n^{LMSR}(L) \equiv \sup_{q\in\mathbb R^n,\,i} \left\{q_i-\big( C_L(q)-C_L(0) \big) \right\} = L\log n.
$$
Suppose instead that the CLOB operates $n$ outcome books, each book must collateralize one-sided settlement exposure $\vartheta$, and the platform permits no event-level cross-book netting.
Then
$$
K_n^{CLOB}(\vartheta)
=
n\vartheta.
$$
The exact collateral ratio is therefore
$$
\frac{K_n^{CLOB}(\vartheta)}
{K_n^{LMSR}(L)}
=
\frac{\vartheta}{L}
\frac{n}{\log n}.
$$
For any family of comparisons in which the depth normalization satisfies $0<\underline c \le \frac{\vartheta}{L} \le \overline c <\infty$
uniformly in $n$, the collateral ratio satisfies
$$
\frac{K_n^{CLOB}}
{K_n^{LMSR}}
=
\Theta
\left(
\frac{n}{\log n}
\right).
$$
The LMSR equality is the fee-free cost-function benchmark.
If every trading fee is nonnegative and retained by the LMSR LP, fees weakly reduce realized loss, so $L\log n$ remains a conservative collateral bound.
\end{proposition}

The comparison is conditional on book-level collateral segregation and comparable depth normalization.
It does not claim that event-level CLOB cross-margining is impossible, nor does it derive the ratio $\vartheta/L$ from equilibrium primitives.
Rather, it isolates the scaling consequence of the two architectures: one LMSR potential nets settlement exposure across the event, while independently collateralized CLOB books impose one solvency constraint per outcome.
See~\S\ref{App:proof_of_event_level_coll} for the proof.

\paragraph{The uniqueness pincer.}
the committed vector fails at beliefs and collateral where the committed potential succeeds; a classical characterization says this is no accident.
Any market maker that quotes a full probability vector path-independently and admits no arbitrage is a convex-potential cost-function maker \citep{chen2007betting,abernethy2013efficient}; Lemma~\ref{lem:information_defect} supplies the complementary impossibility, that a committed quote \emph{vector} cannot identify the probability it purports to sell.
Between the two, the cost-function AMM is not merely a better way to sell a probability over $n\ge3$ outcomes---up to the axioms, it is the only committed mechanism that sells one at all.

These are not free lunches.
The AMM purchases identification, tight uninformative pricing, and logarithmic collateral with the adverse-selection payment $\ell_c L$ that~\eqref{eq:lp_foc} prices; the book's refusal to pay is what forces screening, and screening in dimension $n\ge3$ externalizes its cost into the gauge, the indeterminate belief display, and the collateral bill.

\section{Conclusion and Open Directions}
\label{sec:conclusion}

We asked what a maker does in a prediction market without classical noise, and whether a cost-function AMM can coexist with the order books real venues run.
The answer is one mechanism whose consequences escalate. 
A prediction-market share has no reference price, so the informed move first
and any pre-shock CLOB quote inside the band the signal will move into is
picked off; the surviving spread screens informed flow off the book, and the
maker becomes a rent extractor carrying the under-priced side to resolution.
The venue hierarchy inverts: the AMM aggregates information while the CLOB
lives on behavioral rent.
Coexistence is not neutral between the two: the book's committed premium caps how far the pool's price can walk, while pool depth strips revenue from the book, so deeper liquidity widens or compresses the book's spread according to whether its quotes are contestable or committed, and past a threshold removes the book's clientele entirely.
With $n\ge3$ outcomes the book's failure moves from rent to product.
Its quotes price only by differences, so the level they display---the probability itself---is unidentified, wrong by at least the shock size wherever normalized, and left free by competition, which pins the premium and nothing else.
Nor is that level backed: a maker can relocate it with no order executed and no payoff changed.
The cost-function maker prices the whole simplex from one potential, reprices its display with every trade it absorbs, and backs the event with a factor $\Theta(n/\log n)$ less collateral.
Set beside the classical characterization of cost-function makers, this leaves the AMM as the only committed mechanism that sells a probability over $n\ge3$ outcomes at all.

Four limitations mark an opening: (i) the mechanism rests on a single pre-shock quote, so a dynamic microfoundation in which makers requote as the volatility regime evolves is an important open question; (ii) AMM depth enters only through a scalar $L$, leaving concentrated, Uniswap-v3-style~\cite{adams2021uniswap} bucketed depth a natural next step; (iii) our results are positive throughout---we characterize which configurations survive, not which one a platform should choose; and (iv) the calibration rests on one
short-horizon contract type, so whether the same results survive in other markets remains an empirical question. 
A venue without uninformative flow either screens information out and lives on rent, or aggregates it at a bounded loss, and may sit better on two coexisting mechanisms than one. 
The tension is not peculiar to prediction markets: any asset whose only value is the information it reveals forces its makers to choose between screening that information out and paying to absorb it. 
Prediction markets, whose product is the probability itself, simply make the choice unavoidable---and make it worth asking, as we do here, when a venue is better served by running both mechanisms than by choosing one.

\newpage
\bibliographystyle{plainnat}
\bibliography{references}

\newpage
\appendix

\section{Proofs}

\subsection{Proof of Lemma~\ref{lem:no_pickoff}}
\label{App:proof_for_lem_no_pickoff}
\begin{proof}
Write $v(\sigma_c)\equiv\mu_Y(\sigma_c)=p_0+\sigma_c$ for the post-shock value of the YES claim.
The two possible values are $p_0+\sigma$ and $p_0-\sigma$, each with probability $\tfrac12$.
A maker commits a two-sided YES quote $(b,a)$ with $b\le a$ before the shock is realized.

\smallskip
\noindent\emph{Step 1 (quote-level characterization).}
After observing the shock, a $c$-trader's maximum surplus from the quote is
$$
\max\{v(\sigma_c)-a,\ b-v(\sigma_c),\ 0\}.
$$
The quote is pickoff-proof under a given shock sign if and only if $b\le v(\sigma_c)\le a$.
Imposing this condition under both signs gives
$$
a\ge p_0+\sigma,
\qquad
b\le p_0-\sigma.
$$
If the ask inequality fails, the positive-shock $c$-mass buys YES and earns strictly positive surplus $p_0+\sigma-a$ and all $p$-traders will also buy from CLOB until the quote is exhausted if the active direction is also positive, and the bid inequality can fail in the same way.
Conversely, if both inequalities hold, neither shock sign gives a $c$-trader and same direction $p$-trader a strictly profitable CLOB trade.
This proves the quote-level characterization.

\smallskip
\noindent\emph{Step 2 (active-leg representation).}
Write the realized active-leg quote as
$$
A^D=\mu_D+s_D,
\qquad
B^C=\mu_C-s_D,
\qquad
A^D+B^C=1.
$$
If $D=F$, the executable active-leg quote is the ask $A^F=a=\mu_F+s_F$.
The ask condition from Step~1 is therefore equivalent to $s_F\ge0$.
If $D=U$, the executable active-leg quote is the complement of the bid, so $A^U=1-b=\mu_U+s_U$.
The bid condition from Step~1 is therefore equivalent to
$$
A^U\ge1-\mu_U=\mu_F,
$$
which gives $s_U\ge\mu_F-\mu_U=2\sigma$ under $p_0=\tfrac12$.
Both cases can be written as
$$
s_D\ge\underline s_D\equiv\mu_F-\mu_D.
$$
The equivalent carry-leg condition is $B^C=1-A^D\le\mu_U$.

\smallskip
\noindent\emph{Step 3 (positive tail demand).}
Conditional on $D=F$, the CLOB serves mass $zH_F(s_F)$, which is strictly positive if and only if $s_F<\bar y_F$.
Conditional on $D=U$, the CLOB serves mass $zH_U(s_U)$, which is strictly positive if and only if $s_U<\bar y_U$.
Intersecting these support conditions with the no-pickoff floor gives
$$
s_D\in[\underline s_D,\bar s_D),
\qquad
\bar s_F=\bar y_F,
\qquad
\bar s_U=\bar y_U.
$$
If the lower endpoint is weakly above the support ceiling, the corresponding screening band is empty.
This proves the lemma.
\end{proof}

\subsection{Proof of Lemma~\ref{lem:marginal_and_average_cost}}
\label{app:proof_for_marginal_and_average_cost}
\begin{proof}
The binary LMSR cost function is
$$
C_L(q_Y,q_N)
=
L\log\left(e^{q_Y/L}+e^{q_N/L}\right).
$$
The LMSR price of claim $X\in\{Y,N\}$ is the partial derivative
$$
p_X(q_Y,q_N)
=
\frac{\partial C_L(q_Y,q_N)}{\partial q_X}
=
\frac{e^{q_X/L}}{e^{q_Y/L}+e^{q_N/L}}.
$$
To initialize the market at YES price $p_0$, choose the initial inventory
so that
$$
\frac{e^{q_Y^0/L}}{e^{q_Y^0/L}+e^{q_N^0/L}}=p_0.
$$
Since only inventory differences matter, write $q=q_Y-q_N-(q_Y^0-q_N^0)$ for net YES inventory relative to the initial state. 
Then the YES and NO marginal price after net YES inventory $q$ is
$$
p_Y(q;p_0,L)
=
\frac{p_0e^{q/L}}{p_0e^{q/L}+1-p_0} \qquad p_N(q;p_0,L)=1-p_Y(q;p_0,L).
$$
Under the uniform prior $p_0=\tfrac12$, this reduces to $p_Y(q;1/2,L)=\frac{1}{1+e^{-q/L}}$.
For a finite purchase of $\Delta>0$ YES shares from net YES inventory $q$, the raw cost for such transaction is
$$
T_Y(\Delta;q,p_0,L)
=
C_L(q+\Delta)-C_L(q)
=
L\log\left(
\frac{p_0e^{(q+\Delta)/L}+1-p_0}
     {p_0e^{q/L}+1-p_0}
\right).
$$
For a finite purchase of $\Delta>0$ NO shares, net YES inventory falls
by $\Delta$, so the raw transaction cost is
$$
T_N(\Delta;q,p_0,L)
=
L\log\left(
\frac{p_0e^{q/L}+(1-p_0)e^{\Delta/L}}
     {p_0e^{q/L}+1-p_0}
\right) = \int_0^\Delta p_X(q+\iota_X u;p_0,L)\,du,
$$
where $\iota_Y=+1$ and $\iota_N=-1$. 
With proportional taker fee $\tau$, the fee-adjusted total purchase cost is
$$
\widetilde T_X(\Delta;q,p_0,L,\tau)
=
\frac{T_X(\Delta;q,p_0,L)}{1-\tau}.
$$
The corresponding fee-adjusted average purchase price is
$$
\overline P_X(\Delta;q,p_0,L,\tau)
=
\frac{\widetilde T_X(\Delta;q,p_0,L,\tau)}{\Delta}.
$$
Taking the limit as $\Delta\downarrow0$ gives the fee-adjusted marginal purchase price
$$
P_X(q;p_0,L,\tau)
=
\frac{p_X(q;p_0,L)}{1-\tau}.
$$
Finally, differentiating the YES marginal price gives
$$
\frac{\partial p_Y(q;p_0,L)}{\partial q}
=
\frac{p_Y(q;p_0,L)(1-p_Y(q;p_0,L))}{L}.
$$
Thus larger $L$ lowers local price impact for any fixed price level.
This proves the lemma.
\end{proof}

\subsection{Proof of Proposition~\ref{prop:active_leg_existence}}
\label{app:proof_for_active_leg_existence}
\begin{proof}
Fix $L\in(0,\bar L)$.
Because exactly one tail direction is active after Stage~1, the favored and unfavored clearing problems below are conditional continuation problems and are never imposed on the same terminal LMSR state.
The proof follows the recursive order of capacity, premium, and conditional clearing.

\smallskip
\noindent\emph{Step 1 (common-flow capacity).}
Since $\bar L=1/\lambda_F$, the capacity condition~\eqref{eq:c_capacity} binds at
$$
\beta_F^*(L)=L\lambda_F\in(0,1).
$$
The post-$c$-flow inventory is $\iota_F L\lambda_F$.
The fee-adjusted favored price is therefore $\mu_F$, while the fee-adjusted unfavored purchase price is $\mu_U+w_\tau$.
Assumption~\ref{ass:low_fee}(i) gives $\lambda_F>0$ and $w_\tau<2\sigma$.

\smallskip
\noindent\emph{Step 2 (free-entry premium).}
The expected screening revenue of one committed book is
$$
\Phi_D(s)=\tfrac12 z\big[sH_F(s)+(2\sigma+s)H_U(2\sigma+s)\big],
$$
where the second term is zero whenever $2\sigma+s\ge\bar y_U$.
The free-entry condition is $\Phi_D(s)=\kappa$.
Assumption~\ref{ass:tail_demand_distribution}(v) makes $\Phi_D$ continuous and strictly quasi-concave on $[0,\bar y_F)$, with unique peak $s^\ddagger$.
Assumption~\ref{ass:low_fee}(ii) gives
$$
\Phi_D(0)<\kappa<\Phi_D(s^\ddagger).
$$
The intermediate value theorem therefore gives a root on the increasing branch, and strict quasi-concavity makes that root unique on the branch.
Write this root as
$$
s^{FE}\in(0,s^\ddagger).
$$
Under maker Convention~I, any premium above the increasing-branch root can be profitably undercut until the root is reached, while any further cut below the root yields expected revenue below $\kappa$.
Hence the unique surviving premium is $s^*=s^{FE}$.
Assumption~\ref{ass:low_fee}(iii) implies $s^*<s^\ddagger<\bar y_F$, so favored-direction CLOB demand is strictly positive.
Since $\Phi_D$ contains no $L$, the selected premium satisfies
$$
\frac{ds^*}{dL}=0.
$$

\smallskip
\noindent\emph{Step 3 (favored-direction clearing).}
Given $(s^*,L)$, define
$$
G_F(y)=z\big[H_F(y)-H_F(s^*)\big]-L\big[\lambda(\mu_F+y)-\lambda_F\big].
$$
The function $G_F$ is continuous and strictly decreasing on $[0,s^*]$ because $H_F'<0$ and $\lambda'>0$.
At the lower endpoint,
$$
G_F(0)=z\big[H_F(0)-H_F(s^*)\big]>0.
$$
At the upper endpoint,
$$
G_F(s^*)=-L\big[\lambda(\mu_F+s^*)-\lambda_F\big]<0.
$$
Hence there is a unique root $y_F(L)\in(0,s^*)$.
The captured favored-direction mass is
$$
M_F(L,s^*)=z\big[H_F(y_F(L))-H_F(s^*)\big],
$$
and the corresponding routing share is
$$
\alpha_F^*(L)=z(1-H_F(s^*))\in(0,1).
$$

\smallskip
\noindent\emph{Step 4 (unfavored-direction clearing).}
Define
$$
G_U(y)=z\big[H_U(y)-H_U(2\sigma+s^*)\big]-L\big[\lambda(\mu_U+y)+\lambda_F\big],
$$
where $H_U(2\sigma+s^*)$ is zero whenever $2\sigma+s^*\ge\bar y_U$.
The identity
$$
\lambda(\mu_U+w_\tau)=-\lambda_F
$$
follows from $p_0=\tfrac12$ and $(1-\tau)(\mu_U+w_\tau)=1-(1-\tau)\mu_F$.
Therefore,
$$
G_U(w_\tau)=z\big[H_U(w_\tau)-H_U(2\sigma+s^*)\big]>0,
$$
because $w_\tau<2\sigma+s^*$ and $w_\tau<\bar y_U$.
If $2\sigma+s^*<\bar y_U$, then
$$
G_U(2\sigma+s^*)=-L\big[\lambda(\mu_U+2\sigma+s^*)+\lambda_F\big]<0.
$$
Strict monotonicity therefore gives a unique root $y_U(L)\in(w_\tau,2\sigma+s^*)$.
If $2\sigma+s^*\ge\bar y_U$, then $H_U(2\sigma+s^*)=H_U(\bar y_U)=0$ and
$$
G_U(\bar y_U)=-L\big[\lambda(\mu_U+\bar y_U)+\lambda_F\big]<0.
$$
Strict monotonicity therefore gives a unique root $y_U(L)\in(w_\tau,\bar y_U)$.
On the first branch, the CLOB retains mass $zH_U(2\sigma+s^*)>0$ and
$$
\alpha_U^*(L)=1-\frac{H_U(2\sigma+s^*)}{H_U(y_U(L))}\in(0,1).
$$
On the second branch, the CLOB retains no unfavored mass and every unfavored trader who trades uses the LMSR, so
$$
\alpha_U^*(L)=1.
$$

\smallskip
\noindent\emph{Step 5 (comparative statics and zero-depth limits).}
Implicit differentiation of the favored clearing equation gives
$$
y_F'(L)=\frac{\lambda(\mu_F+y_F)-\lambda_F}{zH_F'(y_F)-L\lambda'(\mu_F+y_F)}<0.
$$
Implicit differentiation of the unfavored clearing equation gives
$$
y_U'(L)=\frac{\lambda(\mu_U+y_U)+\lambda_F}{zH_U'(y_U)-L\lambda'(\mu_U+y_U)}<0.
$$
Since the terminal displacements fall with $L$, the captured masses increase with $L$.
On every branch with positive CLOB participation, the formulas for $\alpha_F^*$ and $\alpha_U^*$ then imply that the routing shares are strictly increasing in $L$.
As $L\downarrow0$, the favored clearing equation implies
$$
y_F(L)\uparrow s^*.
$$
Indeed, any smaller limit would leave the mass difference $H_F(y_F)-H_F(s^*)$ strictly positive while the capacity term converges to zero.
Moreover,
$$
M_F(L,s^*)=L\big[\lambda(\mu_F+y_F(L))-\lambda_F\big]=O(L),
$$
while $H_F(y_F(L))\to H_F(s^*)>0$.
Hence $\alpha_F^*(L)=O(L)$.
If $2\sigma+s^*<\bar y_U$, the unfavored clearing equation similarly implies
$$
y_U(L)\uparrow2\sigma+s^*,
$$
because any smaller limit would leave a strictly positive unfavored mass difference while the capacity term converges to zero.
Moreover,
$$
M_U(L,s^*)=L\big[\lambda(\mu_U+y_U(L))+\lambda_F\big]=O(L),
$$
and $H_U(y_U(L))\to H_U(2\sigma+s^*)>0$.
Hence $\alpha_U^*(L)=O(L)$ on this branch.
If $2\sigma+s^*\ge\bar y_U$, the unfavored clearing equation instead implies
$$
y_U(L)\uparrow\bar y_U.
$$
Indeed, any smaller limit would leave $H_U(y_U)$ strictly positive while the capacity term converges to zero.
On this branch $\alpha_U^*(L)=1$ for every $L$, while
$$
zH_U(y_U(L))=L\big[\lambda(\mu_U+y_U(L))+\lambda_F\big]=O(L).
$$

\smallskip
\noindent\emph{Step 6 (screening and uniqueness).}
The selected premium satisfies $s^*>0$.
The favored-direction markup is therefore strictly above its no-pickoff floor $0$, and the unfavored-direction markup $2\sigma+s^*$ is strictly above its no-pickoff floor $2\sigma$.
Lemma~\ref{lem:no_pickoff} then implies that every $c$-trader obtains strictly negative surplus from every executable CLOB trade.
Each block has a unique solution, so the state-contingent continuation policy is unique.
This proves all claims.
\end{proof}

\subsection{Proof of Proposition~\ref{prop:interior_optimal_L}}
\label{App:proof_of_interior_L}
\begin{proof}
By mirror symmetry, the sign of the common shock does not affect payoff conditional on the active direction.
The favored and unfavored active-direction states occur with probability $\tfrac12$ each, and exactly one of them occurs after the Stage-1 common-flow trade.

\smallskip
\noindent\emph{Step 1 (conditional recovery accounting).}
The LMSR finite-trade cost is homogeneous in trade size, inventory, and liquidity:
$$
\widetilde T_X(\theta L;\xi L,p_0,L,\tau)
=L\widetilde T_X(\theta;\xi,p_0,1,\tau).
$$
The Stage-1 common-flow term is $-\ell_c L$, with $\ell_c$ as defined before~\eqref{eq:lp_payoff}.
Conditional on $D=F$, the fee-adjusted LMSR price rises from $\mu_F$ to $\mu_F+y_F(L)$.
Revenue net of the favored posterior liability is therefore
$$
L\int_{\lambda_F}^{\lambda(\mu_F+y_F(L))}\big(\lambda^{-1}(u)-\mu_F\big)\,du
=L\int_0^{y_F(L)}y\,d\lambda(\mu_F+y).
$$
Conditional on $D=U$, the fee-adjusted LMSR price rises from $\mu_U+w_\tau$ to $\mu_U+y_U(L)$.
Revenue net of the unfavored posterior liability is therefore
$$
L\int_{w_\tau}^{y_U(L)}y\,d\lambda(\mu_U+y).
$$
Weighting the two mutually exclusive active-direction states by $\tfrac12$ and adding the common-flow loss gives~\eqref{eq:lp_payoff}.

\smallskip
\noindent\emph{Step 2 (marginal recovery and its zero-depth limit).}
By Proposition~\ref{prop:active_leg_existence}, the conditional terminal displacements are continuously differentiable on $\mathcal I$.
Differentiating the recovery part of~\eqref{eq:lp_payoff} gives
$$
\rho(L)
=\frac12\int_0^{y_F(L)}y\,d\lambda(\mu_F+y)
+\frac12\int_{w_\tau}^{y_U(L)}y\,d\lambda(\mu_U+y)
+\frac{L}{2}y_F(L)\lambda'(\mu_F+y_F(L))y_F'(L)
+\frac{L}{2}y_U(L)\lambda'(\mu_U+y_U(L))y_U'(L).
$$
Proposition~\ref{prop:active_leg_existence} gives
$$
y_F(L)\uparrow s^*
$$
as $L\downarrow0$.
It also gives
$$
y_U(L)\uparrow\min\{2\sigma+s^*,\bar y_U\}
$$
as $L\downarrow0$.
When either limiting endpoint lies strictly inside its tail support, the corresponding derivative remains bounded, so its terminal-compression term multiplied by $L$ converges to zero.
It remains to verify the unfavored terminal-compression term when $2\sigma+s^*\ge\bar y_U$.
On that branch, the clearing equation is
$$
zH_U(y_U(L))=L\big[\lambda(\mu_U+y_U(L))+\lambda_F\big].
$$
Combining this equation with its derivative gives
$$
L y_U'(L)
=\frac{H_U(y_U(L))}{H_U'(y_U(L))-H_U(y_U(L))\dfrac{\lambda'(\mu_U+y_U(L))}{\lambda(\mu_U+y_U(L))+\lambda_F}}.
$$
The denominator is negative and has absolute value at least $-H_U'(y_U(L))$.
Because the tail density $-H_U'$ decreases toward the support endpoint,
$$
0\le\frac{H_U(y)}{-H_U'(y)}\le\bar y_U-y.
$$
Consequently,
$$
\big|L y_U'(L)\big|
\le\frac{H_U(y_U(L))}{-H_U'(y_U(L))}
\le\bar y_U-y_U(L)
\longrightarrow0.
$$
Thus every terminal-compression term vanishes as $L\downarrow0$, and
$$
\rho(0^+)
=\frac12\int_0^{s^*}y\,d\lambda(\mu_F+y)
+\frac12\int_{w_\tau}^{\min\{2\sigma+s^*,\bar y_U\}}y\,d\lambda(\mu_U+y).
$$
Assumption~\ref{ass:lmsr_recovery_boundary}(i) therefore gives
$$
\rho(0^+)>\ell_c.
$$

\smallskip
\noindent\emph{Step 3 (unique interior optimum).}
The payoff derivative is
$$
\Pi^{LMSR}{}'(L)=-\ell_c+\rho(L).
$$
Assumption~\ref{ass:lmsr_recovery_boundary}(ii) makes $\rho$ strictly decreasing on $\mathcal I$ and gives
$$
\lim_{L\uparrow\bar L}\rho(L)<\ell_c.
$$
Since $\rho(0^+)>\ell_c$, continuity implies that $\rho$ crosses $\ell_c$ exactly once at some $L^*\in\mathcal I$.
Hence
$$
\Pi^{LMSR}{}'(L)>0
\quad\text{for }L<L^*,
\qquad
\Pi^{LMSR}{}'(L)<0
\quad\text{for }L>L^*.
$$
Therefore $L^*$ is the unique maximizer on $\mathcal I$ and satisfies~\eqref{eq:lp_foc}.
Because $\Pi^{LMSR}(L)\to0$ as $L\downarrow0$ and the payoff initially rises, $\Pi^{LMSR}(L^*)>0$.
Assumption~\ref{ass:lmsr_recovery_boundary}(iii) implies that no choice outside the branch yields a larger payoff.

\smallskip
\noindent\emph{Step 4 (routing and recovery at the optimum).}
Since $L^*\in(0,\bar L)$, Proposition~\ref{prop:active_leg_existence} gives
$$
\beta_F^*(L^*)\in(0,1),
\qquad
\alpha_F^*(L^*)\in(0,1).
$$
If $2\sigma+s^*<\bar y_U$, the same proposition gives
$$
\alpha_U^*(L^*)\in(0,1).
$$
If $2\sigma+s^*\ge\bar y_U$, it instead gives
$$
\alpha_U^*(L^*)=1.
$$
The favored terminal satisfies $y_F(L^*)>0$, so the favored conditional recovery integral is strictly positive.
The unfavored terminal satisfies $y_U(L^*)>w_\tau$, so the unfavored conditional recovery integral is strictly positive.
Every unfavored-direction execution occurs at displacement at least $w_\tau$ above the unfavored posterior.
This proves the proposition.
\end{proof}

\subsection{Proof of Lemma~\ref{lem:queue_priority}}
\label{App:proof_for_queue_priority}
\begin{proof}
Fix the favored-direction state.
After Stage~1, the favored LMSR price is $\mu_F$, and the CLOB quote is $\mu_F+s$.
The mass that can execute at the CLOB is $zH_F(s)$.
By the log-odds capacity formula, the number of LMSR units available below the CLOB quote is
$$
L\big[\lambda(\mu_F+s)-\lambda_F\big].
$$
Convention~B assigns those units first to CLOB-eligible traders.
If their mass does not exceed this capacity, every CLOB-eligible trader receives an LMSR unit, the CLOB residual is zero, and lower-valuation traders continue filling the LMSR until the terminal displacement $y$ satisfies
$$
z-L\big[\lambda(\mu_F+s)-\lambda_F\big],
$$
which gives $V_F(s,L)$ after imposing nonnegativity.
If their mass on LMSR does not exceed this capacity, every $p$-trader receives an LMSR unit and the CLOB residual is zero, and the terminal displacement $y$ satisfies
$$
zH_F(y)=L\big[\lambda(\mu_F+y)-\lambda_F\big].
$$
The left side is strictly decreasing in $y$, the right side is strictly increasing, the left side exceeds the right side at $y=0$, and the residual condition makes the left side weakly below the right side at $y=s$.
Hence the solution is unique in $(0,s]$.

Fix next the unfavored-direction state.
After Stage~1, the fee-adjusted unfavored LMSR price is $\mu_U+w_\tau$, while the CLOB quote is $\mu_U+2\sigma+s$.
The mass that can execute at the CLOB is $zH_U(2\sigma+s)$.
The log-odds identity $\lambda(\mu_U+w_\tau)=-\lambda_F$ makes the number of LMSR units available below the CLOB quote
$$
L\big[\lambda(\mu_U+2\sigma+s)+\lambda_F\big].
$$
Convention~B again assigns those units first to CLOB-eligible traders.
If their mass exceeds this capacity, the LMSR reaches the CLOB quote and the remaining mass is $V_U(s,L)$.
If their mass does not exceed this capacity, the CLOB residual is zero and lower-valuation traders continue filling the LMSR until
$$
zH_U(y)=L\big[\lambda(\mu_U+y)+\lambda_F\big].
$$
The left side is strictly decreasing, the right side is strictly increasing, the right side is zero at $y=w_\tau$, and the zero-residual condition reverses the inequality by $y=2\sigma+s$.
Hence the solution is unique in $(w_\tau,2\sigma+s]$.

Dividing the LMSR-served CLOB-eligible mass by the corresponding eligible mass gives the two routing-share formulas.
The direction draw occurs after Stage~1, so these calculations are conditional and the probability $\tfrac12$ enters only when the two conditional revenues are averaged.
This proves the lemma.
\end{proof}

\subsection{Proof of Proposition~\ref{prop:binary_spillover}}
\label{App:proof_for_prop_binary_spillover}
\begin{proof}
Write
$$
\Lambda_F(s)\equiv\lambda(\mu_F+s)-\lambda_F,
\qquad
\Lambda_U(s)\equiv\lambda(\mu_U+2\sigma+s)+\lambda_F=\lambda(\mu_F+s)+\lambda_F.
$$
Both functions are strictly increasing, with $\Lambda_F(0)=0$ and $\Lambda_U(s)>0$.
The two uncapped residuals are strictly decreasing in $s$ because
$$
zH_F'(s)-L\Lambda_F'(s)<0,
\qquad
zH_U'(2\sigma+s)-L\Lambda_U'(s)<0.
$$
Thus each positive residual occupies an interval and can disappear only at a unique cutoff.

\smallskip
\noindent\emph{Step 1: business stealing on every smooth residual branch.}
When both residuals are positive,
$$
\Phi_D(s;L)=\tfrac12\Big[s\big(zH_F(s)-L\Lambda_F(s)\big)+(2\sigma+s)\big(zH_U(2\sigma+s)-L\Lambda_U(s)\big)\Big].
$$
Differentiating in depth gives
$$
\frac{\partial\Phi_D}{\partial L}=-\tfrac12\Big[s\Lambda_F(s)+(2\sigma+s)\Lambda_U(s)\Big]<0.
$$
Since $\Lambda_F'(s)=\Lambda_U'(s)=\lambda'(\mu_F+s)$,
$$
\frac{\partial^2\Phi_D}{\partial s\,\partial L}=-\Big[\lambda(\mu_F+s)+(\sigma+s)\lambda'(\mu_F+s)\Big]<0.
$$
If only the favored residual is positive, the corresponding derivatives are
$$
\frac{\partial\Phi_D}{\partial L}=-\tfrac12s\Lambda_F(s)<0,
\qquad
\frac{\partial^2\Phi_D}{\partial s\,\partial L}=-\tfrac12\Big[\Lambda_F(s)+s\Lambda_F'(s)\Big]<0.
$$
If only the unfavored residual is positive, they are
$$
\frac{\partial\Phi_D}{\partial L}=-\tfrac12(2\sigma+s)\Lambda_U(s)<0,
\qquad
\frac{\partial^2\Phi_D}{\partial s\,\partial L}=-\tfrac12\Big[\Lambda_U(s)+(2\sigma+s)\Lambda_U'(s)\Big]<0.
$$
Therefore depth strictly lowers maker revenue at every premium with positive revenue and tilts every smooth revenue branch toward lower premia.

\smallskip
\noindent\emph{Step 2: the entry margin.}
Let $L_2>L_1$ and let $s_2$ maximize $\Phi_D(\cdot;L_2)$ at a depth for which the peak value is positive.
Step~1 gives
$$
\Phi_D(s_2;L_1)>\Phi_D(s_2;L_2).
$$
Hence
$$
\sup_s\Phi_D(s;L_1)>\sup_s\Phi_D(s;L_2).
$$
The positive-part formula makes $\Phi_D$ jointly continuous in $(s,L)$, and the support can be closed at $\bar y_F$ by setting revenue equal to zero at the boundary.
The maximum theorem and the maintained uniqueness of the maximizer therefore make the peak value and $s^\ddagger(L)$ continuous.
It follows that the set of depths satisfying $\sup_s\Phi_D(s;L)\ge\kappa$ is an interval.
If its boundary $L^{det}$ lies below $\bar L$, continuity gives
$$
\sup_s\Phi_D(s;L^{det})=\kappa.
$$

\smallskip
\noindent\emph{Step 3: Convention~I.}
Assumption~\ref{ass:low_fee}(ii) and Step~1 imply
$$
\Phi_D(0;L)\le\Phi_D(0;0)<\kappa
$$
for every $L$ in the capacity-binding region.
The depth-uniform shape restriction therefore gives a unique increasing-branch root whenever the peak exceeds $\kappa$.
Take two such depths $L_2>L_1$ and write their roots as $s_D^*(L_2)$ and $s_D^*(L_1)$.
At the lower-depth root, Step~1 gives
$$
\Phi_D(s_D^*(L_1);L_2)<\Phi_D(s_D^*(L_1);L_1)=\kappa.
$$
Because $\Phi_D(\cdot;L_2)$ is strictly increasing up to its unique peak, its increasing-branch root must satisfy
$$
s_D^*(L_2)>s_D^*(L_1).
$$
Joint continuity and uniqueness make the root continuous in $L$.
On a differentiable segment with an unchanged active residual set, the implicit function theorem gives
$$
\frac{ds_D^*}{dL}=-\frac{\partial\Phi_D/\partial L}{\partial\Phi_D/\partial s}>0,
$$
because the selected root lies on the increasing branch.
If $L^{det}<\bar L$, the unique root at the boundary where the peak equals $\kappa$ is the peak itself, so
$$
\lim_{L\uparrow L^{det}}s_D^*(L)=s^\ddagger(L^{det}).
$$
Whenever the favored CLOB-eligible mass is positive,
$$
\alpha_F^*(L)=\min \left(1, \frac{L\Lambda_F(s_D^*(L))}{zH_F(s_D^*(L))}\right).
$$
While this share is below one, logarithmic differentiation of the uncapped ratio gives $\frac{d\log\alpha_F^*}{dL}>0$.
Whenever the unfavored CLOB-eligible mass is positive,
$$
\alpha_U^*(L)=\min\left(1, \frac{L\Lambda_U(s_D^*(L))}{zH_U(2\sigma+s_D^*(L))}\right).
$$
While this share is below one, the analogous logarithmic derivative is $\frac{d\log\alpha_U^*}{dL}>0$.
Every bracket is positive because capacity increases and tail mass decreases with the premium.
Thus each uncapped ratio rises strictly, and each capped share rises until it reaches one and remains there.

\smallskip
\noindent\emph{Step 4: Convention~II.}
Under Convention~II, the selected premium is the unique maximizer $s^\ddagger(L)$.
On every smooth branch, Step~1 makes $\partial^2\Phi_D/\partial s\,\partial L<0$, so an increase in $L$ makes the derivative at the old maximizer negative and moves the unique new maximizer weakly to the left.
At a regular interior peak, the implicit function theorem gives
$$
\frac{ds^\ddagger}{dL}=-\frac{\partial^2\Phi_D/\partial s\,\partial L}{\partial^2\Phi_D/\partial s^2}<0.
$$
It remains to check that a residual cutoff cannot create a rightward jump.
A favored cutoff satisfies $zH_F(s)=L\Lambda_F(s)$ and therefore moves according to
$$
\frac{ds}{dL}=\frac{\Lambda_F(s)}{zH_F'(s)-L\Lambda_F'(s)}<0.
$$
An unfavored cutoff satisfies $zH_U(2\sigma+s)=L\Lambda_U(s)$ and therefore moves according to
$$
\frac{ds}{dL}=\frac{\Lambda_U(s)}{zH_U'(2\sigma+s)-L\Lambda_U'(s)}<0.
$$
Both kinks move left, while Step~2 makes the unique maximizer continuous.
Consequently, $s^\ddagger(L)$ is globally nonincreasing through changes in the active residual set and strictly decreasing on every regular smooth branch.

Whenever the corresponding CLOB-eligible mass is positive, the logarithmic derivative of the uncapped routing ratio has the same form as in Step~3, but now $ds^\ddagger/dL<0$.
The elasticity condition in Proposition~\ref{prop:binary_spillover}(ii) makes the positive direct-depth term $1/L$ dominate the negative premium-compression term.
Hence each capped routing share is nondecreasing and rises strictly while it is below one.

\smallskip
\noindent\emph{Step 5: common markup movement and the handoff to Proposition~\ref{prop:crowdout_deterrence}.}
The favored- and unfavored-direction CLOB markups are $s_D^*(L)$ and $2\sigma+s_D^*(L)$ under either maker convention.
Their derivatives are therefore identical and their difference remains $2\sigma$.
Step~2 supplies the monotone entry margin used in Proposition~\ref{prop:crowdout_deterrence}(ii), while Lemma~\ref{lem:queue_priority} supplies the side-specific zero-residual thresholds used in part~(i) of that proposition.
This proves the proposition.
\end{proof}

\subsection{Proof of Proposition~\ref{prop:crowdout_deterrence}}
\label{App:proof_for_crowdout_deterrence}

\begin{proof}
\noindent\emph{Step 1: fixed-premium crowd-out thresholds.}
Fix $s>0$ with $H_F(s)>0$.
Lemma~\ref{lem:queue_priority} gives
$$
V_F(s,L)
=
\Big(
z
-
L\big[\lambda(\mu_F+s)-\lambda_F\big]
\Big)_+,
\qquad
V_U(s,L)
=
\Big(
z
-
L\big[\lambda(\mu_U+2\sigma+s)+\lambda_F\big]
\Big)_+.
$$
Because $s>0$ and $\lambda$ is strictly increasing,
$$
\lambda(\mu_F+s)-\lambda_F>0.
$$
Because $\mu_U+2\sigma=\mu_F$,
$$
\lambda(\mu_U+2\sigma+s)+\lambda_F
=
\lambda(\mu_F+s)+\lambda_F
>
\lambda(\mu_F+s)-\lambda_F.
$$
Hence $V_F(s,L)=0$ if and only if $L\ge L_F^{co}(s)$, while $V_U(s,L)=0$ if and only if $L\ge L_U^{co}(s)$.
Assumption~\ref{ass:tail_demand_distribution}(iii) gives $H_U(2\sigma+s)\le H_F(s)$ on the stated coverage range.
Using $\lambda_F>0$ and $H_F(s)>0$ therefore yields
$$
L_U^{co}(s)
=
\frac{zH_U(2\sigma+s)}
{\lambda(\mu_F+s)+\lambda_F}
\le
\frac{zH_F(s)}
{\lambda(\mu_F+s)+\lambda_F}
<
\frac{zH_F(s)}
{\lambda(\mu_F+s)-\lambda_F}
=
L_F^{co}(s).
$$
At an equilibrium pair satisfying
$$
L_U^{co}(s_D^*(L))
\le
L
<
L_F^{co}(s_D^*(L)),
$$
the first threshold equivalence gives $V_U(s_D^*(L),L)=0$, while the second gives $V_F(s_D^*(L),L)>0$.
Lemma~\ref{lem:queue_priority} then implies that the unfavored-direction state clears entirely on the LMSR and that the favored-direction state retains positive CLOB volume.

\smallskip
\noindent\emph{Step 2: the entry boundary and no-CLOB clearing.}
Proposition~\ref{prop:binary_spillover}(iii) makes
$$
L
\longmapsto
\sup_s\Phi_D(s;L)
$$
continuous and strictly decreasing whenever it is positive.
Assumption~\ref{ass:low_fee}(ii) gives
$$
\sup_s\Phi_D(s;0)>\kappa.
$$
If the set in the definition of $L^{det}$ is nonempty, continuity and strict monotonicity imply
$$
\sup_s\Phi_D(s;L)
\begin{cases}
>\kappa, & 0\le L<L^{det},\\
=\kappa, & L=L^{det},\\
<\kappa, & L^{det}<L<\bar L.
\end{cases}
$$
A book operates if and only if its peak revenue weakly exceeds $\kappa$, so no maker enters for $L\in(L^{det},\bar L)$.
For this proof only, write $\widehat y_F(L)$ and $\widehat y_U(L)$ for the unique solutions to
$$
zH_F\big(\widehat y_F(L)\big)
=
L
\Big[
\lambda\big(\mu_F+\widehat y_F(L)\big)
-
\lambda_F
\Big],
$$
and
$$
zH_U\big(\widehat y_U(L)\big)
=
L
\Big[
\lambda\big(\mu_U+\widehat y_U(L)\big)
+
\lambda_F
\Big].
$$
The favored equation has a unique solution in $(0,\bar y_F)$ because its left side is strictly decreasing from a positive value to zero and its right side is strictly increasing from zero.
The unfavored equation has a unique solution in $(w_\tau,\bar y_U)$ because Assumption~\ref{ass:low_fee}(v) gives $H_U(w_\tau)>0$, the log-odds identity gives $\lambda(\mu_U+w_\tau)+\lambda_F=0$, and the two sides again move in opposite directions.
The derivatives of the two clearing differences with respect to their displacement arguments are strictly negative, so the implicit function theorem makes both solutions continuous in $L$.
When no maker enters, all active tail demand faces the LMSR, and these two equations are exactly the unconstrained favored- and unfavored-direction clearing conditions from Lemma~\ref{lem:queue_priority}.
Thus both active-direction states clear entirely on the LMSR for every $L\in(L^{det},\bar L)$.

\smallskip
\noindent\emph{Step 3: the payoff jump when the favored CLOB ceiling binds at the boundary.}
For this proof only, let $\Pi^C(L)$ and $\Pi^M(L)$ denote the LP's induced coexistence and no-CLOB payoffs under Convention~B.
Lemma~\ref{lem:queue_priority} implies that the coexistence terminal displacement is the smaller of the unconstrained displacement and the corresponding CLOB markup.
Consequently,
$$
\begin{aligned}
\Pi^C(L)
={}&
-\ell_cL
+
\frac{L}{2}
\int_0^{\min\{\widehat y_F(L),s_D^*(L)\}}
y\,d\lambda(\mu_F+y)
\\
&+
\frac{L}{2}
\int_{w_\tau}^{\min\{\widehat y_U(L),2\sigma+s_D^*(L)\}}
y\,d\lambda(\mu_U+y),
\end{aligned}
$$
while the no-CLOB payoff is
$$
\Pi^M(L)
=
-\ell_cL
+
\frac{L}{2}
\int_0^{\widehat y_F(L)}
y\,d\lambda(\mu_F+y)
+
\frac{L}{2}
\int_{w_\tau}^{\widehat y_U(L)}
y\,d\lambda(\mu_U+y).
$$
Under maker Convention~I, Proposition~\ref{prop:binary_spillover}(i) gives
$$
\lim_{L\uparrow L^{det}}s_D^*(L)
=
s^\ddagger(L^{det}),
$$
and under maker Convention~II the same identity follows from $s_D^*(L)=s^\ddagger(L)$ and continuity.
Suppose that the unconstrained favored displacement at the boundary satisfies
$$
\widehat y_F(L^{det})
>
s^\ddagger(L^{det}).
$$
Continuity of the clearing roots and of the last coexistence premium then gives
$$
\begin{aligned}
&
\lim_{L\downarrow L^{det}}\Pi^M(L)
-
\lim_{L\uparrow L^{det}}\Pi^C(L)
\\
={}&
\frac{L^{det}}{2}
\int_{s^\ddagger(L^{det})}^{\widehat y_F(L^{det})}
y\,d\lambda(\mu_F+y)
\\
&+
\frac{L^{det}}{2}
\int_{\min\{\widehat y_U(L^{det}),2\sigma+s^\ddagger(L^{det})\}}^{\widehat y_U(L^{det})}
y\,d\lambda(\mu_U+y)
>
0.
\end{aligned}
$$
The first integral is strictly positive because its upper endpoint exceeds its lower endpoint and $\lambda$ is strictly increasing, while the second integral is nonnegative.
The common-flow loss is continuous and cancels from the two one-sided limits.
Therefore removal of the binding CLOB ceiling creates an upward jump in the LP's induced value at the entry boundary.

\smallskip
\noindent\emph{Step 4: optimal deterrence.}
Let a depth in $(L^{det},\bar L)$ attain the best monopoly payoff referred to in the proposition.
By Step~2, choosing that depth prevents maker entry.
If this payoff strictly exceeds the best coexistence payoff over depths below $L^{det}$, it also exceeds the coexistence value at $L^{det}$ because $\Pi^C(L)$ is continuous up to the boundary and that boundary value is the limit of payoffs from below.
It therefore dominates every depth at which a maker operates.
Hence the Stackelberg LP optimally chooses a no-entry depth and deters maker entry.
This proves the proposition.
\end{proof}

\subsection{Proof of Lemma~\ref{lem:separability}}
\label{App:proof_for_separability}

\begin{proof}
Suppose first that
$$
s^{kj}=a_k+b_j
$$
for every $k\ne j$.
Then
$$
\begin{aligned}
s^{12}+s^{23}+s^{31}
&=
(a_1+b_2)+(a_2+b_3)+(a_3+b_1)
\\
&=
(a_1+a_2+a_3)+(b_1+b_2+b_3)
\\
&=
(a_2+b_1)+(a_3+b_2)+(a_1+b_3)
\\
&=
s^{21}+s^{32}+s^{13}.
\end{aligned}
$$
Because the benchmark differences telescope around either directed cycle, the same identity holds for the executable prices:
$$
S^{12}+S^{23}+S^{31}
=
S^{21}+S^{32}+S^{13}.
$$

Conversely, suppose the premium cycle identity holds.
Choose the gauge normalization
$$
a_1=0
$$
and define
$$
\begin{aligned}
b_2&=s^{12},
&
b_3&=s^{13},
&
a_2&=s^{23}-s^{13},
\\
b_1&=s^{21}-s^{23}+s^{13},
&
a_3&=s^{31}-s^{21}+s^{23}-s^{13}.
\end{aligned}
$$
These definitions immediately give
$$
a_1+b_2=s^{12},
\qquad
a_1+b_3=s^{13},
\qquad
a_2+b_3=s^{23},
$$
and
$$
a_2+b_1=s^{21},
\qquad
a_3+b_1=s^{31}.
$$
The cycle identity can be rearranged as
$$
s^{32}
=
s^{12}+s^{23}+s^{31}-s^{21}-s^{13}.
$$
Therefore the remaining equation also holds:
$$
a_3+b_2=s^{32}.
$$
Hence the premium matrix is additively separable.

Setting
$$
A_i=r_i+a_i,
\qquad
B_i=r_i-b_i
$$
then gives
$$
A_k-B_j
=
(r_k-r_j)+a_k+b_j
=
(r_k-r_j)+s^{kj}
=
S^{kj}
$$
for every $k\ne j$.
A common gauge shift sends these primitive quotes to
$$
A_i+c,
\qquad
B_i+c
$$
without changing any switch price.
The algebraic implementation is a feasible committed CLOB quote vector exactly when some common shift satisfies
$$
0
\le
B_i+c
\le
A_i+c
\le
1
$$
for every $i$.
This proves both the separability characterization and the separate feasibility clause.
\end{proof}

\subsection{Proof of Lemma~\ref{lem:information_defect}}
\label{App:proof_for_information_defect}
\begin{proof}
Fix a benchmark belief $r$ and suppress the $r$ subscript. By
\eqref{eq:switch_prices_and_premium},
$$
s^{kj}=a_k+b_j
$$
for every ordered switch $k\to j$. If a real number $c$ shifts primitive
markups as
$$
a_i\mapsto a_i+c,\qquad b_i\mapsto b_i-c,
$$
and the shifted primitive quotes remain feasible, then
$$
(a_k+c)+(b_j-c)=a_k+b_j.
$$
Thus every ordered switch premium is unchanged. Since
$A_i=r_i+a_i$ and $B_i=r_i-b_i$, the same gauge shifts both primitive quote
levels in book $i$ by the same amount:
$$
A_i\mapsto A_i+c,\qquad B_i\mapsto B_i+c.
$$
Therefore each executable switch price is also unchanged:
$$
(A_k+c)-(B_j+c)=A_k-B_j.
$$
Conversely, consider any perturbation of primitive markups that preserves all ordered switch premia. Then for every $k\ne j$,
$$
\delta a_k+\delta b_j=0.
$$
Because there are three outcomes, for any two distinct origins one can choose a destination different from both. 
Comparing the two equations with that common destination gives the same value of $\delta a$ for the two origins. 
Hence all $\delta a_i$ are equal. 
Substituting back into $\delta a_k+\delta b_j=0$ gives all $\delta b_j$ equal to the negative of that common value. 
Therefore every premium-preserving perturbation is exactly
$$
\delta a_i=c,\qquad \delta b_i=-c.
$$
The kernel is one-dimensional, truncated only by primitive quote feasibility.

The economic implication follows immediately. 
The gauge leaves every ordered switch premium, and hence every realization's
premium $s+\delta_r$, participation threshold, routing share, and the pooled
free-entry value $\Phi_D$, unchanged. 
But it changes the common level of primitive asks and bids, for example every quote midpoint shifts by $c$. 
Thus the CLOB identifies relative switch costs but not a unique claim-level quote belief without an external normalization.
\end{proof}

\subsection{Proof of Corollary~\ref{cor:belief_error}}
\label{App:proof_for_belief_error}

\begin{proof}
\noindent\emph{Part (i).}
Fix a realization in which outcome $k$ is favored and outcomes $i$ and $j$ are unfavored.
By the symmetric shock,
$$
\mu_i(\sigma_c)
=
\mu_j(\sigma_c)
=
\frac13-\frac{\sigma}{2}.
$$
Therefore,
$$
\mu_i(\sigma_c)-\mu_j(\sigma_c)=0.
$$
The symmetric committed quote vector uses the same ask markup $a^*$ and bid markup $b^*$ on every outcome.
Hence,
$$
S^{ij}
=
A_i-B_j
=
a^*+b^*,
$$
and similarly,
$$
S^{ji}
=
a^*+b^*.
$$
The free-entry construction in \S\ref{subsec:3o_info} gives
$$
a^*+b^*
=
s^*+\frac32\sigma.
$$
It follows that
$$
S^{ij}
-
\big(
\mu_i(\sigma_c)-\mu_j(\sigma_c)
\big)
=
s^*+\frac32\sigma,
$$
and
$$
S^{ji}
-
\big(
\mu_j(\sigma_c)-\mu_i(\sigma_c)
\big)
=
s^*+\frac32\sigma.
$$
Because every realization contains exactly two unfavored outcomes, this distortion occurs with probability one.

\smallskip
\noindent\emph{Part (ii).}
Because $r$ is a probability vector, choose
$$
\ell^*
\in
\arg\min_{\ell}r_\ell.
$$
Then,
$$
r_{\ell^*}
\le
\frac13.
$$
In the realization in which outcome $\ell^*$ is favored,
$$
\mu_{\ell^*}(\sigma_c)
=
\frac13+\sigma.
$$
Therefore,
$$
\begin{aligned}
\big|
r_{\ell^*}
-
\mu_{\ell^*}(\sigma_c)
\big|
&\ge
\left(
\frac13+\sigma
\right)
-r_{\ell^*}
\\
&\ge
\sigma.
\end{aligned}
$$
Outcome $\ell^*$ is favored with probability $\tfrac13$, which proves
$$
\Pr\!\left(
\max_{\ell}
\big|
r_\ell-\mu_\ell(\sigma_c)
\big|
\ge
\sigma
\right)
\ge
\frac13.
$$

\smallskip
\noindent\emph{Part (iii).}
Stage-1 $c$-flow changes only the inventory of the favored outcome.
The two unfavored inventories therefore remain equal.
The LMSR softmax map then gives
$$
p_i\big(\beta^*(L)\big)
=
p_j\big(\beta^*(L)\big),
$$
and hence
$$
p_i\big(\beta^*(L)\big)
-
p_j\big(\beta^*(L)\big)
=
0.
$$
This equals their realized posterior gap.
The favored $c$-flow stops when its fee-adjusted purchase price reaches the favored posterior:
$$
\frac{
p_k\big(\beta^*(L)\big)
}{
1-\tau
}
=
\mu_F^{(3)}.
$$
Multiplying by $1-\tau$ gives
$$
p_k\big(\beta^*(L)\big)
=
(1-\tau)\mu_F^{(3)}.
$$
This proves the LMSR comparison.
\end{proof}

\subsection{Proof of Theorem~\ref{thm:no_quoting_equilibrium}}
\label{App:proof_no_quoting_eq}

\begin{proof}
\noindent\emph{Step 1: expected revenue of one symmetric book.}
Fix a common live-switch premium $s$ and a common split $m_a+m_b=s$.
Each ordered direction is active with probability $\tfrac13$ and then carries mass $\tfrac{z^{(3)}}2H^{(3)}(s)$, so its ex-ante mass is $\tfrac{z^{(3)}}6H^{(3)}(s)$.
A given book supplies the ask leg of two ordered directions and the bid leg of two ordered directions.
Its expected ask volume and expected bid volume are therefore both $\tfrac13z^{(3)}H^{(3)}(s)$.
Its expected revenue is
$$
\frac13z^{(3)}H^{(3)}(s)(m_a+m_b)
=
\frac13z^{(3)}sH^{(3)}(s)
=
\frac13\Phi^{(3)}(s).
$$
At $s=s^*$, equation~\eqref{eq:lp_free_entry_3o} gives
$$
\frac13\Phi^{(3)}(s^*)=\kappa.
$$
Thus one active book per outcome breaks even at every feasible common split of $s^*$.

\smallskip
\noindent\emph{Step 2: competition selects the smallest root.}
Log-concavity makes $sH^{(3)}(s)$ single-peaked, and the selected root $s^*$ lies on its strictly increasing branch.
Hence
$$
sH^{(3)}(s)<s^*H^{(3)}(s^*)
$$
for every $0\le s<s^*$.
If a symmetric standing configuration has $s<s^*$, even a book receiving all available flow earns less than $\kappa$, so the configuration is not entry-stable.
Now suppose that a symmetric standing configuration has $s>s^*$.
Choose $p\in(s^*,s)$ sufficiently close to $s^*$ that
$$
pH^{(3)}(p)>s^*H^{(3)}(s^*).
$$
An entrant can post one book for each outcome with margins
$$
m_a'=\frac{p}{s}m_a,
\qquad
m_b'=\frac{p}{s}m_b.
$$
These quotes remain feasible and pickoff-proof because both margins move weakly toward their no-pickoff floors.
Every positive-margin leg is strictly better than the corresponding standing leg, while a zero-margin leg merely ties and contributes no revenue.
The entrant therefore supplies every positive-margin component of every live switch at premium $p$ and earns total expected revenue
$$
\Phi^{(3)}(p)>\Phi^{(3)}(s^*)=3\kappa.
$$
Posting the three books costs $3\kappa$, so the deviation is strictly profitable.
No symmetric premium above $s^*$ is therefore entry-stable.
Together with the preceding argument, this proves
$$
s_{\mathrm{eff}}=s^*.
$$

\smallskip
\noindent\emph{Step 3: no book-level entrant can profit at $s^*$.}
Consider any entrant that posts a finite collection of two-sided outcome books against a standing feasible split $m_a+m_b=s^*$.
For any outcome, only the entrant's best ask and best bid can receive positive flow, so additional weaker books only add posting costs.
Consider first an ordered direction on which the entrant wins only the ask leg at margin $x<m_a$.
The direction then trades at premium $x+m_b<s^*$, and the entrant's unscaled ask revenue is $xH^{(3)}(x+m_b)$.
For $0\le x<m_a$,
$$
\frac{d}{dx}
\left[
xH^{(3)}(x+m_b)
\right]
=
H^{(3)}(x+m_b)
+
xH^{(3)\prime}(x+m_b).
$$
Because $H^{(3)\prime}\le0$ and $x\le x+m_b$,
$$
H^{(3)}(x+m_b)
+
xH^{(3)\prime}(x+m_b)
\ge
H^{(3)}(x+m_b)
+
(x+m_b)H^{(3)\prime}(x+m_b)
>
0.
$$
The final inequality follows because $(x+m_b)H^{(3)}(x+m_b)$ lies on the strictly increasing branch below $s^*$.
Therefore
$$
xH^{(3)}(x+m_b)
<
m_aH^{(3)}(s^*).
$$
The symmetric argument gives
$$
yH^{(3)}(m_a+y)
<
m_bH^{(3)}(s^*)
$$
when the entrant wins only the bid leg at margin $y<m_b$.
If the entrant wins both legs of a direction, its margins satisfy $x\le m_a$ and $y\le m_b$, with at least one strict inequality, so
$$
(x+y)H^{(3)}(x+y)
<
s^*H^{(3)}(s^*).
$$
If the entrant merely matches a standing leg, tie splitting weakly reduces its revenue and makes the corresponding bound strict whenever that leg carries a positive margin.
Thus, direction by direction, the entrant earns strictly less than the standing ask component $m_aH^{(3)}(s^*)$, the standing bid component $m_bH^{(3)}(s^*)$, or their sum, according to which legs it supplies.
Summing over the six ordered directions, each distinct outcome on which the entrant posts a book can be charged at most two ask components and two bid components.
After multiplying by the ex-ante direction mass $\tfrac{z^{(3)}}6$, the resulting revenue bound per distinct quoted outcome is
$$
\frac{z^{(3)}}6
\,2(m_a+m_b)H^{(3)}(s^*)
=
\frac13z^{(3)}s^*H^{(3)}(s^*)
=
\kappa.
$$
The inequality is strict for every entrant receiving positive revenue, while posting multiple books on the same outcome only increases total cost.
Hence an entrant posting $r$ books earns strictly less than $r\kappa$ and cannot profit.
If more than one standing book is posted on the same outcome at the common split, those books divide a total outcome-book revenue of $\kappa$ while jointly paying more than $\kappa$ in posting costs.
A contestably entry-stable configuration therefore has exactly one active book per outcome.
This proves part~(i) and the claimed characterization within the symmetric common-split class.

\smallskip
\noindent\emph{Step 4: common-gauge invariance and level indeterminacy.}
At a feasible common split, the primitive quotes are
$$
A_i=\frac13+\sigma+m_a,
\qquad
B_i=\frac13-\frac\sigma2-m_b.
$$
A common translation by $c$ changes the margins to
$$
m_a+c,
\qquad
m_b-c,
$$
so their sum remains $s^*$.
For every switch $k\to j$,
$$
(A_k+c)-(B_j+c)=A_k-B_j.
$$
Because the same constant is added to every active ask and bid, every leg ranking and tie set is unchanged.
Routing, demand, and execution quantities are therefore unchanged.
Each book's expected revenue after translation is
$$
\frac13z^{(3)}H^{(3)}(s^*)
\big[
(m_a+c)+(m_b-c)
\big]
=
\frac13z^{(3)}s^*H^{(3)}(s^*)
=
\kappa.
$$
The deviation argument in Step~3 depends only on the feasible split and the sum $s^*$, so the translated configuration remains contestably entry-stable.
The translated quotes remain feasible and pickoff-proof exactly when
$$
\max
\left\{
-m_a,\,
m_b-\left(\frac13-\frac\sigma2\right)
\right\}
\le
c
\le
\min
\left\{
m_b,\,
\frac23-\sigma-m_a
\right\}.
$$
This interval contains $0$ because the original quotes are feasible.
Assumption~\ref{ass:low_fee_3o}(iii) implies $A_i<1$ for every feasible split.
If $m_b>0$, a sufficiently small positive $c$ is feasible.
If $m_b=0$, then $m_a=s^*>0$ and $B_i=\tfrac13-\tfrac\sigma2>0$, so a sufficiently small negative $c$ is feasible.
The feasible gauge interval therefore contains at least one nonzero translation.
The primitive quote level is consequently not identified.
By contrast, translating only the book on outcome $i$ changes every outgoing switch $i\to j$ by $+c$ and every incoming switch $k\to i$ by $-c$ whenever the shifted book remains the winning leg.
A relative translation can therefore change package prices, rankings, routing, and profit, and is not part of the invariant gauge.
This proves parts~(ii) and~(iii).
\end{proof}

\subsection{Proof of Proposition~\ref{prop:self_collateral}}
\label{App:proof_self_collateral}

\begin{proof}
\noindent\emph{Part (i).}
For the LMSR,
$$
p_i(q)
=
\frac{e^{q_i/L}}{\sum_j e^{q_j/L}}.
$$
Suppose first that
$$
p(q')=p(q).
$$
Equality of all log price ratios gives
$$
q'_i-q'_j
=
q_i-q_j
$$
for every $i$ and $j$.
Therefore $q'_i-q_i$ is constant across states, so
$$
q'-q
=
c\mathbf 1
$$
for some $c$.
Conversely, the softmax formula immediately gives
$$
p(q+c\mathbf 1)
=
p(q).
$$
The cost function satisfies
$$
C_L(q+c\mathbf 1)
=
C_L(q)+c.
$$
A complete basket also increases the state-$i$ settlement liability by $c$ for every $i$.
Its net payoff is therefore zero in every state.

Every interior probability vector $p$ is attained by the inventory representative
$$
q_i
=
L\log p_i.
$$
Any other representative of the same probability vector differs by a common translation.
Thus the quotient of inventory space by $\operatorname{span}(\mathbf 1)$ is in bijection with the interior of the probability simplex.

For the CLOB, Lemma~\ref{lem:information_defect} shows that every admissible switch-price-preserving reposting has the form
$$
A_i\mapsto A_i+c,
\qquad
B_i\mapsto B_i+c.
$$
Reposting quotes without an execution changes no participant's position or settlement payoff.
The gauge nevertheless moves every primitive quote level by $c$.
Other quote movements are not in the kernel characterized by Lemma~\ref{lem:information_defect} and therefore alter at least one executable switch price.

\smallskip
\noindent\emph{Part (ii).}
Choose representatives
$$
q_i
=
L\log p_i+k,
\qquad
q'_i
=
L\log p'_i+k'.
$$
Because the probabilities sum to one,
$$
C_L(q)
=
L\log\sum_i p_i e^{k/L}
=
k,
$$
and similarly
$$
C_L(q')
=
k'.
$$
The state-$i$ inventory change is
$$
q'_i-q_i
=
L\log\frac{p'_i}{p_i}
+
(k'-k).
$$
Subtracting the cost increment gives
$$
\begin{aligned}
(q'_i-q_i)
-
\big(C_L(q')-C_L(q)\big)
&=
L\log\frac{p'_i}{p_i}
+
(k'-k)
-
(k'-k)
\\
&=
L\log\frac{p'_i}{p_i}.
\end{aligned}
$$
Taking expectations under $\pi$ gives
$$
\begin{aligned}
\sum_i
\pi_i
L\log\frac{p'_i}{p_i}
&=
L\sum_i
\pi_i
\left(
\log\frac{\pi_i}{p_i}
-
\log\frac{\pi_i}{p'_i}
\right)
\\
&=
L
\left[
\mathrm{KL}(\pi\|p)
-
\mathrm{KL}(\pi\|p')
\right].
\end{aligned}
$$

Under the proportional fee convention, the trader pays
$$
\frac{C_L(q')-C_L(q)}{1-\tau}
$$
for a purchase with a nonnegative cost increment.
Therefore the fee-adjusted state-$i$ payoff is
$$
\begin{aligned}
(q'_i-q_i)
-
\frac{C_L(q')-C_L(q)}{1-\tau}
&=
(q'_i-q_i)
-
\big(C_L(q')-C_L(q)\big)
\\
&\qquad
-
\frac{\tau}{1-\tau}
\big(C_L(q')-C_L(q)\big)
\\
&=
L\log\frac{p'_i}{p_i}
-
\frac{\tau}{1-\tau}
\big(C_L(q')-C_L(q)\big).
\end{aligned}
$$
The second term is the same in every state.

Suppose the resulting payoff were constant across states.
Then
$$
\log\frac{p'_i}{p_i}
$$
would be constant across $i$.
There would therefore be a constant $d>0$ such that
$$
p'_i
=
d p_i
$$
for every $i$.
Because both probability vectors sum to one, this requires
$$
d=1,
$$
and hence
$$
p'=p.
$$
Thus every nontrivial movement of the LMSR display creates state-dependent exposure, both before and after the state-independent fee charge.

\smallskip
\noindent\emph{Part (iii).}
For every atomic switch $k\to j$, the translated net package price is
$$
(A_k+c)-(B_j+c)
=
A_k-B_j.
$$
Because both quote collections remain pickoff-proof, Stage-1 $c$-flow executes against neither collection.
Because every Stage-2 participation decision depends on the atomic switch price, the same traders execute the same switch quantities under both collections.
Their net cash payments and settlement positions are therefore identical in every realization.
The reposting itself executes no trade, so it changes no existing position, while every primitive quote level moves by $c$.

If the same translation is applied to every active maker, all ask rankings, bid rankings, and tie sets are also unchanged.
A relative translation of only one maker can change those rankings and is therefore not an invariance of the multi-maker market.
This proves that the CLOB display is uncollateralized along the common gauge.
\end{proof}

\subsection{Proof of Proposition~\ref{prop:event_level_collateral_efficiency}}
\label{App:proof_of_event_level_coll}

\begin{proof}
\noindent\emph{LMSR worst-case loss.}
The normalized cost satisfies
$$
C_L(0)=0.
$$
If terminal inventory is $q$ and outcome $i$ realizes, the fee-free maker has collected $C_L(q)-C_L(0)$ and owes $q_i$, so its loss is
$$
q_i-
\big(
C_L(q)-C_L(0)
\big)
=
q_i-C_L(q).
$$
For every $q$ and $i$,
$$
C_L(q)
=
L\log
\left(
\frac1n
\sum_{j=1}^n
e^{q_j/L}
\right)
\ge
L\log
\left(
\frac1n
e^{q_i/L}
\right)
=
q_i-L\log n.
$$
Hence
$$
q_i-C_L(q)
\le
L\log n.
$$
To show that the bound is sharp, set $q_i=t$ and $q_j=0$ for every $j\ne i$.
Then
$$
\begin{aligned}
q_i-C_L(q)
&=
t-
L\log
\left(
\frac{e^{t/L}+n-1}{n}
\right)
\\
&=
L\log n
-
L\log
\left(
1+(n-1)e^{-t/L}
\right)
\longrightarrow
L\log n
\end{aligned}
$$
as $t\to\infty$.
The bound is approached but is not attained at any finite inventory.
Therefore
$$
K_n^{LMSR}(L)
=
L\log n.
$$

\smallskip
\noindent\emph{Segregated CLOB collateral.}
Under the proposition's segregation assumption, each of the $n$ books must independently collateralize exposure $\vartheta$.
Because the platform does not net settlement obligations across books, event-level collateral is the sum
$$
K_n^{CLOB}(\vartheta)
=
\sum_{i=1}^n
\vartheta
=
n\vartheta.
$$

\smallskip
\noindent\emph{Exact and asymptotic comparisons.}
Dividing the two exact requirements gives
$$
\frac{K_n^{CLOB}(\vartheta)}
{K_n^{LMSR}(L)}
=
\frac{\vartheta}{L}
\frac{n}{\log n}.
$$
If $\vartheta/L$ is bounded above and bounded away from zero uniformly in $n$, multiplying $n/\log n$ by that factor preserves its asymptotic order.
Thus
$$
\frac{K_n^{CLOB}}
{K_n^{LMSR}}
=
\Theta
\left(
\frac{n}{\log n}
\right).
$$
Finally, a nonnegative fee retained by the LMSR LP adds cash revenue without increasing state-contingent settlement liabilities.
It therefore weakly lowers realized loss relative to the fee-free calculation, so $L\log n$ is a conservative fee-inclusive collateral bound.
\end{proof}

\section{Empirical Construction of the Stylized Facts}

\subsection{Three Facts}
\label{app:empirical}
This appendix records the data, classification, and measurement behind the three stylized facts of~\S\ref{sec:stylized_facts}. The
purpose is to make each headline number reproducible and to show that the regularities do not depend on discretionary choices in the pipeline.

\paragraph{Data and sample.}
The data are the full transaction record of Bitcoin five-minute
up/down binary markets on a major prediction-market CLOB: $738{,}393$
trades across $275$ markets over a roughly twenty-four-hour window,
with total notional volume of \$10.38M executed by $15{,}134$ distinct
accounts. Each market lists two complementary Arrow--Debreu claims that
settle to $0$ or $1$ at the five-minute mark. Makers pay no fee and
takers pay $1.8\%$; because maker PnL is measured net of fees and makers
are unrebated, the positive maker PnL documented below is not a
fee-rebate artifact.

\paragraph{Account keying.}
We key accounts on wallet address rather than on the platform display
label. The two keys differ for exactly one reason: a single shared
``anonymous'' display label aggregates $1{,}328$ distinct wallets, and
keying on the label alone would pool these unrelated accounts into one
fictitious super-account. Wallet keying dissolves that aggregate and
removes the need for any ad hoc exclusion; every other account is
one-to-one across the two keys.

\paragraph{Account classification.}
We classify each account by its trading footprint into market maker,
retail, or unclassified, using two-sidedness of quoting, round-trip
frequency, within-market inventory mean-reversion, holding horizon, and
clip size. 
Table~\ref{tab:app_classification} reports the wallet-keyed
partition. Within the maker class we isolate the \emph{persistent-maker} cohort of $407$ accounts that quote across many markets; the residual $44$ maker-classified accounts are structurally similar but persistently unprofitable, and we report the persistent cohort in the body.

\begin{table}
\caption{Wallet-keyed account classification. Volume shares sum to one
across classes; PnL figures are net of fees and reconcile to the book
total of $-\$51.0\text{K}$.}
\label{tab:app_classification}
\centering
\begin{tabular}{lccc}
\hline
Class & Accounts & Volume share & Net PnL \\
\hline
Market maker & $451$ & $24.9\%$ & $+\$29.0\text{K}$ \\
\quad of which persistent ($407$) & $407$ & --- & $+\$45.3\text{K}$ \\
\quad of which residual ($44$) & $44$ & --- & $-\$16.3\text{K}$ \\
Retail & $12{,}749$ & $55.0\%$ & $-\$66.7\text{K}$ \\
Unclassified & $1{,}934$ & $20.1\%$ & $-\$13.3\text{K}$ \\
\hline
Total & $15{,}134$ & $100\%$ & $-\$51.0\text{K}$ \\
\hline
\end{tabular}
\end{table}

\paragraph{PnL decomposition (Facts 1 and 2).}
We split each maker's PnL into the component realized at resolution and
the component realized on intraperiod round trips. For the
persistent-maker cohort, $91.4\%$ of profit is resolution PnL and
$8.6\%$ is spread. Decomposing the resolution component by trade
direction and execution price (Table~\ref{tab:app_pnl_decomp}) shows
that the profit comes from buying under-priced favorites rather than
from selling them: the buy side earns roughly $+\$70\text{K}$ while the
entire sell side contributes about $+\$2.6\text{K}$. The carried
positions are informed rather than random---a maker's directional
inventory resolves in its favor $56.8\%$ of the time.

\begin{table}
\caption{Persistent-maker PnL by trade direction and execution-price
band. The listed components are the dominant contributors and do not sum
to the cohort total, which nets all positions across every band.}
\label{tab:app_pnl_decomp}
\centering
\begin{tabular}{lc}
\hline
Component & PnL \\
\hline
Buy, price $0.50$--$0.80$ & $+\$52.4\text{K}$ \\
Buy, price $0.80$--$1.00$ & $+\$24.8\text{K}$ \\
Sell, price ${\ge}\,0.80$ & $-\$0.5\text{K}$ \\
Buy side, all bands & ${\approx}\,+\$70\text{K}$ \\
Sell side, all bands & ${\approx}\,+\$2.6\text{K}$ \\
\hline
Cohort total (net of all positions) & $+\$45.3\text{K}$ \\
Resolution share / spread share & $91.4\%\,/\,8.6\%$ \\
\hline
\end{tabular}
\end{table}

\paragraph{Regime classification and calibration (Fact 3).}

\begin{table}
\caption{Calibration gap (average transaction price minus empirical
resolution frequency; positive ${=}$ overpriced) by price band and
price-path regime. The near-zero unconditional gaps hide a sign-flipping
regime interaction.}
\label{tab:regime_calibration}
\centering
\begin{tabular}{lccc}
\hline
Price band & Unconditional & Trending & Flippy \\
\hline
Longshot $[0.20,\,0.40)$ & $+0.043$ & $+0.225$ & $-0.071$ \\
Favorite $(0.60,\,0.80]$ & $-0.037$ & $-0.205$ & $+0.066$ \\
\hline
\end{tabular}
\end{table}

We label each market's within-window price path as \emph{trending} or
\emph{flippy} with a run-based classifier: paths dominated by a single
directional run are trending, and paths with frequent sign reversals are
flippy. For each trade we define the calibration gap as the execution
price minus the empirical resolution frequency of that side, so a
positive gap means the side is overpriced. Pooled across regimes the
band-level gaps are small and of mixed sign
(Table~\ref{tab:app_unconditional}); conditioning on the regime reveals
that the mid-band pooled gaps are the average of large opposite-signed
regime values, which is the reversal reported in
Table~\ref{tab:regime_calibration}. The boundaries are asymmetric: the
upper boundary is calibrated (an average price of ${\approx}\,\$0.982$
against a $98.4\%$ resolution frequency), while the lower boundary stays
overpriced (${\approx}\,\$0.024$ against $1.4\%$).

\begin{table}[H]
\caption{Unconditional (pooled-across-regime) calibration gap by price
band. The small and mixed-sign mid-band values average over the
sign-flipping regime interaction of Table~\ref{tab:regime_calibration}.}
\label{tab:app_unconditional}
\centering
\begin{tabular}{lc}
\hline
Price band & Unconditional gap \\
\hline
$[0.00,\,0.05)$ & $+0.010$ \\
$[0.20,\,0.40)$ & $+0.043$ \\
$[0.60,\,0.80)$ & $-0.037$ \\
$[0.95,\,1.00)$ & $-0.002$ \\
\hline
\end{tabular}
\end{table}

\paragraph{Robustness.}
The trade- and market-level regularities---the calibration gaps, the
regime sign reversal, and the maker rent decomposition---are invariant
to whether accounts are keyed on display label or wallet address;
keying changes only account-level magnitudes, and does so solely through
the disaggregation of the shared anonymous label described above. The
regime sign reversal survives market-level cluster-bootstrap resampling,
so it is not an artifact of a few heavily traded markets.

\subsection{Spread and the Outcome Count}
\label{app:empirical_spread}

Completing an outcome set on separate binary books costs more than a dollar, and more so as outcomes are added; the excess is per-leg committed markup, while the books' displayed centers sum to par on average and exceed it only when books are compared at the same age.
This subsection documents both patterns in a different slice of the~\citet{akey2026wins} record from the maker-PnL facts above: the large-scale Polymarket transaction history from November 11, 2022 through March 29, 2026, with multi-outcome books reconstructed from their binary components.
The analysis was run as a pre-registered sequence---the rules fixing how each possible result pattern would be read were committed before the robustness battery ran---and every regression below is written out with its estimator and standard-error type; throughout, $\beta$ denotes the coefficient on $\log n_i$, where $n_i$ is event $i$'s outcome count.

\paragraph{Multi-outcome reconstruction.}
The native \texttt{n\_outcomes} field cannot supply multi-outcome price data: the markets it flags all closed before the CTF Exchange that sources the trade record existed, so they carry no trades.
The live platform instead represents an $n$-way question as $n$ independent binary Yes/No markets grouped under one event identifier, one market per outcome with its own Yes/No token pair.
We reconstruct the $n$-outcome book by grouping every market whose outcomes are exactly $\{\text{Yes},\text{No}\}$ under a common event and summing the Yes-leg price across siblings.
Grouping alone is not enough: sports events bundle unrelated prop bets under one identifier, and their Yes prices need not sum to one (``Lakers win'' and ``LeBron scores 19+'' can both resolve Yes).
We keep only the events a genuine $n$-way question must satisfy---every sibling resolved, and exactly one resolved Yes---which discards prop bundles and unresolved events and leaves $16{,}036$ validated events with outcome counts $2$--$10$, of which $9{,}988$ ever have every leg priced on both sides.

\paragraph{Signed premium and its decomposition.}
The premium of a book at time $t$ is the cost of buying YES on every outcome, minus one dollar; it is signed by design (negative = the set is an arbitrage in the buyer's favor, positive = a cost to whoever completes the set).
Pricing each leg $j$ from \emph{both} sides of its book---the ask $a_{jt}$ (last taker-buy execution) and the bid $b_{jt}$ (last maker-buy execution), with mid $m_{jt}=(a_{jt}+b_{jt})/2$---splits the premium as an accounting identity into the model's two objects:
$$
\underbrace{\textstyle\sum_j a_{jt}-1}_{y^{\text{ask}}_{it}}
\;=\;
\underbrace{\textstyle\sum_j m_{jt}-1}_{y^{\text{mid}}_{it}:\ \text{book-center term}}
\;+\;
\underbrace{\textstyle\sum_j (a_{jt}-m_{jt})}_{y^{\text{HS}}_{it}:\ \text{committed-markup term}} ,
$$
so the estimated $\log n$ slope of the premium decomposes exactly, $\beta_{\text{ask}}=\beta_{\text{mid}}+\beta_{\text{HS}}$: mispricing of the display versus accumulation of the per-leg premium of~\eqref{eq:switch_prices_and_premium}.

\paragraph{Pricing each book over its life.}
The record has no live order book, only executed trades, so each leg is priced at any moment by its most recent prior trade on the relevant side, with the age of that price, $\Delta_{jt}$ (minutes since the trade), recorded.
A price stands until the next trade replaces it, so a book's premium is a step function of time: we build one row per \emph{price-update moment} (any timestamp at which some leg trades), carrying the new premium and its duration $w_{ir}$ until the next update, over the window from the first moment every leg has traded on both sides to the event's close ($3{,}343{,}329$ rows across the $9{,}988$ events; each event's durations sum exactly to its window length).
Event $i$'s premium is then the time-weighted average
$$
\bar y_i^{\,\text{tw}}=\frac{\sum_r y_{ir}\,w_{ir}}{\sum_r w_{ir}}
$$
---what a buyer arriving at a uniformly random moment of the book's life would face.
Because the two sides' last trades need not be simultaneous, any single row's markup term can reflect price movement between the two trades rather than a true bid--ask gap (about $8\%$ of rows even show ask below bid); this error is symmetric in direction and cancels both across a book's legs, whose prices must always sum to roughly one dollar, and in the time averages that all estimates use.

\paragraph{Weighting: the three estimators.}
Outcome count is fixed within an event, so whether the premium rises with it is a question about differences \emph{between} events, and pooling only rows might answers it at the wrong level.
Therefore estimation is done three ways.
Pooled, with rows weighted by duration and errors clustered by event,
$$
y_{ir}=\alpha+\beta\,\log n_i+\varepsilon_{ir};
$$
event-level---one event, one vote; the estimator used throughout unless stated---
$$
\bar y_i^{\,\text{tw}}=\alpha+\beta\,\log n_i+\varepsilon_i,
\qquad\text{HC3 robust SE};
$$
and outcome-level, collapsing to one row per outcome count, $\bar y_n$ the mean of the level-$n$ event averages with standard error $s_n$, by weighted least squares,
$$
\bar y_n=\alpha+\beta\,\log n+\varepsilon_n,
\qquad w_n=1/s_n^2,\qquad n\in\{2,\dots,10\}.
$$
As a dependence check, an event-day panel (each event's covered time split across calendar days $d$, with day-level time-weighted means) is estimated by weighted least squares with covered minutes as weights and two-way clustering,
$$
y_{id}=\alpha+\beta\,\log n_i+\varepsilon_{id},
\qquad
\widehat V=\widehat V_{\text{event}}+\widehat V_{\text{date}}-\widehat V_{\text{event}\times\text{date}},
$$
with Driscoll--Kraay errors (five-day bandwidth) as a cross-check; every reported significance survives the more conservative of the two.

\paragraph{Result 1: the decomposition of the premium slope.}
Table~\ref{tab:spread_decomposition} and Figure~\ref{fig:time_weighted_premium_slop} reports the event-level slope of each term.
The cost of the set rises in the outcome count, and $96\%$ of the rise sits in the committed-markup term; the book-center slope is statistically zero.
Attributing the markup term across legs by their price level shows the slope is carried by the cheap legs a large book must mechanically contain---leg prices average $1/n$, and cheap claims carry the widest markups relative to price---while legs above $\$0.80$ contribute \emph{negatively}, since large books hold few expensive YES legs.
A pure tick floor, obtained by substituting the minimum-increment bound $n_i\delta/2$ for the dependent variable in the event-level regression, accounts for between roughly $4\%$ ($\delta{=}\$0.001$) and $44\%$ ($\delta{=}\$0.01$) of the markup slope on the platform's mixed tick grid.

\begin{table}[H]
\caption{Decomposition of the premium--outcome slope. Coefficients on $\log n$, multi-outcome events, time-weighted over each book's life; event-level estimator (one event, one vote; HC3 errors), with the outcome-level estimate in brackets. The identity $\beta_{\text{ask}}=\beta_{\text{mid}}+\beta_{\text{HS}}$ holds by construction. $9{,}988$ events.}
\label{tab:spread_decomposition}
\centering
\begin{tabular}{lc}
\hline
& Time-weighted, event-level \\
\hline
$\beta_{\text{ask}}$ (set premium) & $+0.0595$ (se $0.0056$, $p{<}0.001$) $[+0.053]$ \\
$\beta_{\text{HS}}$ (committed markup) & $+0.0573$ (se $0.0029$, $p{<}0.001$) $[+0.052]$ \\
$\beta_{\text{mid}}$ (book center) & $+0.0022$ (se $0.0034$, $p{=}0.52$) $[-0.000]$ \\
\hline
\end{tabular}
\end{table}

\begin{figure}[H]
    \centering
    \includegraphics[width=0.8\linewidth]{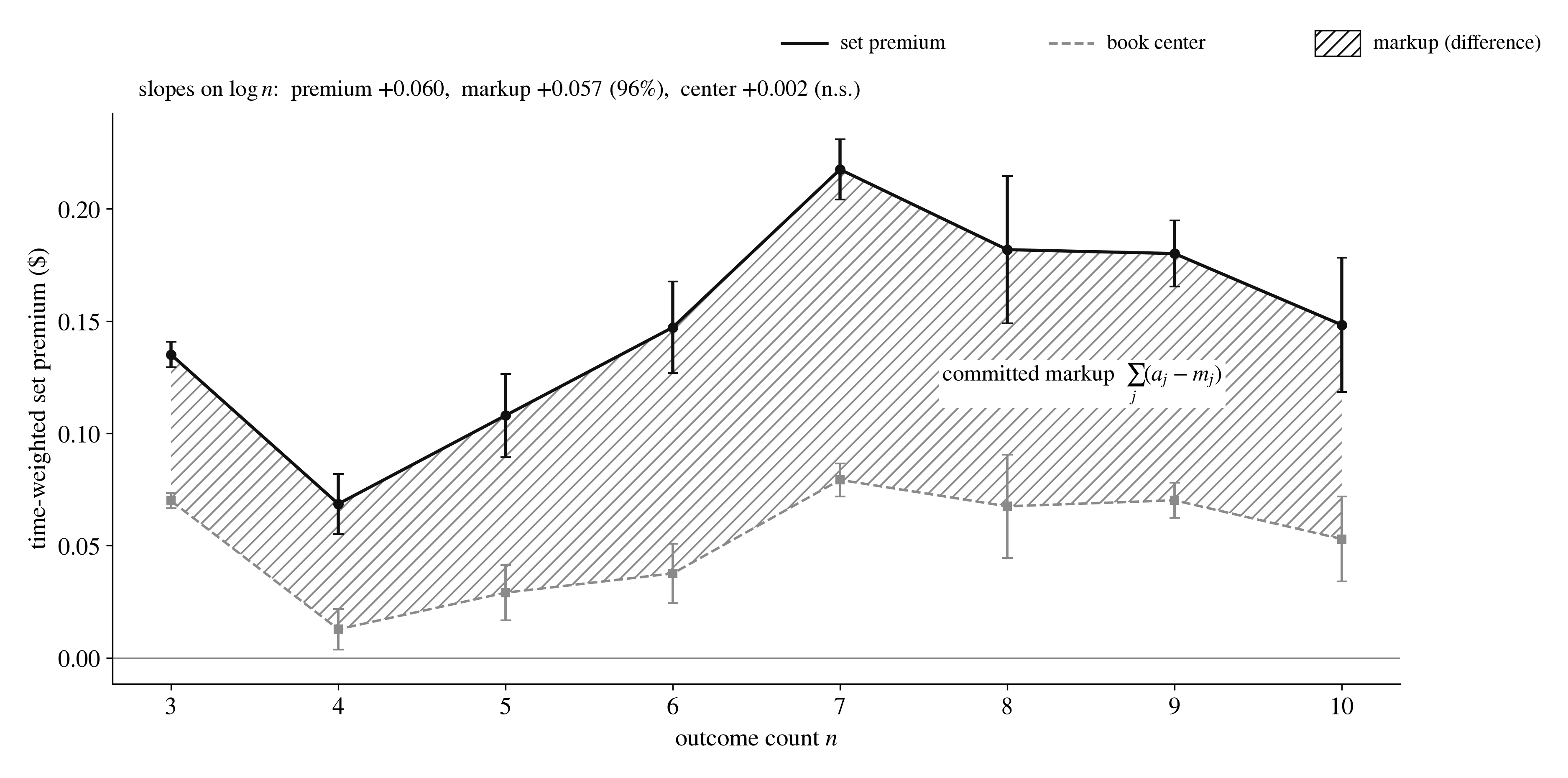}
    \caption{Time-weighted Average Premium v.s. Outcome Count}
    \label{fig:time_weighted_premium_slop}
\end{figure}

\paragraph{Result 2: price staleness is not generating the slope.}
Because prices are last-trade prices, a larger book mechanically carries older prices on average (the oldest of $n$ legs ages faster in $n$), so staleness is the first confound to rule out; three checks close it from three directions.

(i)~\emph{Equal-age comparison.} On the panel of moments at which every leg has just traded---where price age is directly measurable leg by leg and varies widely---we estimate, with event-clustered errors,
$$
y_{it}=\alpha+\beta\,\log n_i
+f_1\!\big(\log(1{+}\Delta^{\max}_{it})\big)
+f_2\!\big(\log(1{+}\bar\Delta_{it})\big)
+\gamma\,\mathrm{sd}(\Delta_{it})
+\delta\,\log n_i\cdot\big(\bar\Delta^{\text{hr}}_{it}-0.5\big)
+\varepsilon_{it},
$$
with $f_1,f_2$ natural cubic splines and the interaction centered at $30$ minutes, so that $\beta$ is the slope comparing books at \emph{equal} price age ($30$ minutes) and the slope at age $m$ hours is $\beta+\delta(m-0.5)$.
Comparing equally-aged books preserves $84$--$98\%$ of the slope (premium: $+0.0546\rightarrow+0.0461$, $p{=}10^{-14}$; center: $+0.0267\rightarrow+0.0263$), and $\delta\approx0$ ($-0.00003$, $p{=}0.83$): the effect is not ``bigger books, older prices.''

(ii)~\emph{Drift correction.} An old price is the price from $\Delta$ minutes ago, so its expected error is how much prices typically move in $\Delta$ minutes.
From $2.7$M (ask) and $5.6$M (bid) consecutive-trade pairs of multi-outcome legs we estimate the drift surface as cell means,
$$
\widehat d(p,\Delta)=\mathbb{E}\big[\,p_{t+\Delta}-p_t\;\big|\;p_t\in\text{price bucket},\ \Delta\in\text{age bucket}\,\big],
$$
excluding eventual winner status (future information; descriptively, at under one day to resolution winning legs drift $+\$0.024$ per gap and losing legs $-\$0.008$), and re-run the event-level regressions on the corrected premium
$$
y^{\text{corr}}_{it}=\textstyle\sum_j\big[a_{jt}+\widehat d(a_{jt},\Delta_{jt})\big]-1 .
$$
Every slope moves by at most $0.006$ (premium $+0.0595\rightarrow+0.0598$; center $+0.0022\rightarrow+0.0081$): the leg-level errors offset within a set, because one leg's rise is its siblings' fall under the adding-up constraint that defines the set.

(iii)~\emph{Placebo.} On $52{,}643$ binary snapshots whose prices are essentially current (age ${\le}5$ minutes), each leg is deliberately re-priced as of $t-\Delta$, with $\Delta$ drawn from the \emph{real} per-leg age distribution $\widehat D_k$ observed in $n{=}k$ books, and the resulting pseudo-premium regressed on $\log k$:
$$
\tilde y_{mtk}=p^{\text{asof}}_{m,\mathrm{YES}}(t-\Delta_1)+p^{\text{asof}}_{m,\mathrm{NO}}(t-\Delta_2)-1,
\qquad
\tilde y=\alpha+\beta_{\text{placebo}}\log k+\varepsilon,
$$
with errors clustered by market $m$.
The result is $\beta_{\text{placebo}}=-0.0031$ (se $0.0005$) pooled and $-0.0050$ (se $0.0009$) between: imposing large-book staleness on clean books produces $-5\%$ to $-8\%$ of the real slope, with the \emph{wrong} sign.
Therefore we confirm that whatever makes bigger books more expensive is in the prices and not in when we read them, see Figure~\ref{fig:Robusteness_check_same_staleness} for a graphical illustration on how fast-priced binary books's premium are evaluated with the same staleness as the multi-outcome books.

\begin{figure}[H]
    \centering
    \includegraphics[width=0.8\linewidth]{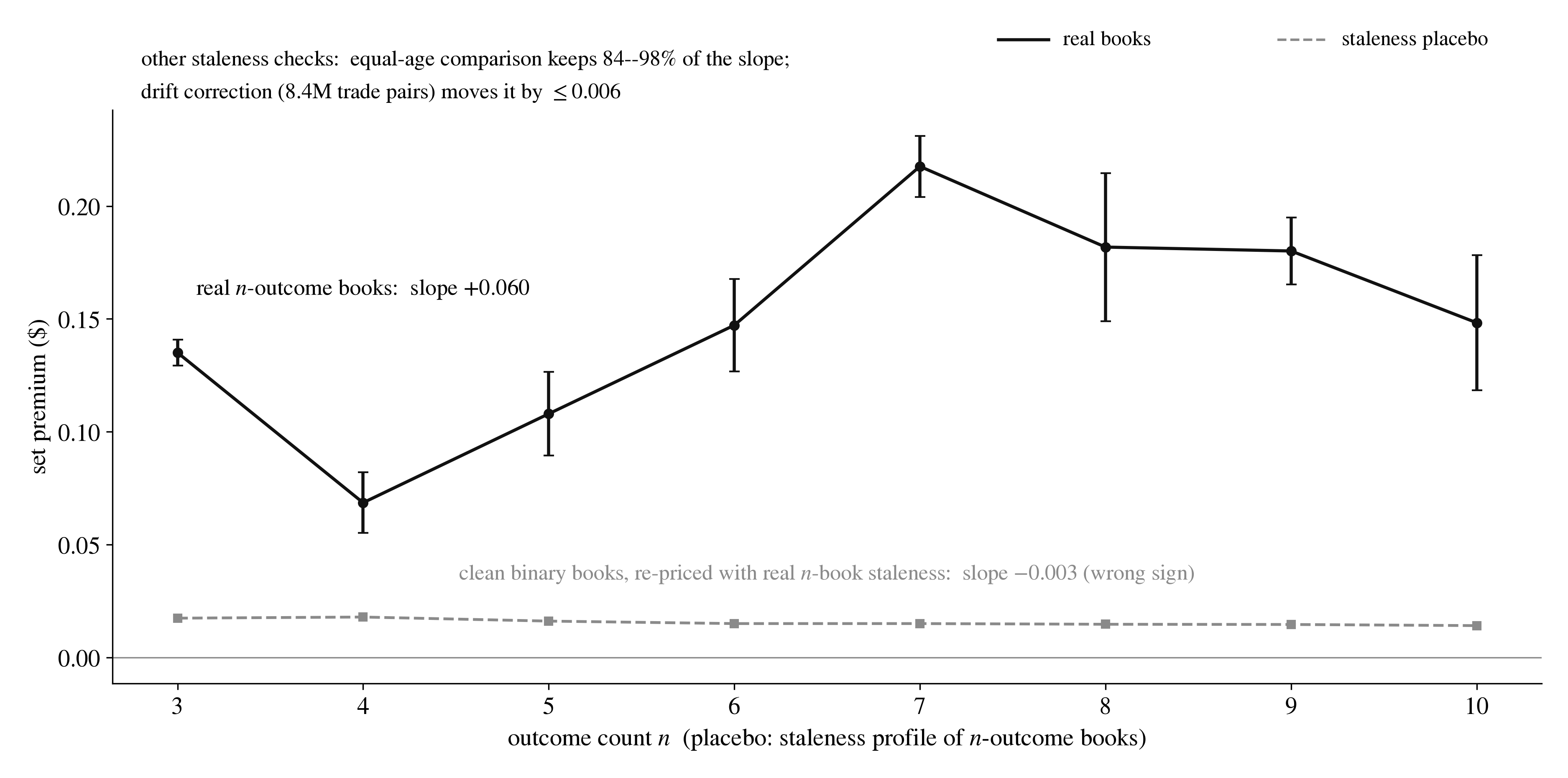}
    \caption{Real $n$-outcome Premium v.s. Binary Premium: Same Staleness}
    \label{fig:Robusteness_check_same_staleness}
\end{figure}

\paragraph{Result 3: the book-center drift appears once books are compared at the same age.}
The zero center slope in Table~\ref{tab:spread_decomposition} is a composition across life stages.
A book enters the panel once every leg has traded on both sides, and that happens at very different points in life: at a median $89\%$ of the way through life for $3$--$4$-outcome (sports-type) books versus $28$--$38\%$ for larger books.
Lifetime averages therefore compare small books in their tight final stretch against large books across their uncertain middle age; three specifications hold age fixed.
First, cumulative controls on the event-level regression (HC3),
$$
\bar y_i^{\,\text{tw}}=\alpha+\beta\,\log n_i
+\gamma\,\log(1{+}\mathrm{medTTR}_i)
+\text{category FE}
+\lambda_1\log(1{+}\mathrm{vol}_i)+\lambda_2\log(1{+}\mathrm{vol}^{\min}_i)
+\text{listing-year FE}
+\varepsilon_i,
$$
added block by block (Table~\ref{tab:spread_conditional}, left): the time-to-resolution control alone moves the center slope from $+0.002$ to $+0.055$ and the premium slope from $+0.060$ to $+0.138$ (all blocks: $+0.076$/$+0.143$).
The controls \emph{raise} the slope.
Second, the slope within each time-to-resolution bucket $b$, with time-weighted averages recomputed from the time each event spends in that bucket,
$$
\bar y_i^{\,\text{tw}(b)}=\alpha_b+\beta_b\,\log n_i+\varepsilon_i
\qquad(\text{HC3, one row per event per bucket}),
$$
reported in Table~\ref{tab:spread_conditional} (right): positive at every horizon.
Third, the same construction within each tenth of a book's life,
$$
\bar y_i^{\,\text{tw}(d)}=\alpha_d+\beta_d\,\log n_i+\varepsilon_i ,
\qquad d\in\{0,\dots,9\},
$$
gives $\beta_d^{\text{mid}}\in[+0.023,+0.063]$ in \emph{all ten} deciles, including $+0.029$ ($p{<}0.001$) in the final decile before resolution, where prices are most informative.
One complementary fact: during stretches in which every leg has traded within the past few hours, the center slope is mildly \emph{negative} (${\approx}-0.04$)---bursts of trading compress large books' centers---which is what intermittent enforcement of the set constraint looks like.
Inference is robust to two-way clustering (event $\times$ date), Driscoll--Kraay errors, clustering on resolution month.
See Figure~\ref{fig:book_center_slop} for a graphical illustration.

\begin{table}[H]
\caption{The conditional book-center slope. Left: event-level slope of the center term on $\log n$ with cumulative controls (HC3). Right: the same slope estimated within time-to-resolution buckets. The zero in Table~\ref{tab:spread_decomposition} is the across-age composition of uniformly positive same-age gradients.}
\label{tab:spread_conditional}
\centering
\begin{tabular}{lc@{\qquad}lc}
\hline
Specification & $\beta_{\text{mid}}$ & TTR bucket & $\beta_{\text{mid}}$ \\
\hline
baseline (no controls) & $+0.002$ (n.s.) & $<1$ day & $+0.009$ ($p{=}0.008$) \\
$+\ \log$ median TTR & $+0.055$ ($p{<}0.001$) & $1$--$7$ days & $+0.051$ ($p{<}0.001$) \\
$+$ category FE & $+0.081$ ($p{<}0.001$) & $7$--$30$ days & $\mathbf{+0.081}$ ($p{<}0.001$) \\
$+$ volumes, listing-year FE & $+0.076$ ($p{<}0.001$) & $>30$ days & $+0.060$ ($p{<}0.001$) \\
\hline
\end{tabular}
\end{table}

\begin{figure}[H]
    \centering
    \includegraphics[width=0.9\linewidth]{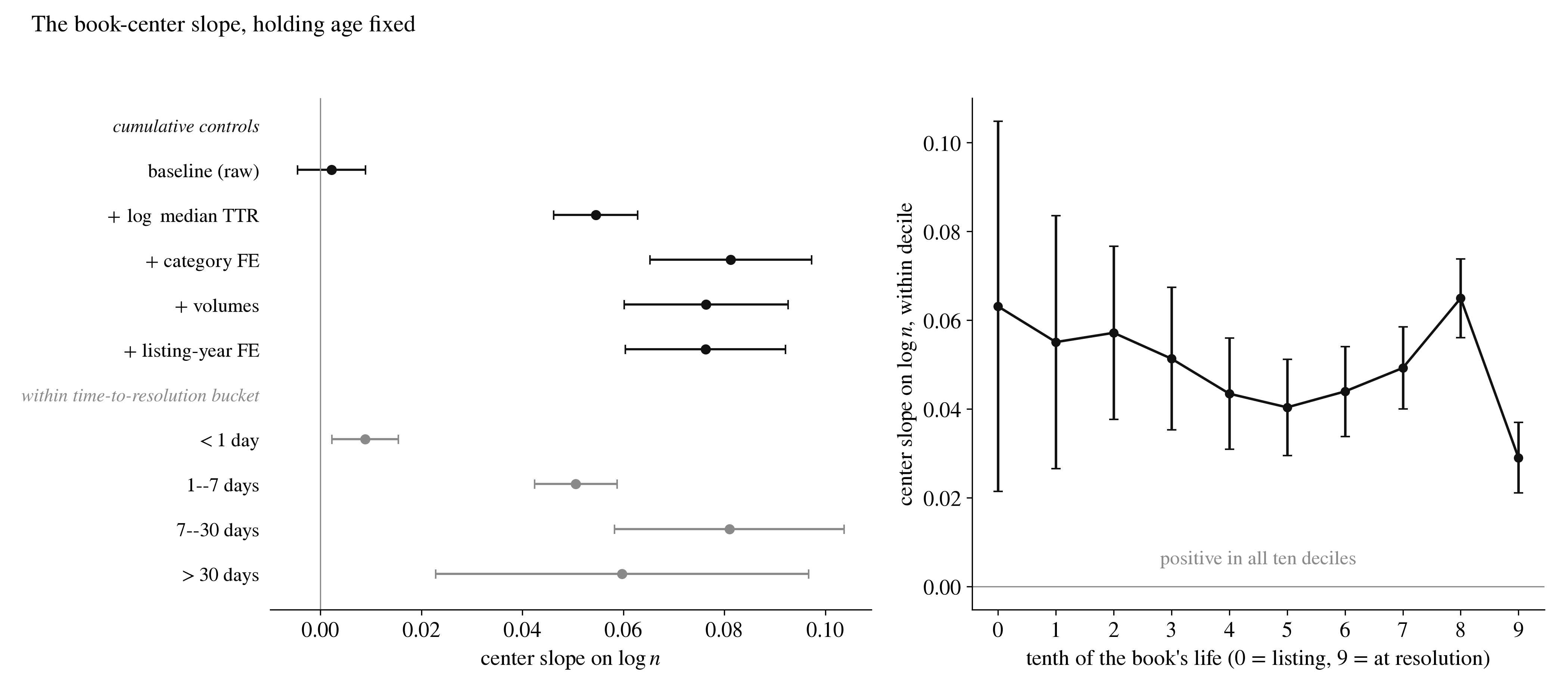}
    \caption{Book Center Slope Fixed Effect and Lifecycle Control}
    \label{fig:book_center_slop}
\end{figure}

\paragraph{Result 4: how often a book is measurable, and whether that matters.}
A book can only be priced at moments when its legs have traded recently, and a large book's longshot legs trade rarely---so the hours at which a large book is visible could be an unrepresentative sample of its life: visible only when busy, and perhaps only when cheap.
We handle this in three steps: measure the visibility problem, reweight for it, and then ask how bad the worst case could be.

\emph{Step 1: measure it.}
We chop every event's life into hours and record, for each hour, whether the book was measurable that hour.
A model then predicts visibility from what we can observe about the hour,
$$
\Pr(\text{measurable}_{ih}=1)=\Lambda\!\big(\theta_0+\theta_1\log n_i+\theta_2\log(1{+}\mathrm{vol}_i)
+g(\log(1{+}\mathrm{TTR}_{ih}))+\text{category FE}+\text{hour-of-day FE}\big),
$$
with $g$ a natural cubic spline, estimated by ridge-penalized logistic regression on a case-control subsample with the intercept corrected back to the population scale.
The problem is real and large: $\theta_1=-1.8$ at a one-hour threshold ($-0.75$ at one day), i.e.\ an hour in a large book's life is far less likely to be visible than an hour in a small book's.

\emph{Step 2: reweight.}
An hour that was visible \emph{despite} long odds---predicted probability $\widehat p_{ih}=0.05$, say---is our only witness for the roughly twenty similar hours we never saw, so we count it twenty times:
$$
\bar y_i^{\,\text{w}}
=\frac{\sum_h y_{ih}\,\mathrm{cov}_{ih}/\widehat p_{ih}}{\sum_h \mathrm{cov}_{ih}/\widehat p_{ih}},
\qquad 1/\widehat p \text{ truncated at its } 99\text{th percentile}.
$$
If the kinds of hours we tend to see carried different premia from the kinds we tend to miss, this reweighting would move the estimates.
It moves none by more than $0.009$ and flips no sign: along everything observable, therefore the hours we can see look like the hours we cannot.

\emph{Step 3: bound the worst case.}
Reweighting cannot rule out differences driven by things no model observes, so we brute-force it: every invisible hour is filled with a deliberately extreme guess---the event's own $10$th-percentile (pessimistic) or $90$th-percentile (optimistic) visible premium,
$$
\bar y_i(\text{fill})=\frac{\bar y_i^{\,\text{cov}}\cdot \mathrm{cov}_i+q_i\cdot(\mathrm{window}_i-\mathrm{cov}_i)}{\mathrm{window}_i},
\qquad q_i\in\{p10_i,\,p90_i\},
$$
including the adversarial combination that works hardest against us (optimistic fills for larger-than-average books and pessimistic for smaller, and the reverse).
At age thresholds of four hours or less, these fills can flip the sign of the regression, therefore we only claim robustness  for full-life and same-age estimates, but not the estimate for tight thresholds.
See~\ref{fig:measurable_plot} for a graphical illustration.

\begin{figure}[H]
    \centering
    \includegraphics[width=0.8\linewidth]{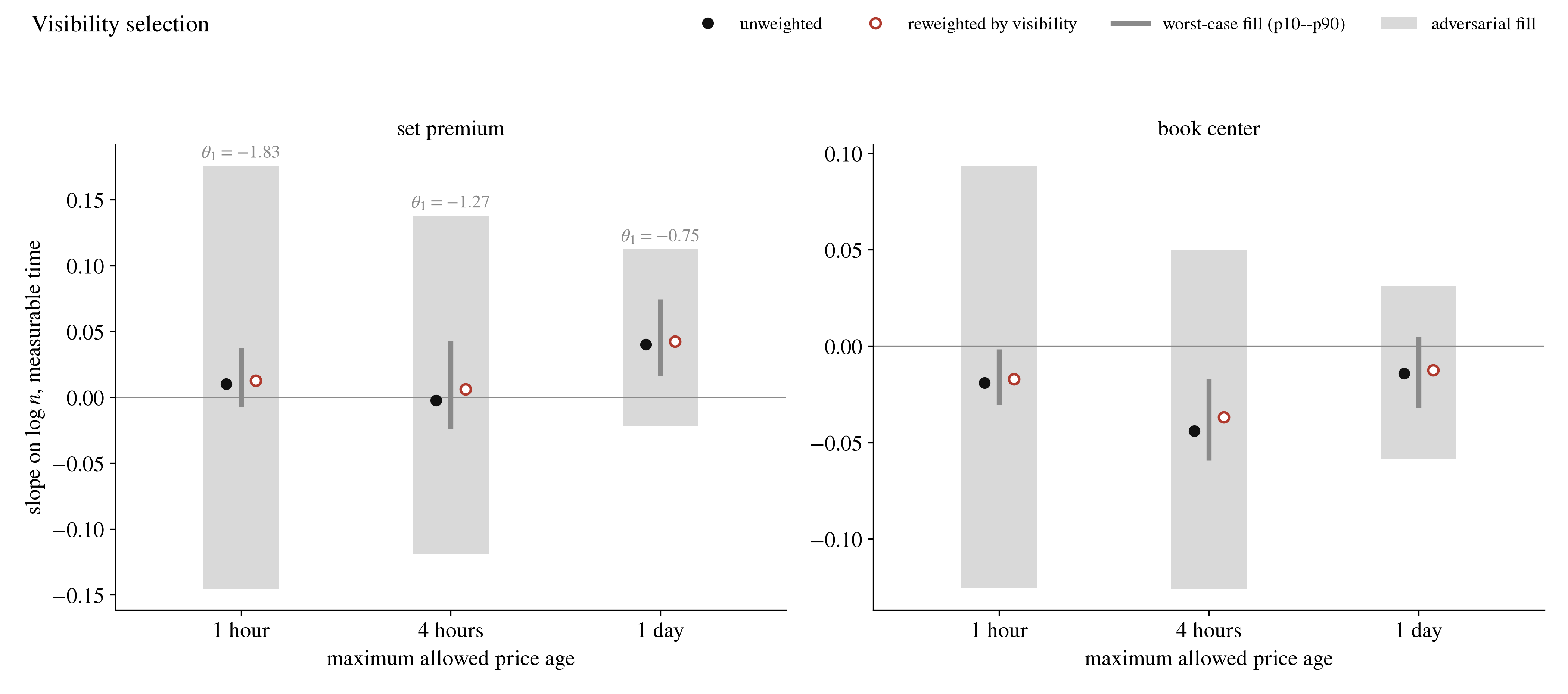}
    \caption{Refill Estimates}
    \label{fig:measurable_plot}
\end{figure}

\paragraph{Result 5: counterfactuals.}
(i)~\emph{Pseudo-events: stacking unrelated books.}
If the cost slope is nothing but per-book markups added leg by leg, then gluing together $k$ binary markets that have nothing to do with one another should reproduce it---no multi-outcome economics required.
We test exactly that: for each real $k$-outcome event we pick $k$ mutually unrelated binary markets of similar calendar week, horizon, and volume ($9{,}919$ bundles) and price the bundle as if it were one book,
$$
\tilde y_{bt}=\textstyle\sum_{m\in b}\big(\mathrm{YES}_{mt}+\mathrm{NO}_{mt}-1\big),
\qquad
\tilde y=\alpha+\beta_{\text{bundle}}\log k+\varepsilon
$$
(clustered by bundle; also on bundle means with HC3, and on $k$-level means by inverse-variance WLS).
Two features make the placebo sharp.
First, the bundle carries per-leg markups but cannot show a book-center effect even in principle: within any one binary market the YES and NO mids sum to a dollar, so whatever mispricing one side has, its own other side cancels.
Second, unrelated legs almost never trade simultaneously---for $k\ge7$, never within the same five-minute window anywhere in the record, while real siblings co-trade constantly (trading together is itself a mark of economically linked claims)---so bundles are priced on a six-hour grid.
The bundles deliver the slope: $\beta_{\text{bundle}}=+0.062$ ($p{<}0.001$); since a bundle stacks $k$ full binary spreads (two half-spreads each), the comparable per-half-spread figure is $\beta_{\text{bundle}}/2=+0.031$, roughly half the real markup slope.
The shortfall is itself informative: a real $n$-outcome book's legs are forced to be cheap---they must average $1/n$---and cheap legs carry proportionally fat markups, which matched stand-alone binaries do not replicate.
A significant placebo slope is therefore confirmation, not contradiction: the cost channel reproduces with no multi-outcome economics at all, while the center channel---which the bundle is structurally unable to fake---appears only in genuinely linked books.

(ii)~\emph{NO side.} Buying NO on every leg of the same events (cost ${\approx}\,\$(n{-}1)$ plus premium) through identical machinery,
$$
y^{\mathrm{NO}}_{it}=\textstyle\sum_j \mathrm{noask}_{jt}-(n_i-1),
\qquad
\bar y^{\mathrm{NO}}_i=\alpha+\beta\,\log n_i+\varepsilon_i ,
$$
gives $\beta=+0.048$ (se $0.008$), decomposing as ${\sim}105\%$ markup and a center slope of $-0.003$ ($p{=}0.66$); the attribution across leg prices mirrors the YES side exactly (the slope loads on legs above $\$0.80$, the complements of the YES side's cheap legs).
A cost paid on both sides, with a center that moves on neither, is markup, not mispricing.

(iii)~\emph{Execution size.} Re-estimating the event-level regression at trade moments, restricted to those where every leg's most recent trade carried at least $q$ shares,
$$
\bar y_i\,\big|\,\min_j \mathrm{size}_{jt}\ge q
=\alpha_q+\beta_q\log n_i+\varepsilon_i ,
$$
gives $\beta_q=+0.074$ (unrestricted at trade moments), $+0.043$ ($q{=}10$), and $+0.006$, n.s.\ ($q{=}100$); the move from a leg's priced trade to its next trade averages $\$0.01$--$0.07$ for mid-priced legs.
The premium lives in small executions at thin, far-from-center quotes---consistent with makers withdrawing depth away from the center---and is not an arbitrage executable at size.
See Figure~\ref{fig:result_5} counterfactual checks for (i)-(iii).

(iv)~\emph{Stability.} The slopes survive winsorization at $1/99\%$, exclusion of Politics, of top-decile-volume events, and of $n{=}4$; they hold within Politics, Sports, Crypto, Finance, and Culture separately (within-category slopes exceed the pooled ones, the same age-composition force as Result~3); Weather is the one discordant category (negative cost slope) and 2024 listings are null.
Influence is bounded exactly: with $X$ the design matrix, $e$ the residuals, and $h_i$ the leverages of the event-level regression, the leave-one-event-out shift $\beta-\beta_{(i)}=\big[(X'X)^{-1}x_i'e_i/(1-h_i)\big]_{\log n}$ is at most $0.0008$ in absolute value across all events on either headline specification.
The long-standing $n{=}4$ anomaly is composition: that bucket is dominated by a recurring ``Fed decision in \emph{\{month\}}?''\ four-outcome template---institutionally watched books with $1$--$3$ cent premia and centers at par, the empirical image of the high-attention limit.

(v)~\emph{Functional form.} Two variants on the event-level regression,
$$
\bar y_i=\alpha+\textstyle\sum_{k=3}^{10}\delta_k\,\mathbf 1\{n_i=k\}+\varepsilon_i
\qquad\text{and}\qquad
\bar y_i=\alpha+\beta\,n_i+\varepsilon_i ,
$$
show that outcome-level dummies fit better than $\log n$, which fits better than linear $n$ (by AIC), with the dummy profile rising to $n{=}7$--$8$ and then flattening: total cost is concave, so the implied cost of the marginal outcome declines---as markup accumulation over increasingly cheap legs implies.

\begin{figure}[H]
    \centering
    \includegraphics[width=0.8\linewidth]{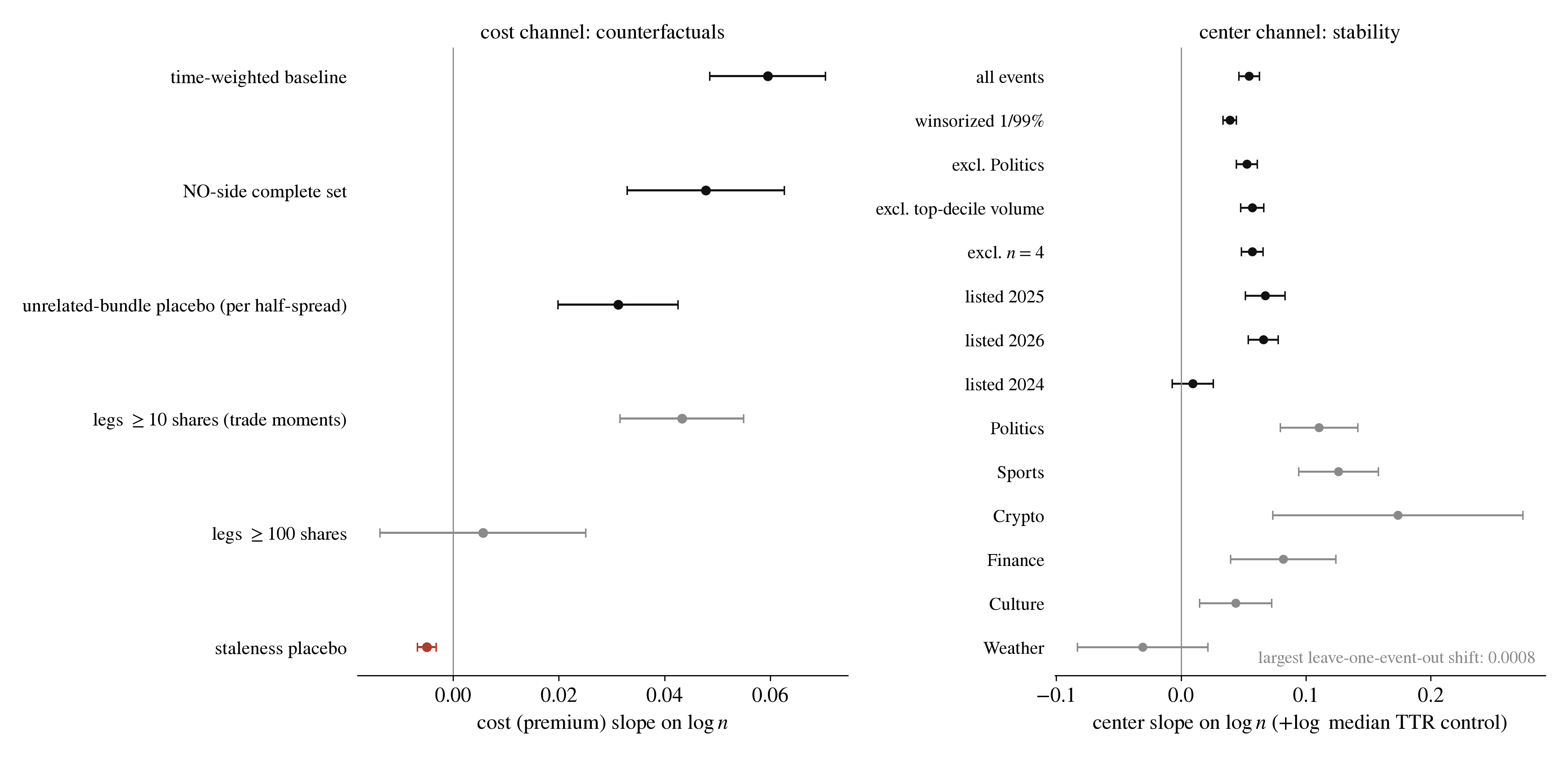}
    \caption{Counterfactual Checks (i)-(iii)}
    \label{fig:result_5}
\end{figure}

\paragraph{Reading through the model.}
The committed-markup term is the object~\eqref{eq:switch_prices_and_premium} prices: it accumulates leg by leg, appears on either side of the book, survives every weighting and every check above, is reproduced by stacking unrelated books, and vanishes only where quotes are competed to par by permanent attention (the Fed-decision template) or where size forces execution through the thin far-from-center quotes that screening implies.
The book-center term is the level of~\S\ref{subsec:3o_beliefs}: backed by no bet, it is not pinned to par moment by moment, and empirically it exceeds par precisely where the correcting trade---funding $\$(n{-}1)$ of collateral per unit payoff across $n$ thin legs and carrying it to resolution---is most expensive: many outcomes, long horizons, mid-life books.
A single cost-function potential prices the complete basket to par by construction (Proposition~\ref{prop:self_collateral}(i)); the measured premium is what its absence costs, and the measured center drift is what its absent normalization fails to pin.

\end{document}

%% file: packages.tex
\usepackage{xspace}             
\usepackage{multirow}            
\usepackage{rotating}            
\usepackage{tikz}                
\usepackage{float}               
\usepackage{algorithm}           
\usepackage{algorithmic}         
\usepackage{soul}                

%% file: macros.tex
\usepackage[colorinlistoftodos,textsize=tiny,textwidth=35pt]{todonotes}

\newcommand{\Comments}{1}

\newcommand{\mynote}[2]{\ifnum\Comments=1\textcolor{#1}{#2}\fi}
\newcommand{\mytodo}[2]{\ifnum\Comments=1%
    \todo[linecolor=#1!80!black,backgroundcolor=#1,bordercolor=#1!80!black]{#2}%
\fi}

\newcommand{\newcommenter}[3]{%
    \expandafter\newcommand\csname #1\endcsname[1]{%
        \mynote{#3!50!black}{[#2: ##1]}%
    }%
    \expandafter\newcommand\csname #1todo\endcsname[1]{%
        \mytodo{#3!20!white}{#2: ##1}%
    }%
}

\ifnum\Comments=1
    \paperwidth=\dimexpr\paperwidth+50pt\relax
    \oddsidemargin=\dimexpr\oddsidemargin+25pt\relax
    \evensidemargin=\dimexpr\evensidemargin+25pt\relax
    \marginparwidth=\dimexpr\marginparwidth+25pt\relax
\fi

\usepackage{array}
\newcolumntype{P}[1]{>{\raggedright\arraybackslash}p{#1}}
\newcolumntype{M}{>{\centering\arraybackslash\footnotesize}m{.78cm}}
\newcolumntype{S}{>{\centering\arraybackslash\tiny}m{2cm}}

\newtcbtheorem[auto counter,number within=section]{obs}%
    {Observation}{fonttitle=\bfseries\upshape, fontupper=\slshape,
        arc=0mm, colback=gensynlighterpink, colframe=gensynburgundy!60!white}%
    {theorem}

\newtcbtheorem[auto counter,number within=section]{ins}%
    {Insight}{fonttitle=\bfseries\upshape, fontupper=\slshape,
        arc=0mm, colback=gensynpink!40!white, colframe=gensynburgundy}%
    {theorem}

\usepackage{bm}

